\documentclass[format=acmsmall, review=false, acmthm=false, nonacm]{acmart}
\usepackage{acm-26}
\usepackage{booktabs} % For formal tables
\usepackage[ruled]{algorithm2e} % For algorithms
\usepackage{xfrac}

\SetAlFnt{\small}
\SetAlCapFnt{\small}
\SetAlCapNameFnt{\small}
\SetAlCapHSkip{0pt}
\IncMargin{-\parindent}

\usepackage{placeins}

\theoremstyle{acmplain}

\newtheorem{theorem}{Theorem}
\newtheorem{lemma}{Lemma}
\newtheorem{corollary}{Corollary}

\theoremstyle{acmdefinition}
\newtheorem{definition}{Definition}

\newtheorem{remark}{Remark}

\makeatletter
\newtheorem*{rep@theorem}{\rep@title}
\newcommand{\newreptheorem}[2]{%
\newenvironment{rep#1}[1]{%
 \def\rep@title{#2 \ref{##1}}%
 \begin{rep@theorem}}%
 {\end{rep@theorem}}}
\makeatother

\newreptheorem{theorem}{Theorem}
\newreptheorem{lemma}{Lemma}
\newreptheorem{proposition}{Proposition}

\usepackage{graphicx}
\usepackage{subcaption}
\usepackage{tikz}

\newcommand\ian[1]{}{}%{\textcolor{red}{Ian: #1}}{}
\newcommand\georgios[1]{}{}%{\textcolor{blue}{ #1}}{}
{}

\setcitestyle{authoryear}

\title{Paradoxes of Game Theoretic Equilibria and Price of Anarchy}
\author{GEORGIOS PILIOURAS\textsuperscript{*}, IAN GEMP, SIQI LIU, LUKE MARRIS}
\affiliation{%
  \institution{Google Deepmind}
  \country{UK}
}
\thanks{\textsuperscript{*}Corresponding author: gpil@google.com}

\begin{abstract}
For decades, static solution concepts (Nash, Correlated, and Coarse Correlated Equilibria) and the Price of Anarchy (PoA) have served as the prescriptive bedrock of algorithmic game theory. Concurrently, no-regret learning has emerged as the dominant paradigm by enabling proofs of fast convergence to such game-theoretic equilibria. We systematically demonstrate that reducing multi-agent learning to static equilibrium and black-box regret analysis introduces fundamental analytical limitations, obscuring underlying dynamic disequilibrium and game-theoretic bounds.

First, we establish that interior Nash equilibria lack $C^1$ vector field information, where agents are incapable of distinguishing between aligned and strictly opposing incentives.
Inheriting this geometry, the worst-case pure Nash equilibria that dictate the tightness of Price of Anarchy bounds 
manifest as strict saddles that are topologically unstable, and in canonical instances of congestion games, as global repellers, states of maximum potential, supported on almost everywhere strictly dominated strategies. 
 Because the worst-case framework explicitly anchors its efficiency guarantees to these dynamically unstable states, its classical quantitative bounds exhibit 
 algebraic sensitivity. 
 We prove that relaxing the syntactic constraint of non-negative coefficients to accommodate all strictly positive, affine costs renders the Price of Anarchy unbounded. 
Furthermore, we demonstrate that projecting continuous-space learning trajectories onto the discrete simplex of correlated play yields a  
permissive framework that systematically accommodates non-rationalizable behavior. Evaluating learning dynamics purely through the lens of Coarse Correlated Equilibria (CCE)---or recent continuous proximal refinements (SCE/PCE)---is structurally insufficient to preclude strictly dominated strategies. Furthermore, optimal $O(1/T)$ swap-regret minimization does not preclude macroscopic turbulence, manifesting as chaotic limit sets even in minimal normal-form games. Finally, we examine the non-atomic limit of congestion games, historically considered a highly stable environment with tight sub-linear $\Theta(p/\ln p)$ PoA bounds, where $p$ is the degree of the polynomial cost function. We prove that under standard discrete-time learning, the unique 
equilibrium  
destabilizes into Li-Yorke chaos as well as global attractors whose time-averaged inefficiency degrades exponentially as $2^p$.  Collectively, these results suggest a rigorous re-evaluation of worst-case, equilibrium-based frameworks in favor of dynamically grounded metrics.
\end{abstract}

\begin{document}

% Title page for title and abstract only.
\begin{titlepage}

\maketitle

\begin{figure}[ht!]
    \centering
    \resizebox{0.98\textwidth}{!}{
    \begin{tikzpicture}[
        >=stealth,
        text node/.style={anchor=west, align=left, font=\small\sffamily},
        label node/.style={font=\bfseries\sffamily, align=center}
    ]
    
    % 1. DRAW THE NESTED ELLIPSES (With subtle, professional shading)
    % Outermost: (\lambda, \mu)-optimal
    \draw[thick, fill=gray!5] (0, 0.4) ellipse (4.5cm and 4.2cm);
    
    % Second: CCE
    \draw[thick, fill=gray!10] (0, 0.0) ellipse (3.8cm and 3.6cm);

    % Third: SCE/PCE
    \draw[thick, fill=gray!15] (0, -0.4) ellipse (3.1cm and 3.0cm);
    
    % Fourth: CE
    \draw[thick, fill=gray!20] (0, -0.8) ellipse (2.4cm and 2.4cm);
    
    % Fifth: NE
    \draw[thick, fill=gray!25] (0, -1.4) ellipse (1.6cm and 1.8cm);
    
    % Innermost: Pure NE (Expanded for full-size text, preserving near-tangency at bottom)
    \draw[thick, fill=gray!35] (0, -2.3) ellipse (1.2cm and 0.85cm);
    
    % 2. LABELS INSIDE THE ELLIPSES (Centered between the boundaries)
    \node[label node] at (0, 4.0) {$(\lambda, \mu)$-optimal};
    \node[label node] at (0, 3.0) {CCE};
    \node[label node] at (0, 2.0) {PCE / SCE};
    \node[label node] at (0, 1.0) {CE};
    \node[label node] at (0, -0.2) {NE};
    
    % Font size constraint removed to perfectly match the rest of the diagram
    \node[label node] at (0, -2.3) {Pure/Strict\\[-0.05cm]NE};
    
    % 3. ARROWS POINTING OUT TO THE RIGHT 
    % (Calculated exact intersections for perfectly horizontal lines)
    \draw[->, thick] (3.15, 3.4) -- (5.0, 3.4);   % (\lambda, \mu)
    \draw[->, thick] (2.83, 2.4) -- (5.0, 2.4);   % CCE
    \draw[->, thick] (2.48, 1.4) -- (5.0, 1.4);   % SCE/PCE
    \draw[->, thick] (2.07, 0.4) -- (5.0, 0.4);   % CE
    \draw[->, thick] (1.43, -0.6) -- (5.0, -0.6); % NE
    
    % Recalculated intersection for the expanded Pure NE bubble
    % x^2/1.44 + (-0.3)^2/0.7225 = 1 -> x = 1.12
    \draw[->, thick] (1.12, -2.6) -- (5.0, -2.6); % Pure NE
    
    % 4. EXPLANATORY TEXT BLOCKS (Anchored uniformly to the right of the arrows)
    \node[text node] at (5.1, 3.4) {
        All agents can be simultaneously and arbitrarily\\ 
        worse off than their MinMax safety payoffs
    };
    
    \node[text node] at (5.1, 2.4) {
        Optimal no-regret learning can be supported\\
        entirely on strictly dominated strategies
    };

    \node[text node] at (5.1, 1.4) {
        Proximal/Semicoarse CE can also be supported\\
        on strictly dominated strategies
    };
    
    \node[text node] at (5.1, 0.4) {
        Chaos with no-swap-regret in symmetric games 
    };
    
    \node[text node] at (5.1, -0.6) {
        Interior Nash Equilibria are strategically insensitive;\\
        yield worst-case safety costs even in purely cooperative \\
        potential/team games
    };
    
    \node[text node] at (5.1, -2.6) {
        Worst-case pure NE are strict saddles that are topologically\\
        unstable, actively repelling physical learning trajectories\\ and even acting as states of maximum potential.\\[0.15cm] 
        %Worst-case pure NE are strict saddles that are topologically\\
        %unstable, actively repelling physical learning trajectories\\ or possessing zero-measure basins of attraction.\\[0.15cm]
        %Worst-case Pure NE are strict saddles with zero\\
        %measure attraction, and even states of maximum\\ potential that globally repel game dynamics.\\[0.15cm] 
        %Worst-case Pure NE are states of maximum potential,\\
        %globally repelling for game dynamics, and supported on\\
        %almost everywhere strictly dominated strategies.\\[0.15cm]
        The PoA in affine congestion games becomes unbounded\\
        even for pure/strict NE under data-driven cost models.
    };
    
    \end{tikzpicture}
    }
    \caption{\textbf{Paradoxes of Game-Theoretic Solution Concepts.} 
    %Static equilibria and standard solution concepts fail to meet basic dynamic and rational desiderata, even in well-studied classes of games.
    }
    \label{fig:solution_concepts}
\end{figure}
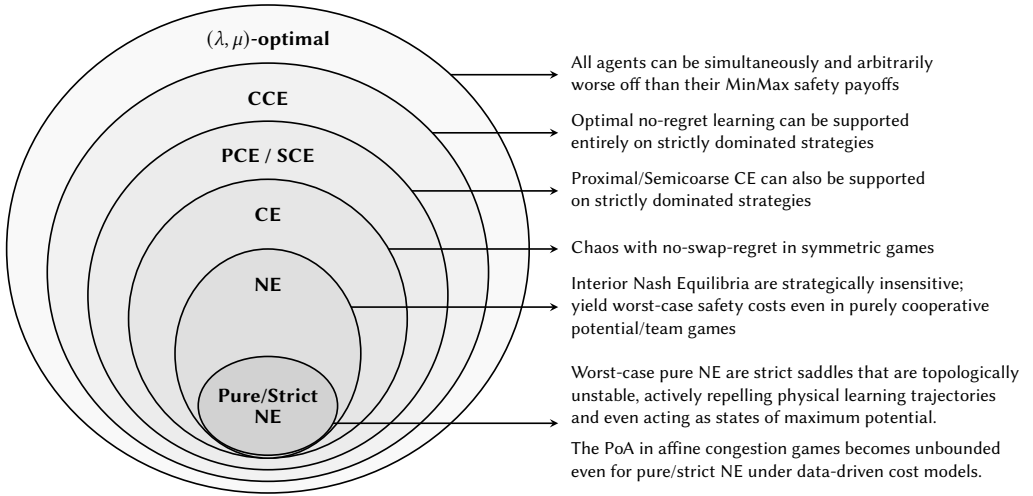

% Optionally include a table of contents
%\vspace{1cm}
%\setcounter{tocdepth}{2} % adjust to 1 if desired
%\tableofcontents

\end{titlepage}

% Paper body

\section{Introduction}
\label{sec:intro}

A fundamental challenge in the study of multi-agent systems and economics is the development of target solution concepts that serve as both mathematically rigorous models of agent behavior and reliable predictors of downstream system performance. Algorithmic Game Theory (AGT) has predominantly addressed this challenge through the efficient computation of, and approximations to, classical static equilibria such as Nash Equilibria (NE), Correlated Equilibria (CE), and Coarse Correlated Equilibria (CCE)~\citep{Nisan:2007:AGT:1296179,roughgarden2016twenty}. A significant body of research is dedicated to developing fast algorithms for the expedient computation of these concepts and their relaxations, such as $\epsilon$-NE/CE/CCE or ($\lambda, \mu$)-optimal states~\citep{Syrgkanis:2015:FCR:2969442.2969573,foster2016learning,daskalakis2021near,anagnostides2022uncoupled,piliouras2022beyond,farina2022near,gemp2023approximating,soleymani2025faster,li2024survey,peng2024fast,anagnostides2022near,anagnostides2022faster,dagan2024external,Daskalakis2024:efficient,papadimitriou2008computing,jiang2015polynomial,farina2024polynomial}. %10.1145/3736252.3742639
Their practical computational tractability is often hailed as a critical test of their real-world applicability, famously captured by Kamal Jain's maxim: ``If your laptop can't find it, then neither can the market.''

This computational focus has driven influential parallel research aimed at validating the systemic efficiency of these equilibria, most notably through the Price of Anarchy (PoA) literature~\citep{KoutsoupiasP99WorstCE,christodoulou,vetta2002nash,rough09,roughgarden2015intrinsic,christodoulou2005price,roughgarden2017price,awerbuch2005price,feldman2016price,lykouris2016learning}. Rooted in worst-case analysis, this framework argues that all equilibria---even worst-case NE/CE/CCE---remain approximately optimal in terms of social welfare. Consequently, enabling proofs of fast convergence to these approximately optimal states via no-regret learning dynamics has become a foundational paradigm in the analysis of multi-agent systems.

In this paper, we systematically demonstrate that this (worst-case) equilibrium framework is misaligned with the natural behavior of game dynamics. We trace this misalignment to a critical information loss: by evaluating multi-agent interactions purely through $C^0$ fixed-point analysis or discrete empirical distributions, standard solution concepts systematically discard the $C^1$ continuation of the underlying vector fields and expected game payoffs. When evaluated under typical, gradient-based learning dynamics, this structural and game-theoretic insensitivity  
routinely accommodates outcomes that are dynamically unstable, highly inefficient, and non-rationalizable.

\vspace{0.2cm}
\noindent \textbf{Our Contributions.} We present a series of constructive results demonstrating that efficiently computable equilibria and optimal regret guarantees can systematically accommodate dynamically unstable and non-rationalizable states. By explicitly prioritizing the topological reality of the vector field, our results reveal critical dynamical limitations of well established  analytical methodologies.

\textbf{The Topological Instability of Worst-Case Pure Nash Equilibria (Theorems~\ref{thm:zero-measure}, \ref{thm:Repeller_pure_NE}):} To establish efficiency guarantees, worst-case analysis frequently isolates pure NE. We prove that the specific worst-case pure NE dictating the tightness of robust PoA bounds in congestion games are strict saddles of their non-convex potential function. Consequently, under natural learning dynamics, they are topologically unstable and possess a region of attraction of measure zero. In fact, in canonical textbook constructions, these pure NE can act as global maxima of the potential, supported on almost everywhere strictly dominated strategies, physically repelling all interior learning trajectories. Their inclusion in performance bounds structurally alters efficiency metrics by accounting for dynamically unstable outcomes.

    \textbf{The Origin of the Instability: $C^0$ Fixed Points and Bifurcations (Theorems~\ref{thm:inverse_gt}, \ref{prop:zero_trace}):} We trace the origin of these dynamically unstable pure states to topological bifurcations induced by worst-case parameterization in combination with the dynamic instability of generic mixed NE. Evaluated strictly as fixed points, interior NE provide only $C^0$ information, lacking the local gradient data to distinguish between aligned and strictly opposing incentives. Up to strategic equivalence, an interior NE ``solves'' exponentially many ($2^{Nm}$) maximally distinct games. The worst-case pure NE inherits this strategic/dynamic insensitivity,  distorting and sometimes even outright inverting the optimization landscape.

    \textbf{The Sensitivity of Price of Anarchy (Theorems~\ref{thm:not_minmax},~\ref{lem:poisson_binomial_minmax}, \ref{thm:congestion_not_smooth}, and~\ref{thm:ergodic_good}):} Building on these topological properties, we show that classical ratio-based efficiency metrics evaluate algebraic artifacts distinct from the physical learning dynamics. First, the set of near/approximate-optimal states (in the sense of robust PoA) accommodates outcomes where all agents experience costs strictly higher than their individual min-max safety levels (due to the insensitivity of interior NE). Second, when the standard constraint of non-negative coefficients of affine/polynomial cost functions is relaxed to accommodate  data-driven cost models, the worst-case PoA becomes unbounded. We establish this divergence via minimal generic congestion games—comprising exactly two agents, two strategies, and (up to) four resources—where all affine latency functions remain strictly positive and monotonically increasing on the physical domain. Because the learning vector fields are strictly invariant to the payoff translations that drive these ratio metrics to infinity, the analytical framework evaluates algebraic representations that are strategically irrelevant to the agents themselves.

    \textbf{The Combinatorial Shadow of Correlated Play (Theorems~\ref{thm:not_rationalizable}, \ref{thm:linear_proximal}, and~\ref{thm:chaos}):} We show that projecting continuous-space learning trajectories onto the discrete simplex of empirical play (CE/CCE) causes critical information loss. Evaluating learning dynamics strictly through the lens of CCE is structurally insufficient to exclude strictly dominated strategies; we explicitly construct uncoupled no-regret algorithms that converge to (strong) CCEs supported entirely on strictly dominated outcomes. Such properties persist even under recent proximal refinements (SCE/PCE). Furthermore, we establish that optimal $O(1/T)$ swap-regret minimization (CE) does not preclude macroscopic turbulence, manifesting as chaotic limit sets even in minimal normal-form games.

    \textbf{Phase Transitions to Exponential Inefficiency and Chaos in the Large Population Limit (Theorem~\ref{thm:2p_degradation}):} Finally, we test the limits of algorithmic stability by examining non-atomic congestion games. Historically, this setting is considered a highly stable environment characterized by tight efficiency guarantees: Wardrop equilibria correspond to global minima of their convex potential, static PoA guarantees are exceptionally tight at $\Theta(p/\ln p)$, and no-regret dynamics are mathematically guaranteed to converge to approximate equilibria. However, by analyzing the discrete-time physical reality of these systems, we demonstrate that this macroscopic stability masks underlying  microscopic fragility. We prove that generic learning algorithms (including Projected Gradient Descent, and Follow-The-Regularized-Leader with $L_2$ regularization) natively induce Li-Yorke chaos. In symmetric polynomial networks of degree $p$, the system is captured by a global  attractor whose empirical time-averaged social cost diverges from the optimum by exactly $2^p$, exponentially diverging from the sub-linear static PoA bound. This highlights a critical behavioral dichotomy: the classical convergence guarantees inherently rely on global fine-tuning, normalizing the learning rate by the maximum possible network congestion $O(N^p)$ to force microscopic, stable steps. While mathematically ensuring convergence, this normalization physically requires agents to become highly unresponsive to absolute delays as the population or network degree grows. If agents instead employ standard, responsive learning rates tuned to locally experienced absolute delays, the effective step size strictly exceeds the critical stability threshold. Rather than converging, the dynamics immediately shatter into a highly inefficient, oscillating/chaotic regime.

%\end{itemize}

Collectively, these results show that the reduction of continuous state-space game dynamics to worst-case equilibrium analysis---or discrete empirical distributions---introduces significant structural slackness. While providing elegant macroscopical benchmarks, 
these classical frameworks abstract away the day-to-day physical trajectories of algorithmic agents, motivating a rigorous re-evaluation of equilibrium-based solution concepts in favor of dynamically grounded metrics.

\section{Background}
\label{sec:background}

\subsection*{Game Theory}

Let \( I =\{1, \dots, N\} \) be the finite set of players of the game \( G \). Each player \( i \in I \) has a finite \emph{strategy set} \( S_i \) and a cost function \( c_i : S_i \times S_{-i} \to [0, 1] \), where \( S_{-i} = \prod_{j \neq i} S_j \). A player \( i \in I \) selects a strategy from their set of \emph{mixed strategies} \( \Delta(S_i) \), defining a probability distribution over \( S_i \). We extend the domain of the cost function to mixed strategies via the standard multilinear extension of expected utility. The aggregate inefficiency of a state is captured by the social cost $C(s) = \sum_i c_i(s)$. 

\begin{definition}[\cite{nash}]
	A \emph{Nash equilibrium} (NE) is a vector of independent distributions \( (p_i^*)_{i \in I} \in \prod_{i \in I} \Delta(S_i) \) such that no agent can strictly decrease their expected cost via a unilateral deviation: \( \forall i \in I, \forall p_i \in  \Delta(S_i) \),
	\[
		\textstyle c_i(p_i^*, p_{-i}^*) \leq c_i(p_i, p_{-i}^*)
	\]
\end{definition}

Equilibria supported exclusively on deterministic strategies are termed \textit{pure}; otherwise, they are \textit{mixed}. To accommodate the limitations of distributed learning, algorithmic game theory frequently studies relaxations of the Nash equilibrium that allow for correlated play.

\begin{definition}[\citep{aumann1974subjectivity}] A \emph{correlated equilibrium} (CE) is a joint distribution \( \pi \) over the set of action profiles \( S = \prod_{i} S_i \) such that for any player \( i \) and any pair of distinct strategies \( s_i, s_i' \in S_i \),
	\[
		\textstyle \sum_{s_{-i} \in S_{-i}} c_i(s_i, s_{-i}) \pi(s_i, s_{-i})
		\leq \sum_{s_{-i} \in S_{-i}} c_i(s_i', s_{-i}) \pi(s_i, s_{-i})
	\]
\end{definition}

\begin{definition}[\citep{moulin1978strategically,young-book-2004}]
	A \emph{coarse correlated equilibrium} (CCE) is a joint distribution \( \pi \) over the set of action profiles \( S \) such that for any player \( i \) and any fixed unconditional deviation \( s_i \in S_i \),
	\[
		\textstyle \sum_{s \in S} c_i(s) \pi(s)
		\leq \sum_{s_{-i} \in S_{-i}} c_i(s_i, s_{-i}) \pi_i(s_{-i})
	\]
	where \( \pi_i(s_{-i}) = \sum_{s_i \in S_i} \pi(s_i, s_{-i}) \) is the marginal distribution of \( \pi \) with respect to the opponents of \( i \). This solution is closely related to Hannan consistency~\cite{hannan,hart2000simple,Cesa06} as well as (external) regret minimization (see next section). 
\end{definition}

When the inequalities in the CCE definition are strict---meaning any unilateral deviation to a fixed strategy strictly increases expected cost---the distribution constitutes a \textit{strong coarse correlated equilibrium}~\cite{anagnostides2022optimistic}.
 
\smallskip
 
The \textit{Price of Anarchy (PoA)} quantifies the degradation of system efficiency due to decentralized, self-interested behavior. It is defined as the ratio between the worst-case equilibrium cost and the centralized social optimum. Let $\mathcal{E}$ denote the set of valid equilibria in the game (e.g., pure NE, mixed NE, CE, or CCE), and let $s^* \in \text{argmin}_{s \in S} C(s)$ be the social optimum. The Price of Anarchy is formalized as:
$$
    \text{PoA}(G) = \frac{\max_{\pi \in \mathcal{E}} C(s)}{\min_{s \in S} C(s)} = \frac{\max_{\pi \in \mathcal{E}} C(s)}{C(s^*)}
$$

\begin{definition}[Smooth Game~\cite{rough09,roughgarden2015intrinsic}]
    A cost-minimization game is $(\lambda, \mu)$-smooth if for every two outcomes $s$ and $s^*$:
    \[
        \textstyle \sum_i c_i(s^*_i , s_{-i}) \leq \lambda~ C(s^*) + \mu ~C(s).
    \]
\end{definition}

The smoothness framework abstracts canonical PoA proofs, guaranteeing a \textit{robust PoA} bound of $\frac{\lambda}{1-\mu}$ that applies simultaneously to all aforementioned game-theoretic solution concepts (pure NE, mixed NE, CE, and CCE). (Randomized) states (produced, e.g., by sampling histories of no-regret dynamics) satisfying these optimality guarantees are typically called approximately optimal~\cite{Syrgkanis:2015:FCR:2969442.2969573} or near-optimal~\cite{rough09}. To resolve any ambiguity, we will call them $(\lambda,\mu)$-approximate optimal
states. Analogous definitions exist for payoff maximization games.

\smallskip

\noindent \textit{Congestion games:} A congestion game is a cost-minimization game defined by a ground set $E$ of resources, a set of $N$ players with strategy sets $S_1, \dots , S_N \subset 2^E$, and a cost function $c_e : \mathbb{Z}^+ \to \mathbb{R}$ for each resource $e \in E$~\citep{rosenthal73}. We assume the standard conditions that cost functions are non-negative and non-decreasing. Given a strategy profile $s = (s_1, \dots, s_N)$, the load $x_e(s) = |\{i : e \in S_i\}|$ is the number of agents utilizing resource $e$. The cost to player $i$ is $c_i(s) = \sum_{e \in s_i}c_e(x_e(s))$, and the social cost is $C(s) = \sum_i c_i(s) = \sum_e c_e(x_e)x_e$. A heavily studied subclass consists of games with affine cost functions $c_e(x) = a_ex + b_e$, where $a_e, b_e \geq 0$. Such games are mathematically guaranteed to be $(\sfrac{5}{3}, \sfrac{1}{3})$-smooth, enforcing a robust PoA of at most $\sfrac{5}{2}$~\cite{christodoulou2005price,rough09}. 

\smallskip

\noindent \textit{Valid utility games:} A canonical class of smooth payoff-maximization games is the set of \textit{valid utility games}~\cite{vetta2002nash}. These environments are defined by a finite ground set $E$, a non-negative submodular function $V : 2^E \to \mathbb{R}_{\ge 0}$, a strategy set $S_i \subseteq 2^E$ for each agent $i$, and a private utility function $u_i$. For any strategy profile $s = (s_1, \dots, s_N)$, let $U(s) = \bigcup_{i=1}^N s_i$ denote the aggregate set of selected elements. The social objective function is defined as $W(s) = V(U(s))$. By definition, a valid utility game must satisfy two bounding inequalities: (i) $u_i(s) \geq W(s) - V(U(s_{-i}))$, guaranteeing that an agent's private utility is bounded below by their marginal contribution to the social welfare, and (ii) $\sum_{i=1}^N u_i(s) \leq W(s)$, ensuring that the sum of private utilities does not exceed the total social value. 

Because these bounding inequalities leave the exact payoffs under-determined, explicit game-theoretic constructions utilize a \textit{basic utility system}. In a basic utility system, an agent's private utility is defined exactly as their marginal contribution: $u_i(s) = W(s) - V(U(s_{-i}))$. We note that this exact specification is fundamentally equivalent to the \textit{Wonderful Life Utility} (WLU) championed in the Multi-Agent Reinforcement Learning (MARL) and collective intelligence literature~\cite{wolpert1999introduction, wolpert2001optimal}, where it is widely deployed to resolve the credit assignment problem by aligning individual rewards with global objectives. 

When the underlying submodular function $V$ is modular (additive)---meaning $V(S) = \sum_{e \in S} v(e)$---this marginal contribution algebraically simplifies to the sum of the values of the elements that agent $i$ secures exclusively:
$$ u_i(s) = \sum_{e \in s_i \setminus U(s_{-i})} v(e) $$
This specification natively satisfies both bounding inequalities, providing a fully determined, mathematically rigorous framework for equilibrium evaluation.

\subsection*{Online Learning and Regret Minimization}

We conceptualize learning through the standard framework of online optimization. For notational brevity, we adopt the perspective of a single arbitrary agent, omitting player indices to focus on the temporal evolution of their mixed strategy $p^t \in \Delta(S)$ over rounds $t = 1, \dots, T$. 

In each round $t$:
\begin{itemize}
    \item The player selects a probability distribution $p^t$ over their strategies.
    \item The environment (or opposing agents) determines a cost vector $c^t$, specifying a cost $c^t(s)$ for every available strategy $s$.
    \item The player incurs an expected cost of $c^t(p^t)= p^t \cdot c^t$.
\end{itemize}
The distribution $p^t$ chosen at time $t$ depends strictly on the historical sequence of observations up to $t-1$. Upon committing to $p^t$, the agent receives full-information feedback, observing the entire realized cost vector $c^t$.

Given a learning algorithm $A$ generating the sequence $\{p^1, p^2,\dots\}$, the cumulative expected cost by time $T$ is $\sum^{T}_{t=1} c^t(p^t)$. Regret quantifies the performance of an algorithm relative to optimal static baselines in hindsight.

\begin{definition}[Regret and Swap Regret]
    The \emph{external regret} and \emph{swap regret} of algorithm $A$ at time $T$ are defined respectively as:
    \begin{align}
        \textstyle \text{Regret}(T) &= \textstyle \sum_{t=1}^T c^t(p^t)-  \min_{p \in \Delta(S)}\sum_{t=1}^T c^t(p) \label{eq:regret} \\
        \textstyle \text{SwapRegret}(T) &= \textstyle \sum_{t=1}^{T} c^t(p^t) - \min_{\rho: S \to S} \sum_{t=1}^{T} c^t(\rho(p^t)). \label{eq:swap_regret}
    \end{align}
\end{definition}

In Equation~\eqref{eq:swap_regret}, the maximum is over all strategy modification functions $\rho$ mapping the action space to itself, extended to act multilinearly on mixed strategies. An algorithm is Hannan consistent (or exhibits \textit{no-regret}) if its external regret grows strictly sublinearly, $o(T)$. It is a standard result that efficient algorithms exist bounding both external and swap regret at an optimal rate of $O(\sqrt{T})$~\citep{Cesa06,blum2007external,young-book-2004}. 

\subsubsection*{Continuous-Time Dynamics and Strategic Equivalence}
\label{sub:Replicator}

Shifting to continuous time, the \textit{replicator dynamic}~\cite{Taylor1978145,Schuster1983533} serves as a foundational model of evolutionary game theory and acts as the exact fluid limit of the discrete-time Multiplicative Weights Update (MWU) algorithm~\cite{Kleinberg09multiplicativeupdates,Arora05themultiplicative,akin}. It is governed by the differential equation:
\begin{equation}\label{eq:system_cost}
    \textstyle \frac{dp_i(t)}{dt} = \dot{p}_i = - p_i[ c_i(\mbox{$p$}) - \hat{c}(\mbox{$p$})], \quad \hat{c}(\mbox{$p$}) = \sum_{i=1}^{N}{p_i c_i(\mbox{$p$})}
\end{equation}

Benefiting from continuous-time smoothing, the replicator dynamic achieves strictly bounded $O(1)$ regret in general adversarial environments~\cite{mertikopoulos2018cycles,sorin2009exponential}. Standard AGT analysis routinely projects the continuous trajectories of such no-regret algorithms onto discrete empirical distributions via time-averaging, invoking convergence to CCE (or CE for swap-regret). However, this linear projection fundamentally discards the topological reality of the vector field. To rigorously analyze the exact limits of these dynamics, it is necessary to identify when distinct games generate identical dynamical trajectories.

\begin{definition}[Strategic and Dynamic Equivalence]
    Two games $G$ and $G'$ over the same strategy space are \textit{strategically equivalent} if for each agent $i$, there exists a non-strategic dummy function $d_i(s_{-i})$ such that $c_i(s)= c'_i(s)+d_i(s_{-i})$ for all $s \in S$. If we further permit a shared positive scalar multiplier $a>0$ such that $c_i(s)=a c'_i(s)+d_i(s_{-i})$, the games are \textit{dynamically equivalent}. Finally, if each agent possesses an independent multiplier $a_i>0$ such that $c_i(s)= a_i c'_i(s)+d_i(s_{-i})$, the games are \textit{game-theoretically equivalent}.  
\end{definition}

Two games are strategically equivalent if and only if the marginal payoff differences between any two unilateral deviations are identical. Consequently, an external observer cannot mathematically distinguish which game within an equivalence class is driving the observed behavior; it represents an innate gauge degree of freedom in the game's definition. Dynamically equivalent games share the exact same set of equilibria and generate identical continuous-time orbits (e.g., under replicator dynamics). Game-theoretical equivalence relaxes this further, preserving all static equilibrium sets but allowing dynamic trajectories to diverge.

%\section*{Part I: Nash Equilibria \& Price of Anarchy Revisited}

\section{The Strategic Insensitivity of Nash Equilibria}
\label{sec:OPT}

In this section, we analyze the topological and dynamical properties of the exact Nash equilibria that anchor classical efficiency bounds in algorithmic game theory. We establish that both the pure and interior equilibria utilized by worst-case analysis exhibit critical topological degeneracies, mathematically manifesting as global repellers and topological deadlocks.

\subsection{The Topological Degeneracy of Worst-Case Pure Nash}

To establish deterministic efficiency guarantees, algorithmic game theory frequently isolates pure Nash equilibria. Because generic pure (strict) NE act as topological attractors, they possess $C^1$ derivative information that strictly bounds local learning trajectories, appearing to resolve the instability of mixed states. However, we demonstrate that the specific worst-case pure NE engineered to prove the tightness of robust Price of Anarchy (PoA) bounds are not standard attractors, but rather exhibit severe topological degeneracies.

Unlike generic strict pure equilibria, which are regular and robust to payoff perturbations, worst-case pure NE are mathematically degenerate. To strictly bind the PoA inequalities, worst-case constructions require tuning the game parameters until agents are exactly indifferent between their optimal and suboptimal strategies. Algebraically, this exact indifference forces the marginal cost difference to vanish at the boundary of the simplex. By smoothly extending the potential function beyond the closed constraints, this tightness condition mandates that the worst-case boundary state acts as an interior critical point ($\nabla \Phi = \mathbf{0}$) within the unconstrained space.

Under the exact algebraic constraints of worst-case parameterization, a topological bifurcation occurs where an interior saddle point is pushed precisely onto the boundary. The resulting degenerate state mathematically satisfies the static inequalities of a pure NE, but its local geometry is severely distorted by the algebraic tightness constraints. These worst-case pure NE are not local minima, but rather strict saddles of the extended potential. Because their stable directions can point entirely outside the physical simplex, they can physically manifest as states of maximum potential.

\begin{figure}[t]
    \centering
    \includegraphics[width=0.55\textwidth]{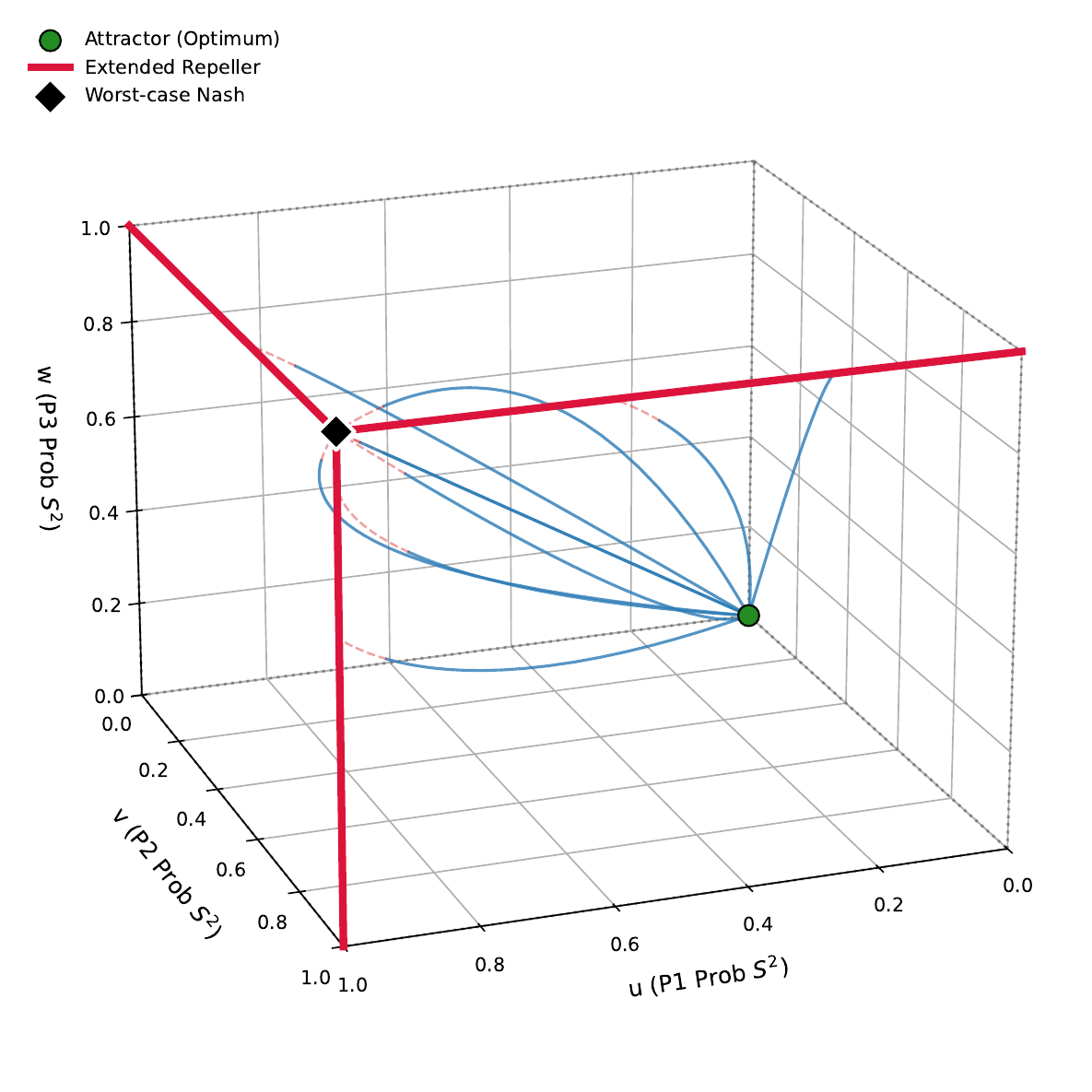}
    \caption{\textbf{The Topological Degeneracy of Worst-Case PoA (I).} Replicator dynamics in the linear congestion game producing the tight lower bound of $5/2$ for pure PoA~\cite{christodoulou2005price},\cite{roughgarden2015intrinsic} (Theorem 5.6, Example 5.7). The worst-case pure NE corresponds to a state of \textbf{maximum potential}---the most dynamically unstable state in the system. 
    The strategies supporting the worst case pure Nash are \textbf{almost everywhere strictly dominated} in the measure-theoretic sense. Under any gradient-like learning dynamic (where $\Phi$ acts as a strict Lyapunov function outside the equilibrium set), the worst-case equilibrium $y$ lies strictly in the \textbf{global repeller} of the interior state space.
    The game is highly degenerate, emerging from a bifurcation where a strict saddle is pushed into an attractor. At the annihilation point, both the attractor and saddle vanish, leaving a global repeller representing all single-agent deviations. (See Appendix~\ref{app:global_max} for a more detailed discussion).} 
    \label{fig:repeller_PoA}
\end{figure}

To analyze the local geometry of the worst-case equilibria, we must evaluate the landscape of the potential function. Because the worst-case state lies on the boundary of the probability simplex, we formalize the notion of a strict saddle via smooth analytic extension into an open domain.

\begin{definition}[Strict Saddle via Analytic Extension]
Let $\Phi$ be a twice-continuously differentiable function defined on a closed manifold with boundaries $\mathcal{X}$. Suppose $\Phi$ admits a smooth extension to an open neighborhood $U \subset \mathbb{R}^N$ encompassing $\mathcal{X}$. A point $q^* \in \mathcal{X}$ is a \textbf{strict saddle} if, evaluated within the extended interior space $U$, its gradient vanishes ($\nabla \Phi(q^*) = \mathbf{0}$) and its Hessian matrix, $H = \nabla^2 \Phi(q^*)$, admits at least one strictly negative eigenvalue ($\lambda_{\min}(H) < 0$).
\end{definition}

\begin{theorem}[Topological Instability of Worst-Case Nash Equilibria]\label{thm:zero-measure}
Let $\Gamma$ be a canonical exact potential congestion game utilized to establish tight robust Price of Anarchy (PoA) bounds (e.g., the affine double-cycle of~\cite{christodoulou2005price,roughgarden2015intrinsic}, or the multi-hop network of~\cite{awerbuch2005price,Nisan:2007:AGT:1296179}) where its worst-case pure Nash equilibrium $y$ perfectly binds the smoothness inequalities. Under these exact algebraic constraints, the following properties hold:
\begin{enumerate}
 \item The worst-case pure Nash equilibrium $y$ manifests as a \textbf{strict saddle} of the exact continuous potential function $\Phi$.
    \item Under any smooth learning dynamic defined by an interior-regular Riemannian metric (strictly positive-definite at the boundary), the Lebesgue measure of the basin of attraction for $y$ within the physical probability simplex is exactly zero.
    \item Under any continuous learning dynamic where the exact potential $\Phi$ acts as a strict Lyapunov function outside the equilibrium set, the equilibrium $y$ is \textbf{topologically unstable}, actively repelling physically valid learning trajectories away from it.
\end{enumerate}    
\end{theorem}

\begin{remark}[Physical Geometry and Manifold Intersections]
While the universal zero-trace law guarantees that the worst-case pure Nash equilibrium $q = \mathbf{0}$ is mathematically a strict saddle in the unconstrained space $\mathbb{R}^N$, its physical manifestation depends entirely on how its stable manifold $W^s$ intersects the valid probability simplex $[0,1]^N$. 

In the 3-player affine bounding game of~\cite{christodoulou2005price}, the quadratic form of the Hessian dictates that all directions of positive curvature point strictly outside the non-negative orthant. Consequently, the stable manifold $W^s$ does not intersect the interior of the physical probability space. Within the confines of the game, every physically valid trajectory flows strictly downhill, causing the saddle to act as a pure topological repeller (and a constrained global maximum of the potential function).  See Figure~\ref{fig:repeller_PoA}, Theorem~\ref{thm:Repeller_pure_NE} and Appendix~\ref{app:global_max} for more details.

\begin{theorem}\label{thm:Repeller_pure_NE}
Let $G$ be the congestion game constructed via the canonical smoothness arguments to yield the worst-case Price of Anarchy for affine cost functions in ~\cite{christodoulou2005price,roughgarden2015intrinsic} for $N=3$ agents. Let %$y^*$ be the socially optimal pure Nash equilibrium and
 $y$ be the worst-case pure Nash equilibrium, then:
\begin{enumerate}
    \item The strategies supporting $y$ are almost everywhere strictly dominated in the measure-theoretic sense.
    \item The state $y$, alongside the manifold of states formed by unilateral deviations from it, constitutes the global maximum of the exact potential function $\Phi$.
    \item Under any gradient-like learning dynamic (where $\Phi$ acts as a strict Lyapunov function outside the equilibrium set), the worst-case equilibrium $y$ lies strictly in the global repeller of the interior state space.
\end{enumerate}
\end{theorem}

In the 4-player bounding network of~\cite{awerbuch2005price}, the worst-case pure NE acts as a geometric bottleneck suspended between the socially optimal Nash equilibrium (the global minimum of the potential/attractor) and the highly congested ``anti-equilibria'' (the global maxima of the potential/repellers) (Figure~\ref{fig:phase_portrait_AAE}).
\end{remark}

\begin{figure}[t]
    \centering
    \includegraphics[width=1.0\textwidth]{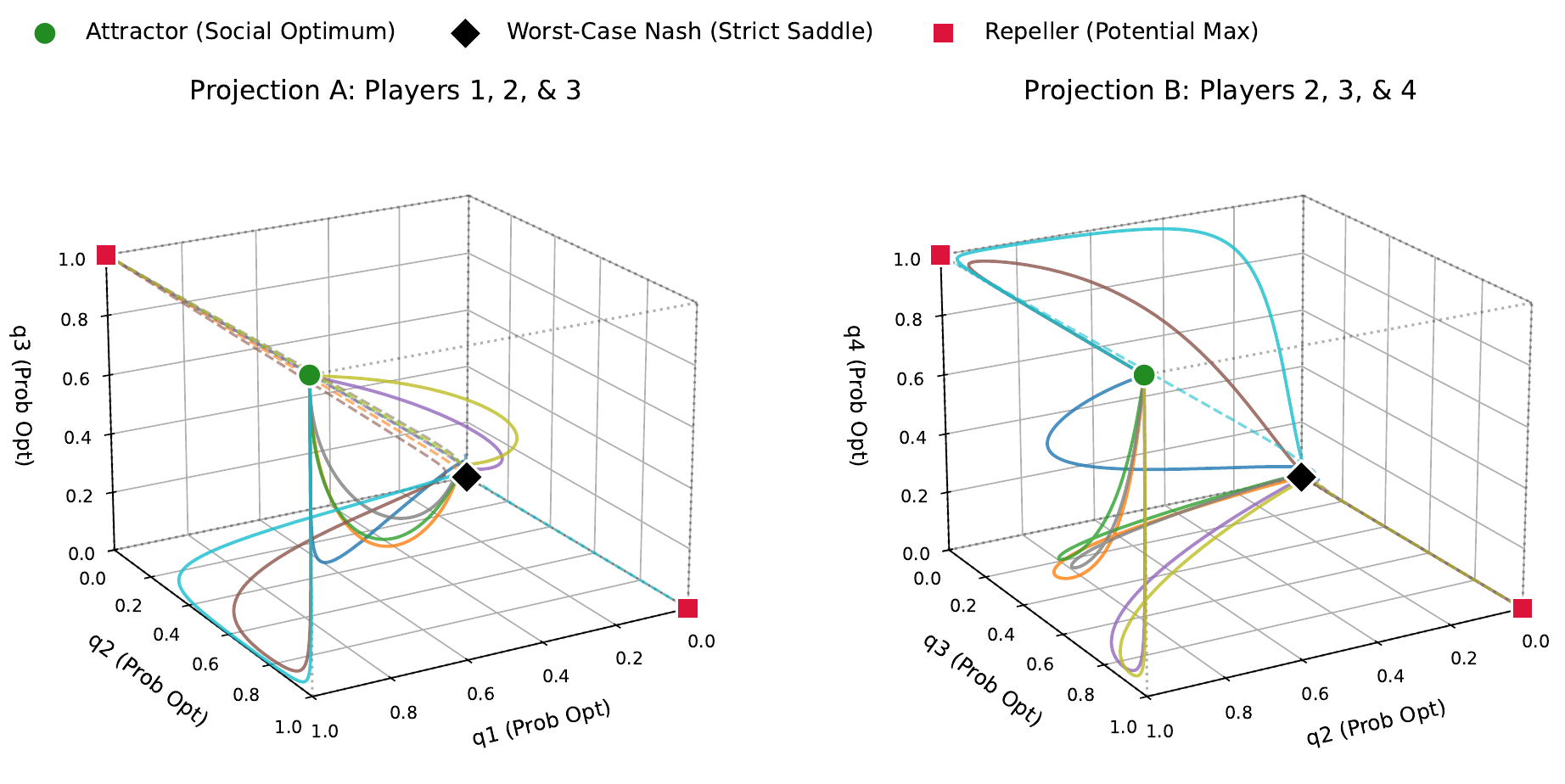}
    \caption{\textbf{The Topological Degeneracy of Worst-Case PoA (II)}.  Replicator dynamics for the congestion game establishing the $5/2$ pure PoA lower bound~\cite{awerbuch2005price}, \cite{Nisan:2007:AGT:1296179} (Chapter 18). Shown are two 3D projections of the 4D state space capturing the player trajectories. The worst-case pure Nash equilibrium (black diamond) manifests as a geometric "mountain pass" (\textbf{strict saddle}), perfectly suspended between the socially optimum Nash equilibrium (green circle, global minimum of the potential) and the heavily congested ``anti-equilibria'' (red squares, global maxima). Trajectories nearby the bad Nash equilibrium exhibit \textbf{slowing down}---crawling asymptotically toward the worst-case Nash before the local geometry of the potential landscape inevitably captures and repels the flow toward the socially optimum Nash equilibrium. Consequently, the worst-case bound is \textbf{topologically unstable}, actively repelling
    learning trajectories. 
    When trajectories are simulated back in time they converge to the two repellers.}
    \label{fig:phase_portrait_AAE}
\end{figure}

Theorems~\ref{thm:zero-measure}, \ref{thm:Repeller_pure_NE} establish that the Nash equilibria defining the tightness of classical PoA bounds are dynamically inaccessible, and sometimes even strictly anti-optimal (acting as states of maximum potential). 
By anchoring efficiency bounds to these dynamically unstable states, worst-case analysis yields performance metrics that are inherently more pessimistic than the empirical trajectories observed under natural learning.

Even if these worst-case games are slightly perturbed to break the degeneracy---undoing the bifurcation so the suboptimal pure state becomes a weakly strict NE ($\epsilon$-attractor)---the fundamental issue remains. The region of attraction for the suboptimal NE becomes a set of negligible measure. Thus, the inefficiencies captured by worst-case bounds are structurally decoupled from the typical behavior of natural learning dynamics.

\subsection{The $C^0$ Topology of Interior Nash: Deadlocks and Inverse Selection}

The topological degeneracy of the worst-case pure equilibria described above is fundamentally inherited from the structural properties of fully mixed (interior) Nash equilibria. Evaluated strictly as fixed points of learning dynamics, interior NE provide only $C^0$ information about the system. They carry no auxiliary $C^1$ information regarding the continuation of the vector field infinitesimally close to the fixed point. 

Because they rely strictly on $C^0$ evaluation, interior NE lack the local gradient data necessary to distinguish between aligned and strictly opposing incentives. If a game $G$ admits an interior NE, that exact same state is simultaneously an equilibrium of the exact negation $-G$ (where agents act as cost maximizers). Furthermore, up to strategic equivalence, the interior NE remains an equilibrium in any adversarial team game derived by inverting the incentives of an arbitrary subset of agents. 

Consequently, in cooperative potential and congestion games, interior NE function as topological deadlocks. Because the $C^0$ fixed-point condition is entirely symmetric, the interior NE effectively functions as an ``each-vs.-all'' adversarial environment: each agent's local gradient acts as if all other agents are actively coordinating to penalize them. Rather than optimizing the shared potential, agents trapped at an interior NE are locked into minimal, non-cooperative safety-level payoffs. 

We formalize the magnitude of this $C^0$ insensitivity by considering the pre-image of the equilibrium mapping. If an equilibrium concept holds prescriptive power, observing the state $p$ should allow an analyst to largely narrow down the strategic structure of the game being played. Instead, we observe a combinatorial explosion of strategically distinct games.

\begin{theorem}[The strategic insensitivity of interior Nash equilibria]
\label{thm:inverse_gt}
    Let $G$ be a generic game with $N$ agents and $m \ge 3$ strategies each, and let $p$ be an interior Nash equilibrium of $G$. The equivalence class $INV(p)$ of games for which $p$ is a Nash equilibrium includes at least $2^{Nm}$ distinct games, no two of which are game-theoretically equivalent to each other.
\end{theorem}

Interior NE do not solve these games to approximate optimality; they are simply trapped, unable to navigate the exponentially many intersecting games they theoretically inhabit. While prior literature has established the instability of these states under specific dynamics (e.g. replicator, multiplicative weights update)~\cite{Kleinberg09multiplicativeupdates,panageas2019multiplicative},  we can distill this instability into a universal geometric law. By directly evaluating the exact potential function, we formally establish that these states are generically strict saddles of the underlying optimization landscape, independent of the choice of learning algorithm.

\begin{theorem}[The Universal Zero-Trace Law]
\label{prop:zero_trace}
In any  potential game where players randomize independently, let $q^*$ be a Nash equilibrium. Let $\mathcal{S}^*$ denote the restricted state space defined by the support of $q^*$. If $q^*$ is partially or fully mixed, and generic within its support (i.e., the restricted Hessian of the potential evaluated at $q^*$ is non-zero), then $q^*$ is a strict saddle point of the potential function restricted to $\mathcal{S}^*$. 
\end{theorem}

Because a state that acts as a strict saddle within a restricted subspace naturally repels trajectories along that manifold, it is dynamically unstable in the full state space. However, because worst-case analysis evaluates all feasible algebraic solutions equally, bounding performance across these dynamically repelling strict saddles introduces significant structural slackness into the worst-case bounds.

\begin{remark}
\textbf{(Structural Robustness of Topological Deadlocks).} The strategic insensitivity and topological deadlocks induced by interior Nash equilibria are not fragile artifacts. By Harsanyi's theorem on generic games~\cite{harsanyi1973games}, generic interior equilibria are regular, possessing a non-singular Jacobian of the static payoff mapping. Consequently, by the Implicit Function Theorem, these states are structurally stable. If the payoff entries are subjected to a generic perturbation, the dynamically insensitive saddle point does not vanish; it merely undergoes a small continuous shift. As a result, the topological traps that actively repel gradient-based learning dynamics persist across a positive Lebesgue measure of the parameter space, rendering local smoothed analysis fundamentally insufficient to resolve the systemic inefficiency.
\end{remark}

\FloatBarrier

\section{The Sensitivity of Price of Anarchy}
\label{sec:metrics_collapse}

In the previous sections, we have established that Nash equilibrium, the predominant game theoretic solution concept, provides certificates of stability to states which are demonstrably and universally unstable from an algorithmic/optimization perspective, the \textit{strict saddles} and even \textit{global maxima} of non-convex minimization landscapes. This structural divergence is robust; it is present not only in the canonical textbook examples that form the foundational bedrock of the Price of Anarchy literature, but also generically in the case of (partially) mixed NE of potential/congestion games. Finally, interior Nash equilibria are strategically insensitive outcomes in the sense that an outside observer with access to their payoff matrices cannot infer if they are cost-minimizing or utility-maximizing agents.

The Price of Anarchy (PoA) literature establishes that in numerous classes of games (with congestion games being arguably the predominant such example), Nash equilibria and wide generalizations thereof (e.g., CE/CCE) are approximately optimal in terms of social cost. However, the social cost and potential in congestion/potential games are strongly correlated (e.g., within a factor of 2 of each other for all linear congestion games). This introduces a fundamental structural tension: PoA frameworks guarantee approximately optimal social cost, yet they mathematically anchor these guarantees to states that dynamically act as repelling strict saddles where agents are entirely indifferent to the optimization objective.

We resolve this tension by revisiting PoA guarantees from three different perspectives. First, we introduce a new baseline, where we compare social cost not multiplicatively against the optimal states but additively against min-max safety baselines. We show that approximate optimality in the sense of robust PoA is strictly weaker than min-max safety.  Second, we establish the sensitivity of PoA in congestion games, showing that the PoA in affine congestion games is unbounded. Finally, we show that even when performing average case performance analysis, PoA ratio-based metrics remain unbounded.  

\subsection{Divergence from Min-Max Safety Baselines}
\label{sec:minmax_regret}

Since its inception, the Price of Anarchy has evaluated multi-agent systems by benchmarking equilibrium outcomes against a centralized social optimum. While this approach provides a metric for social efficiency, it inherently bypasses a more fundamental game-theoretic criterion: baseline individual rationality. 
A fundamental game-theoretic baseline
of any solution concept is whether agents can guarantee the performance of their strictly adversarial safety strategies.

We formalize this via min-max costs: the minimum cost an agent can guarantee themselves even if all opponents play a coordinated, adversarial strategy designed to maximize that agent's cost. Any rational agent can independently secure this safety level; thus, any solution concept that systematically yields costs higher than this min-max threshold admits non-rationalizable behavior.

\begin{definition}[MinMax-Centered Performance]
    Given any game $G$ and any state $s$, the guaranteed safety gap for agent $i$ is:
    \begin{align*}
        u_i^\text{MinMaxCenter}(s) &= \min \max c_i - c_i(s) \\
        &= \min_{p_i \in \Delta(S_i)} \max_{\pi_{-i} \in \Delta(S_{-i})} \sum_{s_{-i}\in S_{-i}} \sum_{s_i} c_i(s_i, s_{-i}) p_i(s_i) \pi_{-i}(s_{-i}) - c_i(s)
    \end{align*}
\end{definition}

\begin{definition}[MinMax-Regret Performance]
    \begin{align*}
        u_i^\text{MinMaxRegret}(s) = c_i(\text{minmax}^*_i ,s_{-i})- c_i(s)
    \end{align*}
    where $\text{minmax}^*_i$ is the max-entropy min-max strategy of agent $i$ against a meta-agent controlling the opponents via correlated randomness. 
\end{definition}
  
A score of zero serves as the absolute baseline of economic rationality. Any rationalizable outcome must satisfy $u_i^\text{MinMaxCenter}(s) \geq u_i^\text{MinMaxRegret}(s) \geq 0$. However, the broad class of $(\lambda, \mu)$-approximate optimal states structurally fails to secure this safety baseline, even in the most canonical, highly symmetric settings.

\begin{theorem}[Divergence from Baseline Rationality in Balls and Bins Games]
\label{thm:not_minmax}
    Consider the class of canonical balls and bins games (symmetric load balancing games with $N$ agents, $M \ge 2$ identical resources, and monomial cost functions $c(x) = x^d$ for $d \ge 1$). The set of $(\lambda, \mu)$-approximate optimal states dictated by the robust Price of Anarchy framework strictly includes states where every agent simultaneously incurs a cost arbitrarily worse---scaling as $-\Theta(N^d)$---than their absolute (correlated) min-max safety guarantee.
\end{theorem}

The proof (detailed in Appendix~\ref{app:PoA_metrics}) is constructive. By explicitly computing the optimal social cost, the exact min-max safety threshold, and the theoretical $(\lambda, \mu)$ upper bound, we demonstrate that the standard robust PoA inequalities are structurally broad enough to include states of perfect miscoordination within the 'approximately optimal' set.

To understand the source as well as commonality of this inefficiency, we must analyze the exact probabilistic nature of the load induced by the adversaries. In these canonical congestion games, the agents' incentives are perfectly aligned to minimize collisions. Yet, the fully mixed interior Nash equilibrium extracts exactly zero cooperative surplus. As formalized below, the interior equilibrium perfectly mimics the optimal adversarial attack.

\begin{theorem}[Independent Min-Max Equivalence]
\label{lem:poisson_binomial_minmax}
    In symmetric balls and bins games with $N$ agents, $M$ bins, and convex monomial costs $c(x) = x^d$ ($d \ge 1$), the expected cost evaluated at the fully mixed interior Nash equilibrium is equal to the independent min-max safety cost. (Proof deferred to Appendix~\ref{app:PoA_metrics}).
\end{theorem}

Modulo the constraint of independent play, the fully mixed interior Nash equilibrium mathematically collapses to the adversarial min-max safety strategy. (For linear congestion games where $d=1$, linearity of expectation forces this independent safety cost to perfectly collapse to the absolute correlated min-max cost). 

This mathematical equivalence between a cooperative equilibrium and a zero-sum adversarial safety strategy is the exact quantitative manifestation of the $C^0$ topological insensitivity established in Section~\ref{sec:OPT}. Because the worst-case framework must analytically accommodate these zero-surplus topological deadlocks, it systematically generates efficiency bounds that fall arbitrarily below the threshold of economic rationality.

\subsection{The Algebraic Sensitivity of Price of Anarchy Bounds}
\label{sec:poa_fragility}

This structural sensitivity extends from individual safety baselines to the global efficiency guarantees of the system. The standard robust Price of Anarchy (PoA) analysis for affine congestion games rests on a specific algebraic axiom: all resource cost functions must take the form $c_e(x) = a_e x + b_e$, where both $a_e, b_e \ge 0$. While the condition $a_e \ge 0$ correctly models the negative externalities of congestion, enforcing a non-negative intercept $b_e \ge 0$ introduces an additional algebraic constraint that is not implied by positivity or monotonicity on the physically realizable congestion domain.

Traditionally, the $b_e \ge 0$ constraint is defended via continuous-domain intuition: a hypothetical ``empty'' resource ($x=0$) cannot incur negative latency. However, atomic congestion games inherently operate over discrete, strictly positive integer loads ($x \ge 1$). In real-world systems, it is trivial to encounter strictly positive, monotonically increasing latency functions whose exact affine representation possesses a negative intercept. Consider a shared machine requiring a fixed ``cool-down'' or setup period between successive uses. If the first task requires $1$ second, and every subsequent task requires $2$ seconds ($1$ for setup, $1$ for execution), the exact affine representation of the physical latency is $c(x) = 2x - 1$. The costs are strictly positive for all physically realizable loads ($x \ge 1$), yet the intercept is definitively negative.

The exclusion of such valid physical models highlights a structural limitation within worst-case equilibrium analysis. It is a fundamental property of the Price of Anarchy that it is highly sensitive to affine transformations and cost translations. Adding a fixed positive constant to all costs artificially inflates the optimal social cost (the denominator), thereby artificially deflating the PoA ratio. Because the PoA metric is inherently \emph{not} invariant to these strategic equivalences, the framework operates under the strict premise that a game's cost functions represent exact, un-normalized empirical measurements. If the applicability of the PoA ratio relies on costs reflecting exact physical measurements, excluding cost functions with negative intercepts structurally limits the application of these bounds, particularly when those functions accurately represent ground-truth measurements of everywhere-positive, strictly increasing physical costs.

This formulation dependence becomes particularly relevant in modern, data-driven applications. In practice, latency functions are rarely provided as exact theoretical primitives; they are typically inferred via regression over historical data. When modeling cost dynamics using standard parametric approaches, statistically optimal best-fit curves frequently yield negative coefficients to accurately capture the curvature over the active physical domain ($x \ge 1$). Imposing strict algebraic non-negativity constraints on these learned parameters restricts the model's expressivity, potentially reducing its fidelity to the empirical observations. This restriction becomes progressively stronger for higher-order polynomial latency models: while affine costs constrain only the intercept, a degree-$d$ polynomial latency function requires $d+1$ simultaneous coefficient-sign constraints. Consequently, the algebraic subclass covered by classical polynomial PoA theory (e.g., \cite{christodoulou,aland2011exact,Nisan:2007:AGT:1296179}) becomes progressively narrower relative to the class of positive, monotone latency functions as the expressive power of the model increases.\footnote{In scientific computing, fitting data using the standard monomial basis is generally avoided due to numerical instability. Regression libraries project data onto orthogonal polynomial bases (e.g., Chebyshev polynomials $T_2(x) = 2x^2 - 1$)~\cite{trefethen2019approximation}. Expanding any stably learned cost function back into the monomial form can thus naturally yield negative coefficients.}

\begin{figure}[t]
    \centering
    \includegraphics[width=.40\textwidth]{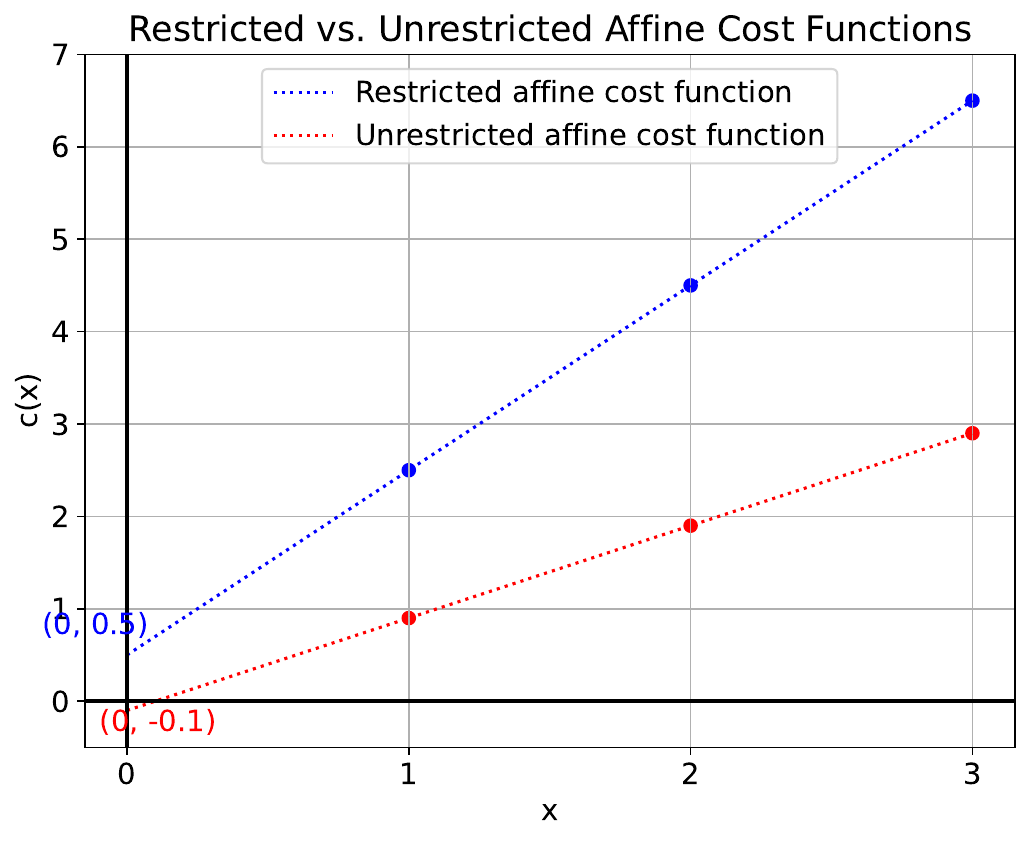}
    \caption{\textbf{Algebraic Constraints vs. Physical Validity.} A comparison between a restricted and an unrestricted affine function in robust Price of Anarchy analysis. While both cost functions map strictly positive integer loads to positive costs, the critical parameter dictating their inclusion in PoA bounds relies entirely on whether the extrapolated, physically inaccessible continuous line crosses the $y$-axis above or below zero.}
    \label{fig:PoA_allowable}
\end{figure}

\begin{figure}[t]
    \centering
    \begin{tabular}{c|c|c|}
        \multicolumn{1}{c}{} & \multicolumn{1}{c}{A} & \multicolumn{1}{c}{B} \\
        \cline{2-3}
        A & $(w+1, w+1)$ & $\mathbf{(w, w)}$ \\
        \cline{2-3}
        B & $\mathbf{(1, 1)}$ & $(w+1, w+1)$ \\
        \cline{2-3}
        \end{tabular}
        \hspace{2cm} % Adjust horizontal space between tables
        \begin{tabular}{c|c|c|}
        \multicolumn{1}{c}{} & \multicolumn{1}{c}{A} & \multicolumn{1}{c}{B} \\
        \cline{2-3}
        A & $(\epsilon, \epsilon)$ & $\mathbf{(\epsilon^2, \epsilon^2)}$ \\
        \cline{2-3}
        B & $\mathbf{(\epsilon^2, \epsilon^2)}$  & $(\epsilon, \epsilon)$  \\
        \cline{2-3}
    \end{tabular}
    \caption{\textbf{Unbounded Price of Anarchy in Affine Congestion Games.} Two parametric families of affine congestion games where the PoA is strictly unbounded. Pure Nash equilibrium states are bolded. \textbf{Left:} For games with $w>1$, the pure PoA is unbounded when accommodating affine cost functions of arbitrarily large slopes. \textbf{Right:} The mixed PoA is unbounded even for families of affine congestion games with an arbitrarily small positive Lipschitz constant (slope).}
    \label{fig:payoff_matrices}
\end{figure}

The subtlety of this algebraic boundary highlights a deep structural dichotomy in the classical literature. While the non-negativity requirement ($b_e \geq 0$) is structurally necessary for classical Price of Anarchy bounds, it is strategically irrelevant to the agents and invariant under the induced learning dynamics. Consequently, the boundedness of the efficiency metric is governed by an algebraic property of the chosen cost representation rather than by the strategic or dynamical behavior of the game itself.

%Consequently, because classical worst-case PoA bounds diverge when valid, everywhere-positive latency functions possess negative intercepts, applying these static bounds to data-driven, learned cost models requires careful structural re-evaluation.

\begin{theorem}[The sensitivity of Price of Anarchy]\label{thm:congestion_not_smooth}
The Price of Anarchy of affine congestion games is unbounded even when restricting to affine costs with arbitrarily small positive slopes. Moreover, the Pure Price of Anarchy of general affine congestion games is unbounded. This holds even for minimal games comprising exactly two agents, two strategies, and four congested elements, where all congested elements possess positive, non-decreasing affine latency functions on the active domain.
\end{theorem}

These constructions (detailed in Appendix~\ref{app:PoA_metrics}) demonstrate that bounded Price of Anarchy guarantees rely fundamentally on the algebraic normalization of the latency functions rather than solely on their strategic or physical properties. The affine intercept is merely the simplest manifestation of this phenomenon. Classical polynomial PoA analyses similarly require every monomial coefficient to be nonnegative, introducing an increasing number of independent algebraic sign constraints as the polynomial degree grows. Consequently, if bounded PoA already fails under the minimal violation of a single affine coefficient constraint, extending these guarantees to empirically learned polynomial latency models becomes increasingly difficult without significant structural modification.

\subsection{Average Price of Anarchy is Similarly Unbounded}
 
To salvage efficiency guarantees in the presence of unstable worst-case equilibria, a natural analytical recourse is the Average Price of Anarchy (APoA)~\cite{panageas2016average,sakos2024beating}. Rather than evaluating the single worst-case equilibrium, APoA evaluates the expected social cost, weighting each Nash equilibrium proportional to the Lebesgue measure of its basin of attraction under natural learning dynamics. We show that while this methodology recovers optimal performance for certain game families, it inherits the critical limitations of classical PoA.
 
\begin{theorem}[The Dynamical Disconnect of Average-Case Ratios]\label{thm:ergodic_good}
    For the first parametric family of affine congestion games yielding unbounded PoA (Figure~\ref{fig:payoff_matrices}), under a uniform Lebesgue measure prior, the expected limit social cost is at most $2$ times the optimal social cost, with instances reaching an average-case inefficiency of at least $1+\frac{1}{27}(18-2\sqrt{3}\pi) \approx 1.2636$. For the second family, given any prior absolutely continuous with respect to the Lebesgue measure, the average-case performance strictly equals the optimal social cost. 

    Conversely, by linearly translating the costs of the first family of games, we can construct affine congestion games with two agents and two strategies where, under a uniform Lebesgue prior, the average-case system performance becomes arbitrarily worse than the optimal social welfare.  %Thus, the Average Price of Anarchy in affine congestion games is unbounded.
\end{theorem}

There are two interpretations of this result. The first is an econometric perspective, which assumes that the observable costs are accurate and exogenously given. Under this lens, Theorem~\ref{thm:ergodic_good} establishes that even when deploying the sophisticated tools of dynamical systems and measure theory to account for the exact regions of attraction of pure/strict Nash equilibria, there is no universal ratio-based efficiency guarantee; performance evaluation naturally reduces to case-by-case empiricism.

The second, more fundamental interpretation is that the Price of Anarchy (and its average-case variants) are 
%is neither dynamically nor game-theoretically well-posed. 
highly sensitive to strategically invariant cost translations.
Both static equilibria and the continuous orbits of game dynamics remain invariant to strategically equivalent transformations of the game. However, the social cost---and its corresponding performance ratio---does not. Consequently, multi-agent behavior is evaluated against  
an efficiency metric that is sensitive to global cost translations, even though these translations are strategically irrelevant to the learning agents and the resulting equilibria.

The performance ratio is  sensitive to the specific affine representation chosen from the strategic equivalence class.  
The convention of enforcing non-negative coefficients, while sufficient to bound this ratio, relies on a specific algebraic representation that does not generalize to all valid physical translations. 
Even when finite Average PoA bounds can be derived under these restricted representations, they possess no inherent game-theoretic or dynamical meaning. The metric evaluates the system's empirical limit sets based on sensitive coordinate choices, fundamentally decoupling theoretical efficiency guarantees from the topological and measure-theoretic reality of the game.

%\section*{Part II: The Combinatorial Shadow of Learning}

\section{The Combinatorial Shadow of Learning: Fast Convergence to Strictly Dominated Strategies}
\label{sec:CCE}

A central focus of modern algorithmic game theory is establishing fast, optimal regret-minimization guarantees for learning dynamics in games~\citep{Syrgkanis:2015:FCR:2969442.2969573,foster2016learning,daskalakis2021near,piliouras2022beyond,farina2022near,farina2024regret,anagnostides2022faster}. A common assumption  is that optimal algorithmic convergence rates can serve as a proxy for rational, efficient multi-agent behavior. We demonstrate that this reduction introduces critical information loss.

It is an established property of generic, unstructured normal-form games that Coarse Correlated Equilibria (CCE) can support strictly dominated strategies~\citep{viossat2013no}. Recognizing this vulnerability, recent literature has introduced specialized algorithmic interventions explicitly engineered to filter out unrationalizable states during the learning process~\citep{wang2023learning,wu2021multi}. However, the mathematical necessity of these bespoke interventions highlights
a fundamental limitation in the standard evaluative framework. While specific algorithms can be forcefully engineered to avoid dominated strategies, the fundamental target concept itself---no-regret learning evaluated through the lens of CCE---provides no inherent mathematical guarantee of economic rationality.

When the continuous trajectories of no-regret learning are projected onto the discrete simplex of empirical play via time-averaging, the resulting combinatorial shadow loses critical topological information. We demonstrate that the speed of convergence, even when constrained by strict informational boundaries, is mathematically insufficient to preclude severe game-theoretic non-rationalizable states, even within the highly structured conditions of smooth games.

\begin{theorem}[Fast No-Regret Convergence to Strictly Dominated Strategies in Smooth Games]
\label{thm:not_rationalizable}
    There exist smooth games (e.g., linear congestion games and valid utility games) where adversarially selected, strongly uncoupled no-regret algorithms converge at an optimal rate of $O(1/T)$ to a strong coarse correlated equilibrium, whilst simultaneously, on every period, all agents play strictly dominated strategies. 
\end{theorem}

%\textbf{Remark on Topological Robustness:} 
\begin{remark}
Theorem~\ref{thm:not_rationalizable} is not an isolated edge-case; it is topologically robust. By leveraging the slackness inherent to the defining inequalities of strong CCE, there exist open families of games exhibiting these properties. This fragility persists both in the native topology of structured games (e.g., perturbing the latency functions of a linear congestion game) and when embedding them into the space of general games, as is inevitable when inferring cost functions from empirical data.
\end{remark}

While the existence of non-rationalizable CCEs has been established in generic games~\citep{viossat2013no}, Theorem~\ref{thm:not_rationalizable} establishes that this problem fundamentally penetrates the highly structured environments of smooth games---the exact environments where CCEs are championed for their robust social welfare guarantees~\citep{anagnostides2022last,roughgarden2015intrinsic,Syrgkanis:2015:FCR:2969442.2969573,foster2016learning,daskalakis2021near}. 

The proof (detailed in Appendix~\ref{app:CCE}) relies on two structural facts that completely decouple optimal learning from rationality. First, we establish the geometric existence of \textit{extreme CCEs} (vertices of the feasible CCE polytope) in smooth games where all probability mass is allocated exclusively to strictly dominated strategies. Second, we demonstrate that optimal regret rates combined with strongly uncoupled learning~\citep{daskalakis2011near,hart2003uncoupled} cannot structurally filter out these problematic states. Because strongly uncoupled algorithms evaluate states purely via local, historical payoff sequences without structural awareness of the opponent's game matrix, they can be induced to perfectly mimic periodic orbits inducing non-rationalizable CCEs. 

Consequently, standard grey-box algorithmic constraints---such as optimal $O(1/T)$ convergence rates, uncoupledness, and bounded memory---are structurally insufficient to guarantee rational game-theoretic behavior. As these solution concepts are increasingly extended to evaluate complex, high-dimensional domains such as general-sum Markov and extensive-form games~\citep{anagnostides2022faster,daskalakis2023complexity,farina2020coarse,jin2021v,mao2023provably,song2021can}, recognizing the systemic danger of this strategic insensitivity becomes even more important. In conclusion, the issue is merely not that coarse equilibria allow strictly dominated outcomes, but that the standard reduction from learning dynamics to such equilibria erases essential structure: it ignores dynamical accessibility and can certify these outcomes as approximately optimal. As a result, the framework loses the ability to distinguish between stable and unstable, or rational and non-rationalizable behavior, revealing a fundamental misalignment between dynamics, equilibrium abstractions, and efficiency guarantees.

\section{Proximal Refinements and Strictly Dominated Strategies}
\label{sec:SCE-PCE}

Recognizing the permissive nature of Coarse Correlated Equilibria, recent literature has attempted to bridge the gap toward  Correlated Equilibria (CE) by introducing intermediate solution concepts rooted in continuous optimization. Chief among these are Semicoarse/Proximal Correlated Equilibria (SCE/PCE) and their generalization, Bregman Proximal Correlated Equilibria (BPCE)~\cite{cai2025proximal,ahunbay2025semicoarse}. These refinements demand that agents exhibit no regret against counterfactual strategy modifications generated by Bregman proximal operators. 

From the perspective of continuous online optimization, these frameworks are mathematically elegant: standard, uncoupled learning algorithms (such as Projected Gradient Descent or Mirror Descent) naturally minimize proximal regret at optimal $O(1/\sqrt{T})$ rates. However, when these continuous algorithmic bounds are evaluated as discrete game-theoretic equilibria---which rely on the empirical distribution of sampled pure actions---a fundamental structural mismatch emerges.

Let $\phi : \Delta(S_i) \to \mathbb{R}$ be a strongly convex distance-generating function, and $D_\phi(y\|x)$ be its associated Bregman divergence. For a weakly convex test function $f$, the Bregman proximal operator is defined as:
$$ \text{prox}_f^\phi(x) = \text{argmin}_{y \in \Delta(S_i)} \left\{ f(y) + D_\phi(y\|x) \right\} $$
A distribution $\sigma \in \Delta(S)$ constitutes a Bregman Proximal Correlated Equilibrium if, for all valid test functions $f$, no agent benefits from applying $\text{prox}_f^\phi$ to their sampled pure actions.

We establish that this continuous-to-discrete translation structurally permits convergence to non-rationalizable states. Because the proximal operator relies on continuous geometric projections, it introduces a ``linear shielding effect'' that allows the probability mass of optimal strategies to effectively shield strictly dominated strategies from yielding positive deviation gains.

\begin{theorem}[SCE/PCE Support Strictly Dominated Strategies]
\label{thm:linear_proximal}
For Euclidean distance-generating functions, the set of Proximal Correlated Equilibria (which coincides exactly with Semicoarse Correlated Equilibria) can assign strictly positive probability mass to strictly dominated strategies. 
\end{theorem}

The formal construction and primal-dual verification of Theorem~\ref{thm:linear_proximal} are provided in Appendix~\ref{app:SCE}. By explicitly constructing a smooth game and a valid SCCE distribution, we demonstrate that the inequalities defining proximal regret are algebraically satisfied even when agents place significant probability mass on unconditionally inferior actions. 

This issue immediately carries over to the broader generalizations of the framework. Because the standard Euclidean distance ($D_\phi(y\|x) = \frac{1}{2}\|y-x\|^2$) is the canonical Bregman divergence, standard Proximal Correlated Equilibria are a strict subset of Bregman Proximal Correlated Equilibria. Consequently, the generalized Bregman PCE framework inherits these problematic equilibria. 

\begin{corollary}
The generalized class of Bregman Proximal Correlated Equilibria does not exclude strictly dominated strategies. 
\end{corollary}

\FloatBarrier

\section{Beyond the Combinatorial Shadow: The Topological Reality of No Swap Regret}
\label{sec:CE}

\begin{figure}[tb]
    \centering
    \begin{subfigure}[t]{0.48\textwidth}
        \centering
        \includegraphics[width=0.85\linewidth]{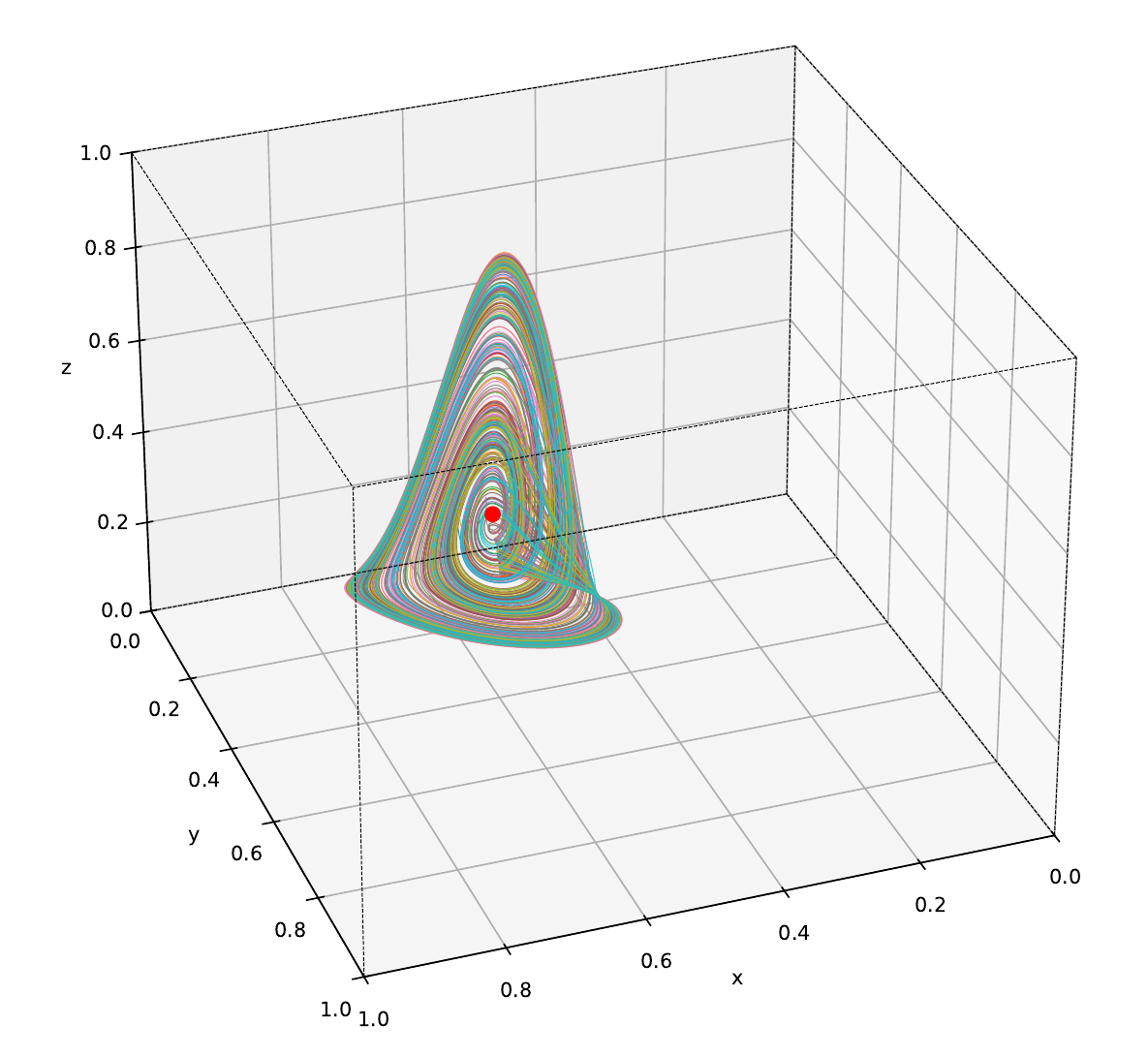}
        \caption{Chaotic limit cycle of replicator dynamics in a normal-form polymatrix game. The red dot is the Nash equilibrium. Colors denote distinct trajectories.}
        \label{fig:chaos}
    \end{subfigure}
    \hfill
    \begin{subfigure}[t]{0.48\textwidth}
        \centering
        \includegraphics[width=\linewidth]{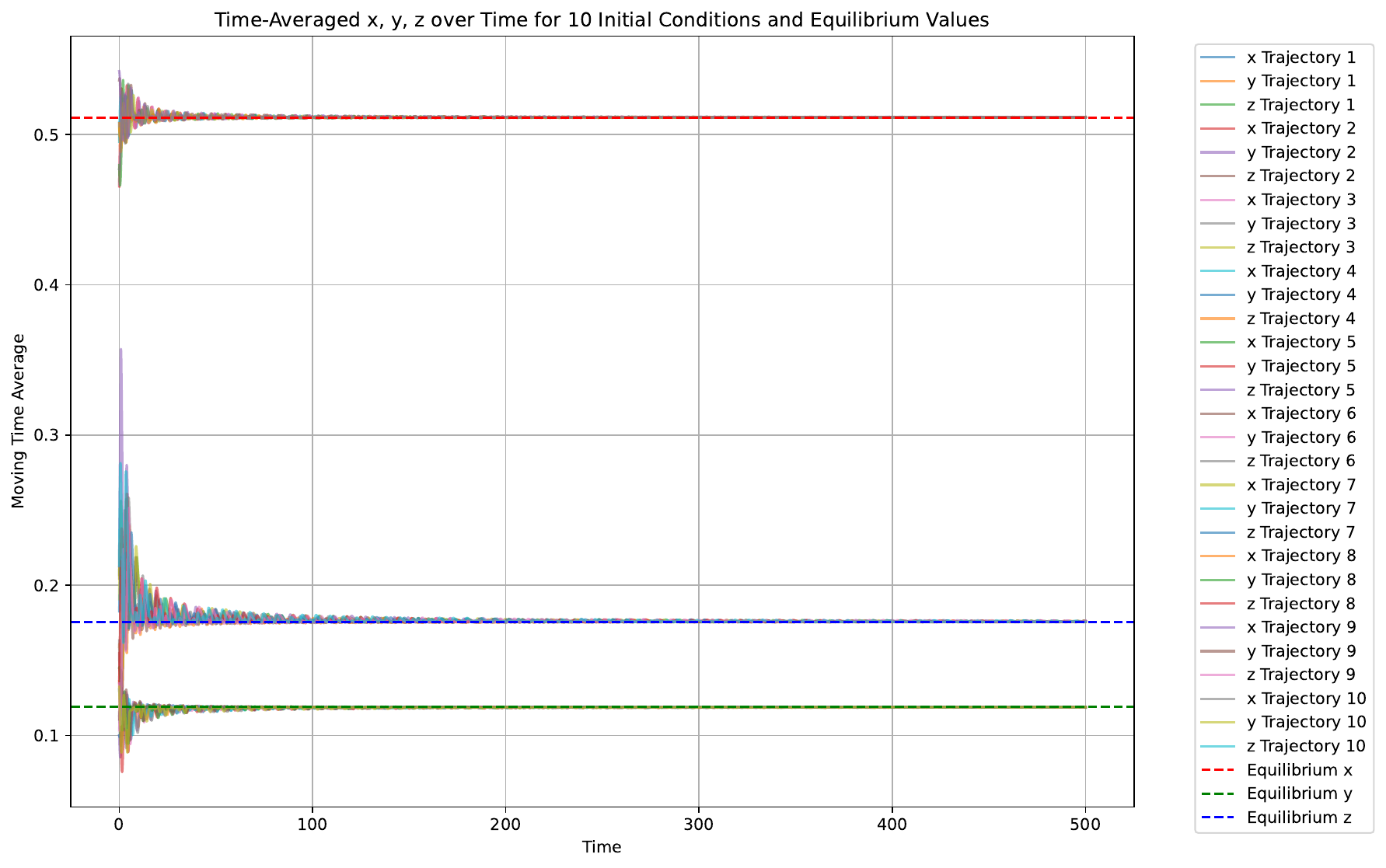}
        \caption{For all interior symmetric initial conditions the time-average of the strategies of all agents converges precisely to the interior Nash equilibrium.}
        \label{fig:NE}
    \end{subfigure}

    \vspace{0.5cm}

    \begin{subfigure}[t]{0.48\textwidth}
        \centering
        \includegraphics[width=\linewidth]{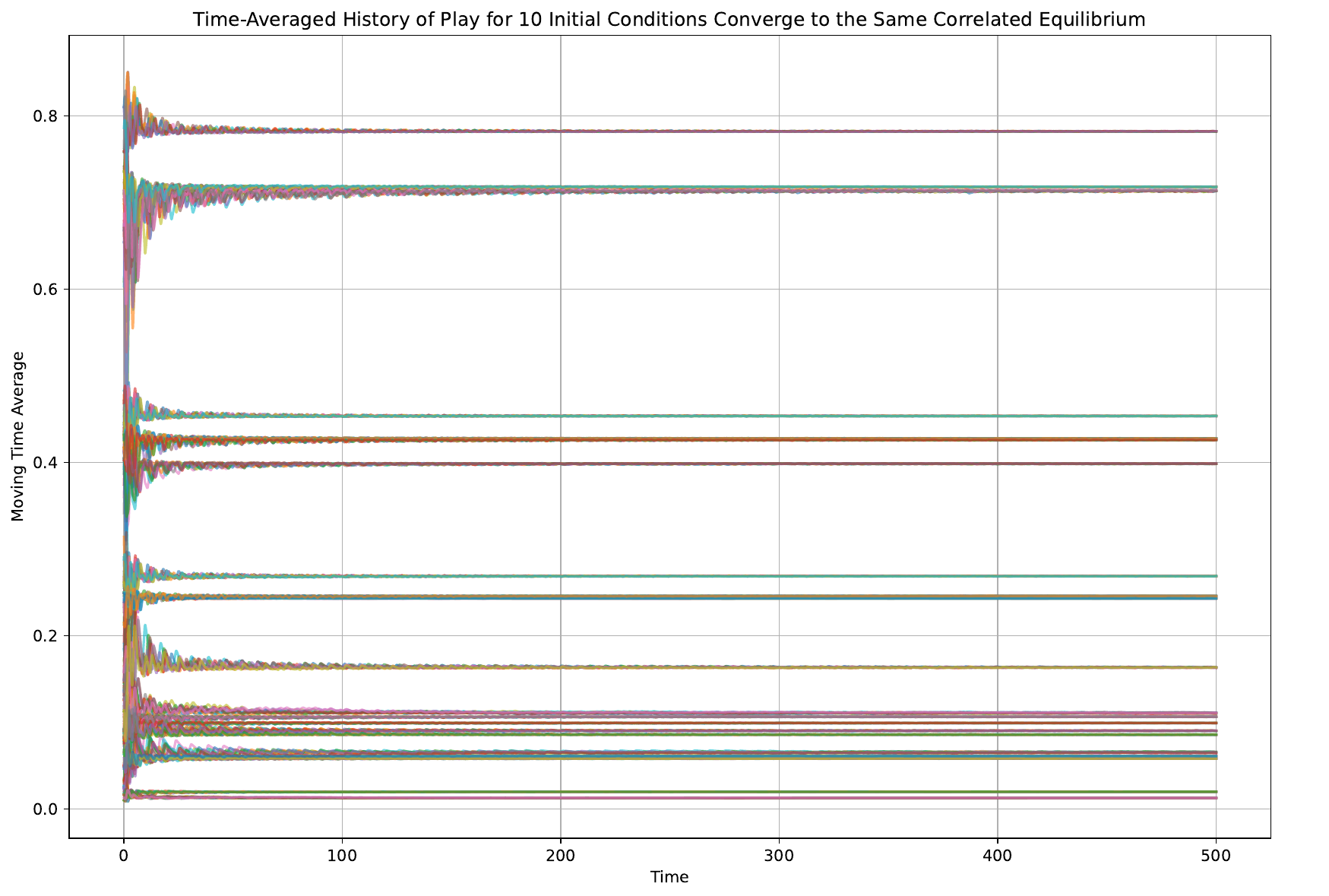}
        \caption{The time-average of the empirical pairwise probabilities converges pointwise to their value at a correlated equilibrium (CE).}
        \label{fig:CE}
    \end{subfigure}
    \hfill
    \begin{subfigure}[t]{0.48\textwidth}
        \centering
        \includegraphics[width=\linewidth]{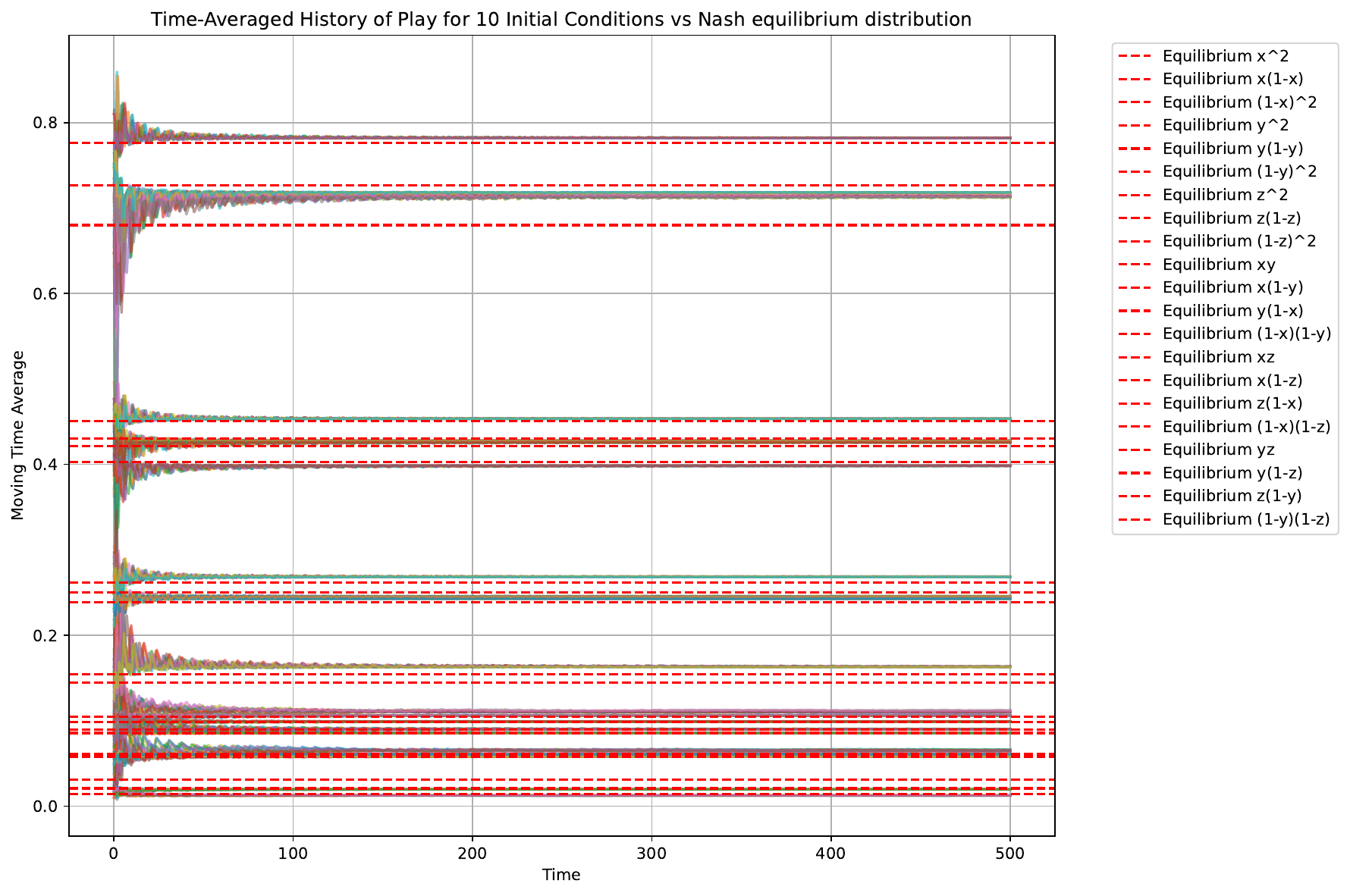}
        \caption{The time-averages computed in (c) strictly do not converge to their independent Nash equilibrium values (encoded by the dashed red lines).}
        \label{fig:notNE}
    \end{subfigure}

    \caption{\textbf{The Combinatorial Shadow of Chaos.} Replicator dynamics in normal-form games simultaneously exhibit macroscopic chaos, optimal $O(1/T)$ convergence to CE, and optimal state-average convergence to NE.}
    \label{fig:2x2grid}
\end{figure}

\begin{figure}[t]
    \centering
    \includegraphics[width=0.65\textwidth]{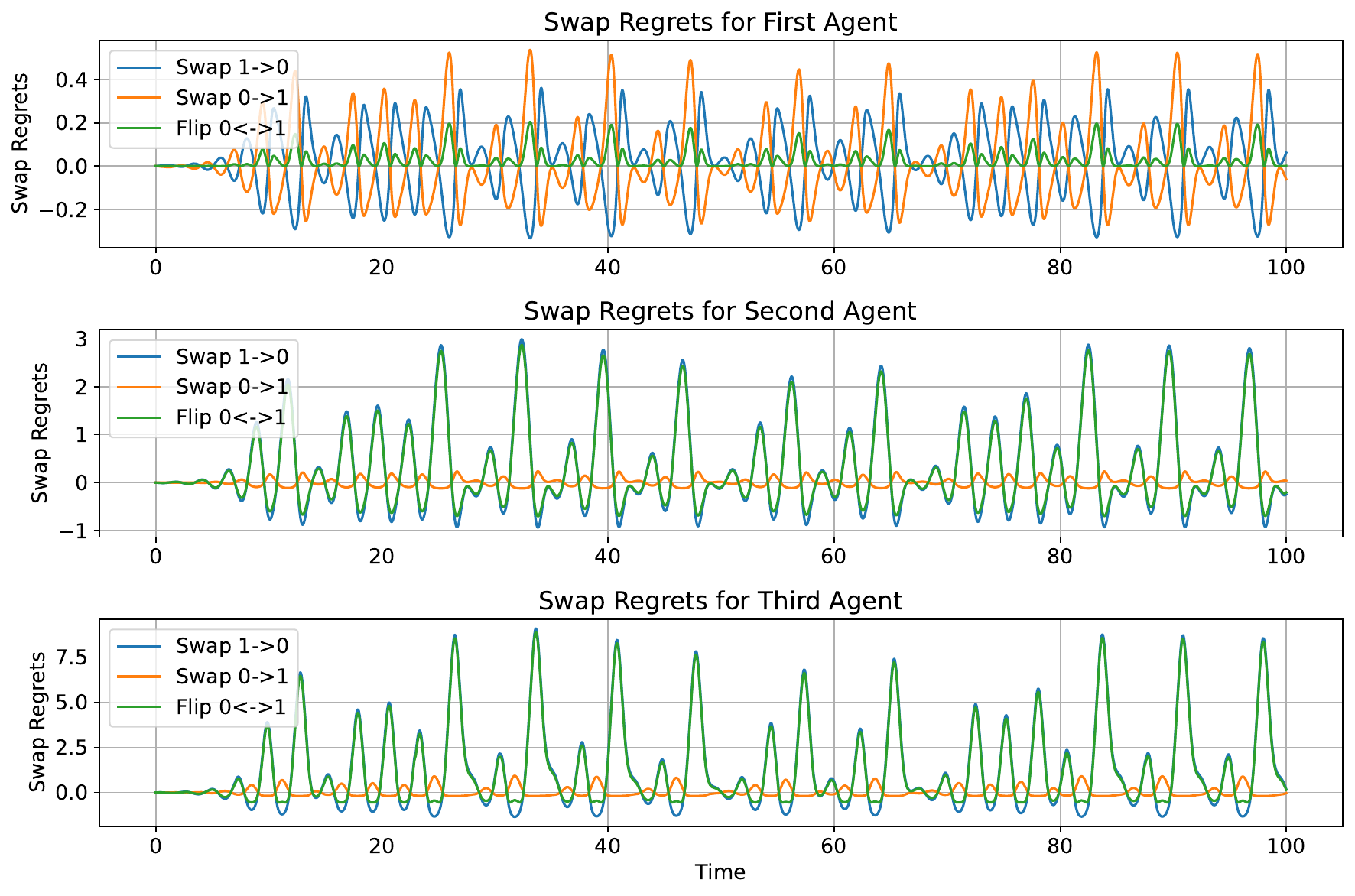}
    \caption{\textbf{Bounded Swap Regret amidst Turbulence.} The instantaneous swap regret of all agents remains strictly bounded, ensuring the time-average of the replicator dynamics converges provably at an optimal $O(1/T)$ rate to a chaotic correlated equilibrium.}
    \label{fig:swap_regret}
\end{figure}

Having established that Coarse Correlated Equilibria and their continuous proximal refinements (SCE/PCE) do not exclude strictly dominated strategies, a natural theoretical recourse is to demand swap-regret minimization. Constrained by the PPAD-hardness of computing exact Nash equilibria~\cite{daskalakis2006complexity,chen2009settling}, swap regret and its sequential extensions have emerged as standard benchmarks for rationalizable learning. 

Recent work has expanded the application of swap regret far beyond classical normal-form Correlated Equilibria (CE).\footnote{Beyond its traditional algorithmic justification as the limit of swap-regret minimization, CE has recently been established as the unique solution concept satisfying a fundamental set of normative logical axioms, including consistency under payoff uncertainty and invariance to strategy cloning~\cite{brandl2026axioms}.} It provides the theoretical foundation for online calibrated forecasting~\cite{foster1997calibrated,fishelson2025full}, and has been extended to imperfect-information, Bayesian, and general convex games via Extensive-Form, Linear, and Profile swap regret~\cite{celli2020no,farina2022near,daskalakis2024efficient,arunachaleswaran2025swap}. Concurrently, novel algorithmic reductions have circumvented the classical polynomial dependence on the action space; however, these theoretical improvements mathematically necessitate an exponential dependence on the precision parameter, requiring $\exp(\Omega(1/\epsilon))$ rounds to converge~\cite{dagan2024external,peng2024fast}. 

Beyond assigning  zero probability mass to strictly dominated strategies, swap-regret minimization imposes strict utility-theoretic constraints. It provably guarantees non-manipulability against strategic adversaries~\cite{deng2019strategizing,mansour2022strategizing}, ensures the Pareto-optimality of the learning algorithm's asymptotic menu~\cite{arunachaleswaran2025swap}, and forces frequent-iterate convergence to Nash equilibria in symmetric zero-sum settings~\cite{leme2024convergence}. Given these theoretical guarantees, it is a standard premise that optimal $O(1/T)$ convergence to the CE polytope imposes sufficient mathematical structure to ensure predictable, stable, and rationalizable trajectories.

We demonstrate that this assumption does not hold dynamically. Evaluating learning algorithms strictly through the discrete combinatorial shadow of swap regret projects away the continuous-time dimensions necessary to capture complex, non-equilibrating dynamics. By directly analyzing the multi-agent learning trajectories, we prove that even in minimal normal-form games, optimal convergence to CE can perfectly coincide with complex chaotic limit sets.

\begin{theorem}[Fast No-Swap-Regret Convergence of Replicator Dynamics to Chaotic Correlated Equilibria]\label{thm:chaos}
    There exists a one-parameter family of normal-form games with payoff symmetries such that continuous-time replicator dynamics, under symmetric interior initial conditions, exhibits chaos while simultaneously:
    \begin{enumerate}
        \item[(i)] Minimizing swap regret at an optimal rate of $O(1/T)$ (fast convergence to CE).
        \item[(ii)] Ensuring the time-average strategies of all agents converge precisely to an interior Nash equilibrium.
    \end{enumerate}
    The games in question are polymatrix games with $6$ agents, each possessing exactly two strategies.
\end{theorem}

The implications of this result are clarified by contrasting it with canonical $3\times3$ zero-sum games, such as Rock-Paper-Scissors. In standard $3$-strategy games, (even in symmetric RPS games under symmetric initializations) replicator dynamics minimize external regret but provably fail to minimize swap regret, thereby converging only to the CCE polytope~\cite{leme2024convergence}. However, because our construction utilizes exactly two strategies per agent, testing against fixed deviations is sufficient to check against all swap deviations. In this specific environment, replicator dynamics inherently functions as an optimal, no-swap-regret algorithm. 
When combined with provable time-average convergence of agents' marginal distributions to NE, it satisfies arguably the most stringent game-theoretic desiderata (outside pointwise convergence to NE) while simultaneously sustaining chaotic trajectories.

The emergence of this chaotic limit set relies on a bipartite mirroring technique that translates non-concave evolutionary game dynamics into multilinear normal-form optimization. By coupling a 3-agent chaotic polymatrix replicator~\cite{peixe2022persistent} with an exact clone population, the resulting 6-agent system preserves the chaotic vector field along a symmetric invariant manifold.

Crucially, standard topological constraints dictate that generic normal-form games cannot possess attracting limit sets strictly within the interior of the state space under continuous-time no-regret learning~\cite{omidshafiei2019alpha,flokas2020no}. Consequently, this chaotic regime mathematically manifests not as an open attractor, but as an invariant limit set strictly bound to the symmetric initialization manifold. 

From a measure-theoretic perspective, under a prior absolutely continuous with respect to the Lebesgue measure (e.g., a uniform prior over the entire state space), this exact chaotic limit set constitutes a measure-zero event. However, this invariant set acts as a topological mechanism for disequilibrium. While the strange attractor exhibits strictly positive internal Lyapunov exponents (verified via perturbation analysis in Appendix E.1) generating chaos within the manifold, the invariant manifold itself is also transversely unstable. As demonstrated numerically in Appendix E.2, introducing microscopic asymmetric noise to break the initial synchronization does not cause the system to relax to a static Nash equilibrium. Instead, the transverse instability exponentially amplifies the asymmetric perturbation, permanently breaking the mirror symmetry. Depending on the exact microscopic noise vector, arbitrarily close initial conditions are repelled toward entirely distinct, lower-dimensional limit cycles residing on different boundary faces of the state space. 

This extreme sensitivity to initial conditions ensures that even when the measure-zero chaos is transient, the ultimate macroscopic support of the agents' strategies remains fundamentally unpredictable. Whether the system remains trapped in the chaotic manifold under symmetric priors, or is ejected into distinct boundary cycles under generic perturbations, the game-theoretic conclusion is identical. The projection of these continuous orbits onto the discrete CE simplex systematically obscures the underlying dynamics. A proof of fast convergence to Correlated Equilibria guarantees neither behavioral predictability nor dynamic equilibration.

\section{Phase Transitions in Non-Atomic Games: The Dynamical Reality of Large Populations}
\label{sec:nonatomic_dynamics}

In the preceding sections, we observed how the combinatorial complexity of finite normal-form games inherently complicates the convergence of independent learning dynamics. One might naturally look to the continuous domain of non-atomic congestion games as a refuge of algorithmic stability. Indeed, this setting typically represents arguably the ideal of game-theoretic and algorithmic optimality. In non-atomic routing, Wardrop equilibria correspond to the global minima of a convex potential function, yielding effectively unique macroscopic flows \cite{beckmann1956studies,Nisan:2007:AGT:1296179}. The celebrated Price of Anarchy (PoA) results, based on tight two-link Pigou networks, provide worst-case inefficiency bounds of $\Theta(p/\ln p)$ for polynomial latency functions of degree $p$ \cite{roughgarden2002bad,roughgarden2004bounding}. These guarantees are significantly tighter than their atomic counterparts, a favorable property that atomic games are known to formally recover in the limit as the population size increases \cite{Feldman:2016:PAL:2897518.2897580}. Furthermore, unlike in atomic games where no-regret algorithms can be consistent with strictly dominated strategies, classic results establish that no-regret learning in non-atomic routing guarantees the day-to-day empirical flow converges to an $\epsilon$-equilibrium flow~\cite{blum2006routing}. 

Together, these results form a powerful macroscopic framework. However, a rigorous examination of the physical timescale of these guarantees reveals a fundamental tension in the day-to-day dynamical reality of the system. To achieve a physically meaningful relative error (e.g., bounding regret by a small fraction $\delta$ of the optimal equilibrium cost), no-regret algorithms must aggressively throttle their learning rate to account for the maximum possible cost in the network. For a network of $m$ paths with degree $p$ latencies, achieving this relative bound requires an effective step-size scaling of $\mathcal{O}(1/m^p)$, which e.g., given the canonical decreasing step-sizes of $\theta(1/\sqrt{t})$, in turn forces the convergence horizon to scale exponentially as $\Omega(m^{2p})$. For standard responsive learning rates, the system operates in a vast meta-stable regime long before asymptotic bounds become binding.

Recent important contributions have illuminated this meta-stable regime, demonstrating the emergence of Li-Yorke chaos and qualitative inefficiencies in non-atomic congestion games under Multiplicative Weights Update (MWU) and Follow-the-Regularized-Leader (FTRL) with steep regularizers \cite{CFMP2019,bielawski2021follow,bielawski2025heterogeneity}. Building upon this work, we provide a generalized and exact quantification of this inefficiency for polynomial latencies, contrasting the dynamical reality directly against the static PoA guarantees.

\begin{theorem}[Exponential Degradation of the Global Attractor]
\label{thm:2p_degradation}
Consider a symmetric non-atomic congestion game with two parallel paths and identical monomial latency functions $c(x) = x^p$ for $p \ge 1$. Under MWU with a fixed effective step-size parameter $a > 0$\footnote{The parameter $a$ formally couples the MWU learning rate $\epsilon$, the population size $N$, and the polynomial degree $p$. It is defined as $a = N^p \ln\left(\frac{1}{1-\epsilon}\right) \approx \epsilon N^p$. This relationship explicitly demonstrates that to maintain a stable effective step-size $a$, the raw learning rate $\epsilon$ must be aggressively throttled inversely to the maximum network cost $\mathcal{O}(N^p)$.}:
\begin{enumerate}
    \item The unique Wardrop equilibrium $x^* = 1/2$ is locally asymptotically stable if $a < \frac{2^{p+1}}{p}$.
    \item If $a > \frac{2^{p+1}}{p}$, the equilibrium strictly destabilizes and the interior system is captured by a global, period-2 attracting orbit $\{\sigma_a, 1-\sigma_a\}$. Accounting for the population size $N$, the empirical time-averaged social cost of this chaotic orbit is exactly $SC_{emp}(a) = N^{p+1}[\sigma_a^{p+1} + (1-\sigma_a)^{p+1}]$, which strictly exceeds the optimal social cost of $N^{p+1}(1/2)^p$.
    \item As the effective load $a \to \infty$ ($\sigma_a \to 0$), the ratio of this empirical time-averaged social cost to the optimal social cost converges exactly to $2^p$.\ian{why do we interpret $a$ as both step size and effective load? can you add a sentence explaining the equivalence?}\georgios{I renamed it effective step-size to be consistent with above. Footnote 3 explains how the effective step-size could both the step-size of MWU as well as the total load of the system. Does this make sense?}
\end{enumerate}
\end{theorem}

Theorem \ref{thm:2p_degradation} establishes a severe divergence between static guarantees and discrete-time dynamics. While the classical Price of Anarchy ensures that the worst-case equilibrium is bounded by $\Theta(p/\ln p)$, the physical, day-to-day behavior of the learning dynamics falls into an attractor that is $2^p$ strictly worse than the optimum. \ian{is it possible to strengthen this argument by analyzing the step size $a=2^{p+1}/p$ relative to the maximum payoff in the game or the Lipschitz constant $L$? In other words, is $a$ large at that critical threshold or $O(L^{-1})$? Also, what is the social cost of the period-2 orbit?} \georgios{I have changed the theorem statement to state the social cost of the period-2 orbit.}

\begin{remark}[The Fragility  of Global Finetuning]
From a strictly centralized optimization perspective, one might object that in this highly stylized, symmetric two-path game, an omniscient designer could algebraically reverse-engineer a ``sweet spot'' learning rate $\epsilon$ that yields rapid convergence. Fastest local convergence occurs when the derivative of the update map at the equilibrium is zero ($f_a'(1/2) = 0$), which requires an optimal effective load of $a_{opt} = \frac{2^p}{p}$. 

Given the parameter coupling $a \approx \epsilon N^p$, the corresponding optimal raw learning rate is exactly:
\begin{equation}
\epsilon_{opt} \approx \frac{2^p}{p N^p} = \mathcal{O}(N^{-p})
\end{equation}
While this sweet spot algebraically exists, realizing it is challenging in practice and theoretically misaligned with decentralized learning. 

The fundamental issue is not merely the magnitude of the step-size (this itself is not scale invariant and is dependent on the choice of units), but the  fragility of the required finetuning. First, matching this optimal rate requires an independent agent to perfectly calibrate their learning against the global maximum network cost $\mathcal{O}(N^p)$. This demands oracle-like knowledge of global system parameters—the exact total population $N$ and the precise latency degree $p$—whereas real-world demand is inherently unpredictable and cost functions are only approximate. 

Crucially, rational, selfish agents do not coordinate such fragile global finetuning; they reactively adapt to localized historical cost signals. As established recently by \cite{bielawski2025heterogeneity}, even when populations consist of heterogeneous subpopulations of agents adapting from general initial conditions with generic, uncoordinated learning rates, the macroscopic system natively bifurcates into robust Li-Yorke chaos. Consequently, while an exact $\mathcal{O}(N^{-p})$ step-size guarantees rapid convergence in principle, its realization relies on a degree of systemic coordination and precise global knowledge that is fundamentally at odds with decentralized learning.
\end{remark}

\begin{figure}[htpb]
    \centering
    \includegraphics[width=0.8\textwidth]{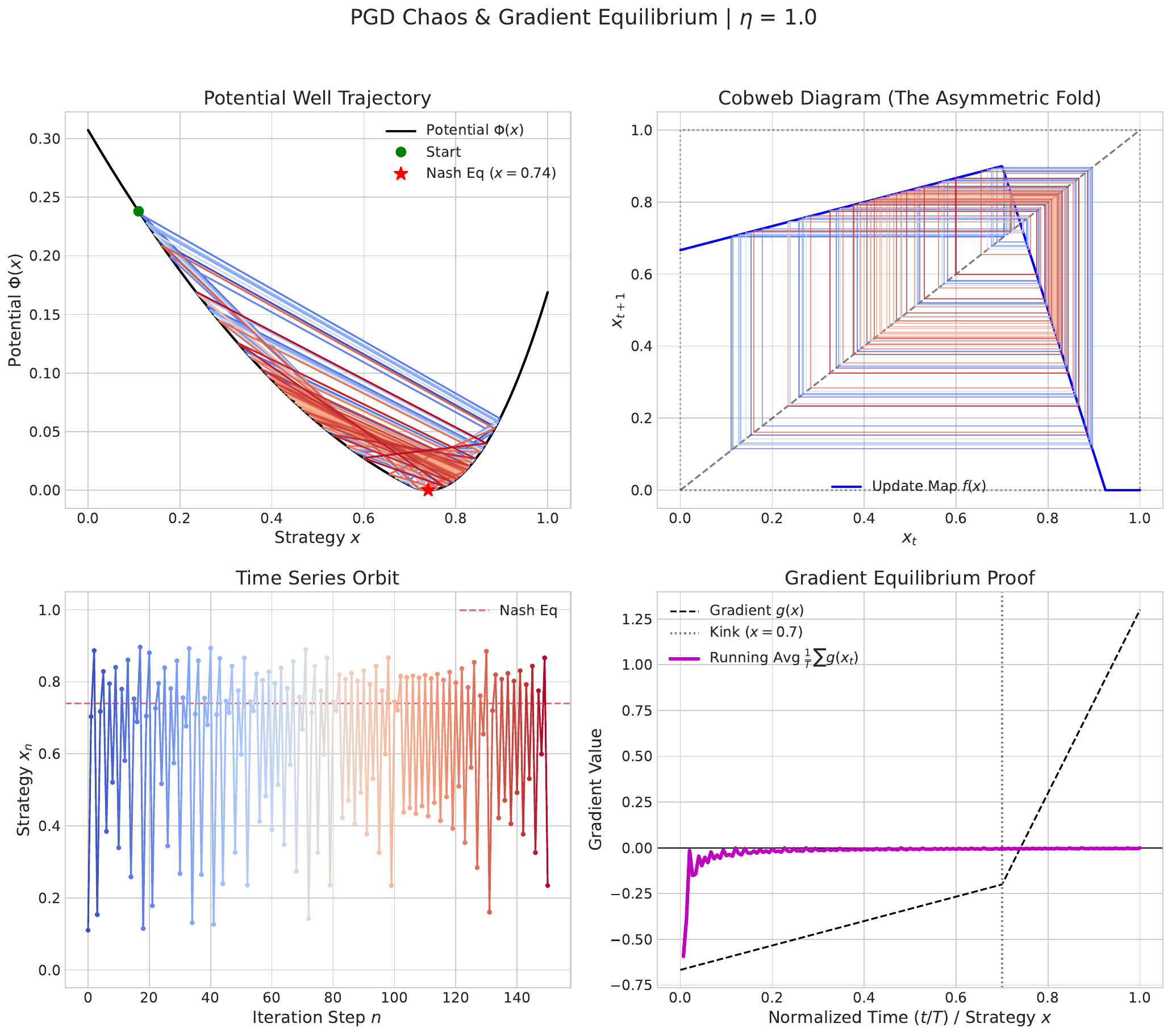}
    \caption{The emergence of Li-Yorke chaos and the concurrent satisfaction of Gradient Equilibrium (GEQ) under discrete-time Projected Gradient Descent (PGD) with $L_2$ regularization ($\eta = 1.0$) in an asymmetric two-link congestion game[cite: 5]. 
    \textbf{(Top-Left)} The state trajectory plotted on the exact convex potential well $\Phi(x)$. Rather than descending to the global minimum (Nash Equilibrium at $x=0.74$), the discrete step-size forces the system to perpetually overshoot, oscillating across the well[cite: 5]. 
    \textbf{(Top-Right)} The cobweb diagram of the unprojected update map $f(x)$. The asymmetric bottleneck induces a steep ``fold'' in the map[cite: 5]. Because the resulting chaotic attractor is strictly bounded away from the boundaries (0 and 1), the Euclidean projection operator is never activated, demonstrating that the chaos is native to the continuous potential landscape itself. 
    \textbf{(Bottom-Left)} The time series of the strategy $x_n$, showcasing persistent, non-equilibrating chaotic fluctuations around the equilibrium. 
    \textbf{(Bottom-Right)} The empirical validation of the GEQ condition~\cite{JMLR:v26:25-0356}. Despite the severe chaotic variance of the primal state $x_n$, the time-averaged gradient strictly converges to zero over time. This illustrates the fundamental paradox: time-averaged convergence of the dual/gradient variables perfectly masks the volatile, inefficient reality of the primal day-to-day flows. See Appendix \ref{app:l2_chaos} for the formal construction.}
    \label{fig:l2_chaos}
\end{figure}

Furthermore, this dynamical instability is not merely an artifact of the boundary-avoiding nature of steep regularizers (like Shannon entropy in MWU). It is a robust topological consequence of discrete-time updates on continuous potential landscapes. As illustrated in Figure \ref{fig:l2_chaos}, applying standard $L_2$ regularization (i.e., Projected Gradient Descent or Lazy Mirror Descent) in an asymmetric network natively induces Li-Yorke chaos within the interior of the simplex.

\begin{remark}[Gradient Equilibria, Jensen's Gap, and Physical Attractors]
An elegant phenomenon observed in these unstable regimes is that the time-averaged cost difference between the two paths strictly vanishes over time \cite{CFMP2019}. From the perspective of recent literature, this formally implies that the system satisfies the conditions of a \textit{Gradient Equilibrium} (GEQ)~\cite{JMLR:v26:25-0356}. This establishes a profound parallel: just as Correlated Equilibria accommodate complex chaotic attractors in continuous-time dynamics, Gradient Equilibria formally accommodate discrete-time Li-Yorke chaos, even in the most benign convex optimization settings.

The GEQ framework explicitly acknowledges that such states do not enforce standard optimality and can incur linear regret. Theorem \ref{thm:2p_degradation} provides the exact, worst-case quantification of this inefficiency via Jensen's gap: the system locks into a maximal-variance period-2 orbit where the gradients perfectly average to zero, yet the strict convexity of the macroscopic cost function guarantees a severe $2^p$ inefficiency.

However, the most critical limitation of GEQ lies not in its inefficiency, but in its predictive looseness. Because GEQ is fundamentally a kinematic condition on the time-averaged gradients, it structurally bypasses dynamical stability. By the definition of Li-Yorke chaos, the GEQ condition is simultaneously satisfied by infinitely many unstable periodic orbits, as well as an uncountable scrambled set of aperiodic trajectories. Yet, in well-behaved physical reality, almost all initial conditions are captured by just one or a few structurally attracting limit sets (such as the period-2 attractor). The GEQ solution concept lacks the mathematical machinery to distinguish between these uncountably many unobservable ``ghost'' orbits and the true physical attractor. Consequently, to generate more accurate predictions within the GEQ framework, one must deploy the tools of ergodic theory and dynamical systems to isolate the specific invariant measures that physically manifest.
\end{remark}

\FloatBarrier

\section{Related Work}
\label{sec:related_work}

This paper bridges classical algorithmic game theory with the continuous state-space analysis of dynamical systems, prompting a formal re-evaluation of prescriptive solution concepts. Our work extends beyond adjusting existing performance bounds to systematically analyze the geometric and topological validity of the target states anchoring three major pillars of the classical literature.

\subsection{Equilibrium Relaxations and the Topology of Potential Games}
The observation that natural learning dynamics fail to converge to Nash equilibria is well documented, tracing back to classical non-convergence results in cyclic environments by Shapley and Jordan~\cite{Shapley64,Jordan475}. Historically, the algorithmic game theory and multi-agent learning communities responded to these impossibility results---such as the non-convergence of uncoupled dynamics~\cite{hart2003uncoupled,hart2006stochastic,milionis2023impossibility}---by relaxing the target solution concepts (e.g., to Correlated and Coarse Correlated Equilibria) or by augmenting the learning algorithms with stochasticity and memory. Under these relaxations, the equilibrium-centric paradigm has remained foundational, scaling to extensive-form games~\cite{farina2019efficient,anagnostides2022faster,zhang2024efficient}, Markov and stochastic games~\cite{daskalakis2023complexity,jin2021v,mao2023provably,song2021can,cai2024near,gempconvex}, polyhedral games~\cite{farina2024polynomial}, convex and non-convex games~\cite{daskalakis2024efficient,hazan2017efficient}, time-varying environments~\cite{anagnostides2023convergence}, and continuous/differentiable games~\cite{balduzzi2018mechanics,mazumdar2020gradient,mertikopoulos2019optimistic,daskalakis2018training}. 

Rather than evaluating these equilibrium concepts in adversarial or cyclic environments where non-convergence is topologically expected, we focus our analysis on exact potential and congestion games. These environments form the canonical baseline for Algorithmic Game Theory as a whole~\cite{Nisan:2007:AGT:1296179,roughgarden2016twenty}; they serve as the foundation for seminal Price of Anarchy (PoA) bounds~\cite{KoutsoupiasP99WorstCE,christodoulou2005price,awerbuch2005price,roughgarden2015intrinsic} and are conventionally characterized by guaranteed convergence and macroscopic stability. By examining these constructions through the formalism of continuous optimization theory, we show that the pure Nash equilibria engineered to tightly bind the classical PoA inequalities are topologically unstable. Because their construction enforces exact indifference---embedding the algebraic structure of a mixed equilibrium onto the boundary of the simplex---they manifest locally as strict saddles of the extended potential. Furthermore, because their stable manifolds reside entirely outside the valid physical state space, they act dynamically as global repellers and local maxima of the potential function. Consequently, the canonical tight bounds are governed by implicit geometric constraints, anchoring efficiency metrics to topological repellers that diverge under gradient-like learning trajectories.

%ather than evaluating these equilibrium concepts in adversarial or zero-sum environments where non-convergence is topologically expected, we focus our analysis on exact potential and congestion games. These environments serve as the foundation for seminal Price of Anarchy (PoA) bounds~\cite{poa,christodoulou2005price,awerbuch2005price,Nisan:2007:AGT:1296179,roughgarden2016twenty} and are conventionally characterized by guaranteed stability. By examining these constructions through the formalism of continuous optimization theory, we show that the pure Nash equilibria engineered to tightly bind the classical PoA inequalities are mathematically degenerate. Because their construction enforces exact indifference---embedding the algebraic structure of a mixed equilibrium onto the boundary of the simplex---they manifest locally as strict saddles of the extended potential. Furthermore, because their stable manifolds reside entirely outside the valid physical state space, they act dynamically as global repellers and local maxima of the potential function. Consequently, the canonical tight bounds are governed by implicit geometric constraints, anchoring efficiency metrics to topological repellers that strictly diverge under gradient-like learning trajectories.

\subsection{Correlated Play and the Limitations of $C^0$ Analysis}
Given the topological instability of these deterministic anchors, their probabilistic generalizations require analogous geometric scrutiny. The reduction of multi-agent learning to regret minimization~\cite{littlestone1994weighted,freund1999adaptive,hart2000simple} and its extension into robust Price of Anarchy~\cite{roughgarden2015intrinsic} provided a tractable algebraic framework. Supported by results establishing near-optimal convergence rates for time-averaged (social) regret and robust PoA social welfare guarantees   ~\cite{Syrgkanis:2015:FCR:2969442.2969573,daskalakis2021near}, this paradigm has motivated numerous follow-ups~\cite{foster2016learning,anagnostides2022faster,daskalakis2023complexity,farina2020coarse,jin2021v,mao2023provably,song2021can,piliouras2022beyond,soleymani2025faster,peng2024fast,celli2020no} including proximal refinements as well as gradient-based variations~\cite{angelopoulos2025gradient,cai2025proximal,ahunbay2024first,ahunbay2025semicoarse}.

%li2024survey,peng2024fast,anagnostides2022near,anagnostides2022faster,dagan2024external,Daskalakis2024:efficient,papadimitriou2008computing,jiang2015polynomial,farina2024polynomial

We establish that reliance on discrete empirical histories (time-averaging) introduces  structural slackness. By evaluating trajectories exclusively through time-averaged distributions, optimal regret minimization can obscure underlying physical instability. We prove that uncoupled learning and continuous proximal refinements can be supported on strictly dominated strategies. Furthermore, optimal zero swap-regret dynamics can admit chaos and positive Lyapunov exponents even in minimal normal-form games. %Thus, while the algorithms satisfy standard algebraic regret criteria, the underlying physical system  remains in a non-equilibrating chaotic regime.

%This over-inclusiveness inadvertently  carries over to the set of 
%approximately/near-optimal states~\cite{rough09,Syrgkanis:2015:FCR:2969442.2969573} (strict relaxation of CCEs), %which are shown to be consistent 
%the states of approximately optimal social welfare in the sense of robust PoA. Specifically, this type of approximate socially optimality is consistent 
%in which every agent can be arbitrarily worse off than their minmax safety performance even in games with strongly aligned incentives.

This over-inclusiveness inadvertently carries over to the set of approximate or near-optimal states~\cite{rough09,Syrgkanis:2015:FCR:2969442.2969573} (a strict relaxation of CCEs), %which are shown to be consistent
%the states of approximately optimal social welfare in the sense of robust PoA. Specifically, this type of approximate socially optimality is consistent
in which every agent can be arbitrarily worse off than their minmax safety performance even in settings with strongly aligned incentives.
Such inefficiencies do not require artificial constructions. They emerge from the strategic insensitivity of interior NE. Such interior states emerge as solutions in all doubly symmetric load balancing games (symmetric players and machines), while the expected cost of each agent is equal to their independent min-max safety costs. Moving from interior NE to relaxations allowing for correlated play further weakness these performance guarantees.

\subsection{Non-Atomic Games and Macroscopic Instability}
Finally, we re-evaluate the stability of the non-atomic limit, an environment historically characterized by tight sub-linear efficiency bounds~\cite{roughgarden2002bad,roughgarden2004bounding,blum2006routing}. Recent literature has identified instability and Li-Yorke chaos in non-atomic routing games for numerous dynamics including Multiplicative Weights, Follow-the-Regularized-Leader dynamics with steep regularizers, Experienced Weighted Attraction learning~\cite{CFMP2019,bielawski2021follow,bielawski2025heterogeneity,bielawski2024memory}. Despite the prevalence of such results, the impact of such unstable behavior to system performance is 
less understood. The only  
prior known result  established globally attracting period-2 limit cycles only in the case of linear congestion games, corresponding only to a mild degradation of social welfare over the canonical PoA bounds (from $4/3$ to $2$ times worse than optimal).

Our analysis extends the scope of this macroscopic instability. We establish the first Li-Yorke chaos result for non-steep regularizers (such as $L_2$), demonstrating that instability is a robust topological feature of standard discrete-time optimization algorithms such as (Projected) Gradient Descent. Second, we show that addressing polynomial cost functions of degree $p \ge 3$ necessitates more intricate arguments compared to typically analyzed linear variants. We prove that in these higher-degree environments, the system destabilizes into a global period-two attractor whose the time-averaged inefficiency scales exponentially as $2^p$, diverging fundamentally from the sub-linear $\Theta(p/\ln p)$ static bounds. This formally decouples the classical algebraic PoA guarantees from the worst-case discrete-time dynamical reality raising many novel questions and analytical challenges.

\section{Conclusion}
\label{sec:conclusion}

Algorithmic Game Theory has traditionally derived its prescriptive guarantees from the efficient computation of static equilibria and the formulation of worst-case performance bounds. Concurrently, the reduction of multi-agent learning to black-box regret minimization has provided a standard mathematical framework for bounding system inefficiency. In this paper, we demonstrated that this analytical paradigm is structurally decoupled from the topological and measure-theoretic realities of game dynamics. By evaluating multi-agent systems purely through $C^0$ fixed-point analysis or discrete empirical distributions, standard solution concepts inherently abstract away the $C^1$ continuation of the underlying learning maps and expected game payoffs.

Through a series of constructive results, we established that this $C^0$ restriction introduces fundamental analytical limitations across five dimensions:

\begin{enumerate}
    \item \textbf{The Strategic Insensitivity of (Interior) Equilibria:} (Interior) Nash equilibria lack the local gradient information necessary to distinguish between aligned and strictly opposing incentives. Up to strategic equivalence, they function as adversarial ``each-vs.-all'' games, locking agents into minimal, min-max safety payoffs even in perfectly cooperative potential games.
    
    \item \textbf{The Topological Instability of Worst-Case Pure Nash:} The worst-case pure Nash equilibria that dictate the tightness of robust Price of Anarchy bounds mathematically manifest as strict saddles of the extended potential. Consequently, under natural learning dynamics, they are topologically unstable and possess a region of attraction of measure zero. In canonical instances, they are dynamically inaccessible and physically manifest as states of maximum potential, actively repelling valid trajectories and supported on \textit{almost everywhere strictly dominated} strategies---actions that are strictly inferior against all but a measure-zero set of mixed opponent profiles.
    
    \item \textbf{The Sensitivity of Efficiency Metrics:} Classical PoA guarantees in affine congestion games are highly sensitive to specific algebraic representations. When standard non-negativity constraints are relaxed to accommodate strictly increasing, data-driven cost models ($x \ge 1$), both the worst-case PoA and the Average PoA become strictly unbounded. Consequently, the metric evaluates performance based on parameters that are not invariant under strategic equivalence.
    
    \item \textbf{The Combinatorial Shadow of Correlated Play:} Projecting trajectories from the continuous state space of learning algorithms onto the discrete simplex of correlated equilibria is structurally insufficient to preclude non-rationalizable behavior. Optimal uncoupled learning---and continuous proximal refinements (SCE/PCE)---can be supported entirely on strictly dominated strategies. Furthermore, optimal $O(1/T)$ swap-regret minimization does not preclude complex, non-equilibrating dynamics, manifesting as chaotic limit sets even in minimal normal-form games.

    \item \textbf{Exponential Inefficiency and Chaos in the Non-Atomic Limit:} In non-atomic congestion games, standard discrete-time learning dynamics natively induce Li-Yorke chaos. Rather than stabilizing at the Wardrop equilibrium, the system is captured by a global chaotic attractor whose empirical time-averaged social cost degrades exponentially as $2^p$, severely diverging from static sub-linear bounds.
\end{enumerate}

These divergences stem directly from a worst-case methodology that assigns equal analytical weight to all algebraically feasible states, regardless of their dynamic stability. By constructing efficiency bounds that must mathematically account for measure-zero topological deadlocks, dynamically inaccessible states, and chaotic limit sets, the current framework risks structurally decoupling theoretical guarantees from the typical trajectories of learning agents.

The paradigms of worst-case equilibrium analysis and black-box regret minimization have undeniably served as indispensable foundations for algorithmic game theory. They provided the essential computational tractability required to navigate the PPAD-hardness of multi-agent systems and established the first rigorous benchmarks for distributed learning. However, as the field matures toward modeling data-driven, empirically grounded environments, our analytical tools must evolve to capture higher-fidelity dynamic realities. 

The  path forward requires directly interrogating the continuous state-space geometry of the learning dynamics---irrespective of whether the algorithmic updates occur in discrete or continuous time---using tools from ergodic theory and the study of dynamical invariant measures. By shifting focus to physical and natural invariant measures---limit sets possessing basins of attraction with strictly positive Lebesgue measure---we can systematically isolate the specific, dynamically stable outcomes within the broader worst-case sets (see e.g.,~\cite{bielawski2026natural} for some first applications of these ideas to game dynamics). Integrating the topological and measure-theoretic mechanics of dynamical systems into Algorithmic Game Theory provides the necessary foundation for formulating robust, predictive, and empirically grounded guarantees for complex multi-agent systems.

% Bibliography
\bibliographystyle{ACM-Reference-Format}
\bibliography{references/refs,references/refs2,references/limits,references/sigproc5}

%\appendix

\appendix

\section*{}

\section{Proofs of Claims in Section~\ref{sec:OPT}}
\label{app:A}

\subsection{Proof of Theorem~\ref{thm:zero-measure}}

Before formally stating the theorem~\ref{thm:zero-measure}, we introduce the primary mathematical tool utilized to establish topological instability in cases where the boundary metric becomes degenerate. When the Jacobian evaluates identically to the zero matrix, standard linear approximation fails, and we must rely entirely on the exact geometric properties of the non-linear vector field. For this, we employ Chetaev's Instability Theorem, a foundational result in non-linear dynamical systems theory~\cite{khalil2002nonlinear}.

\begin{theorem}[Chetaev's Instability Theorem \cite{khalil2002nonlinear}, Theorem 4.3]
\label{thm:chetaev}
Let $\mathbf{0}$ be an equilibrium point of the continuous dynamical system $\dot{x} = f(x)$, where $f$ is continuously differentiable on a local neighborhood $D \subset \mathbb{R}^N$ around the origin. Suppose there exists a continuously differentiable function $W(x): D \to \mathbb{R}$ and an open sub-region $\Omega \subset D$ such that:
\begin{enumerate}
    \item The origin $\mathbf{0}$ lies on the boundary of the region $\Omega$.
    \item $W(x) > 0$ for all $x \in \Omega$.
    \item $W(x) = 0$ for all $x \in \partial\Omega \cap D$ (the boundary of $\Omega$ that lies strictly inside $D$).
    \item $\dot{W}(x) = \nabla W(x) \cdot f(x) > 0$ for all $x \in \Omega$ (meaning $W$ strictly increases along trajectories).
\end{enumerate}
Then, the equilibrium point $\mathbf{0}$ is topologically unstable.
\end{theorem}

This theorem provides the rigorous mathematical foundation for Property (iii) of our main result. By defining a localized escape region inside the physically valid probability space and demonstrating that an inverted potential function strictly grows within it, Chetaev's Theorem allows us to bypass the zero-Jacobian boundary degeneracy entirely and guarantee that the equilibrium actively repels interior trajectories.

\begin{reptheorem}{thm:zero-measure}[Topological Instability of Worst-Case Nash Equilibria]
Let $\Gamma$ be a canonical exact potential congestion game utilized to establish tight robust Price of Anarchy (PoA) bounds (e.g., the affine double-cycle of~\cite{christodoulou2005price,roughgarden2015intrinsic}, or the multi-hop network of~\cite{awerbuch2005price,Nisan:2007:AGT:1296179}) where its worst-case pure Nash equilibrium $y$ perfectly binds the smoothness inequalities. Under these exact algebraic constraints, the following properties hold:
\begin{enumerate}
    \item[(i)] The worst-case pure Nash equilibrium $y$ manifests as a \textbf{strict saddle} of the exact continuous potential function $\Phi$.
    \item[(ii)] Under any smooth learning dynamic defined by an interior-regular Riemannian metric (strictly positive-definite at the boundary), the Lebesgue measure of the basin of attraction for $y$ within the physical probability simplex is exactly zero.
    \item[(iii)] Under any continuous learning dynamic where the exact potential $\Phi$ acts as a strict Lyapunov function outside the equilibrium set, the equilibrium $y$ is \textbf{topologically unstable}, actively repelling physically valid learning trajectories away from it.
\end{enumerate}
\end{reptheorem}

\begin{proof}
The proof proceeds in four structural parts. First, we analytically extend the potential function to circumvent the boundary topology and link algebraic tightness to an interior critical point. Second, we calculate the exact Hessians to establish the strict saddle. Third, we compute explicit eigenvectors to guarantee the unstable directions intersect the physical interior. Finally, we map this geometry to continuous learning dynamics, proving physical repulsion and topological instability.

\paragraph{Part I: Analytic Extension and the Critical Point}
The Smoothness Framework algebraically bounds the inefficiency of a game via parameters $(\lambda, \mu)$ such that for any states $y$ and $y^*$:
\begin{equation}
\sum_{i=1}^N C_i(y_i^*, y_{-i}) \le \lambda C(y^*) + \mu C(y)
\end{equation}
To prove that this bound is perfectly tight, the framework requires a game parameterization where the inequality becomes a strict equality at the worst-case pure Nash equilibrium $y$. This algebraically mandates that every player $i$ is exactly indifferent between their worst-case strategy $y_i$ and a unilateral deviation to their socially optimal strategy $y_i^*$. 

Let the independent mixed strategy profile be represented by $q \in [0,1]^N$, where $q_i$ represents the probability that player $i$ plays $y_i^*$, and $1-q_i$ is the probability they play $y_i$. Let the worst-case state $y$ be mapped to the origin, $q = \mathbf{0}$. 

Because the latency functions of canonical bounding games are polynomial, the exact potential function $\Phi(q)$ is a continuous polynomial over the closed simplex $[0,1]^N$. Because polynomials are infinitely differentiable ($C^\infty$) and analytic everywhere, we smoothly extend $\Phi$ from the closed simplex to the open, unconstrained space $\mathbb{R}^N$. Within $\mathbb{R}^N$, the boundary origin $q = \mathbf{0}$ formally acts as an interior point of a manifold without boundary.

Let $\Delta_i(q) = \mathbb{E}[C_i(y_i, q_{-i})] - \mathbb{E}[C_i(y_i^*, q_{-i})]$ denote the expected cost difference. The strict indifference condition mathematically mandates:
\begin{equation}
\Delta_i(\mathbf{0}) = 0 \quad \forall i \in \{1, \dots, N\}
\end{equation}
In an exact potential game, the expected cost difference corresponds to the negative gradient of the potential function: $\Delta_i(q) = -\frac{\partial \Phi}{\partial q_i}$. Because $\Delta_i(\mathbf{0}) = 0$ for all $i$, the gradient vanishes locally at the origin in the extended space:
\begin{equation}
\nabla \Phi(\mathbf{0}) = \mathbf{0}
\end{equation}
Therefore, the algebraic requirement of tightness formally forces the worst-case Nash equilibrium to act as an interior critical point of the extended potential function.

\paragraph{Part II: Explicit Calculation of the Hessian and the Strict Saddle}
We evaluate the local geometry of the potential function at the interior critical point by computing its Hessian matrix, $H$, where $H_{ij} = \frac{\partial^2 \Phi}{\partial q_i \partial q_j} = -\frac{\partial \Delta_i}{\partial q_j}$.

First, we examine the main diagonal elements $H_{ii} = -\frac{\partial \Delta_i}{\partial q_i}$. In any congestion game utilizing independent mixed strategies, the expected latency of a path depends exclusively on the background traffic generated by the opposing players ($q_{-i}$). Player $i$'s own mixed probability $q_i$ dictates the weighting of their choice, but it does not alter the physical distribution of the opposing traffic. Therefore, the partial derivative of player $i$'s marginal cost with respect to their own probability is identically zero:
\begin{equation}
\frac{\partial \Delta_i}{\partial q_i} = 0 \quad \forall i \in \{1, \dots, N\}
\end{equation}
Consequently, the main diagonal of the Hessian is entirely zero, guaranteeing that the trace of the Hessian vanishes universally:
\begin{equation}
\text{Trace}(H) = \sum_{i=1}^N H_{ii} = 0
\end{equation}

By Clairaut's Theorem, $H$ is a real symmetric matrix. To definitively establish that $q = \mathbf{0}$ is a strict saddle, we must prove that $H \neq \mathbf{0}$. We achieve this by explicitly calculating the cross-derivatives for the canonical tightness construction.

\textbf{Case 1: The 4-Player \cite{awerbuch2005price} Network} \\
In the standard 4-player construction yielding the $5/2$ bound, the expected cost differences $\Delta_i(q)$ at the worst-case pure Nash equilibrium evaluate to:
\begin{align*}
\Delta_1(q) &= 2q_2 - q_3 + 2q_4 \\
\Delta_2(q) &= 2q_1 + 2q_3 - q_4 \\
\Delta_3(q) &= -q_1 + 2q_2 \\
\Delta_4(q) &= 2q_1 - q_2
\end{align*}
Taking the partial derivatives $H_{ij} = -\frac{\partial \Delta_i}{\partial q_j}$ yields the exact, explicit Hessian matrix at the origin:
\begin{equation}
\label{eq:hessian_4p}
H = \begin{pmatrix}
0 & -2 & 1 & -2 \\
-2 & 0 & -2 & 1 \\
1 & -2 & 0 & 0 \\
-2 & 1 & 0 & 0
\end{pmatrix}
\end{equation}
This matrix is explicitly non-zero and symmetric. Because its trace is $0$, its eigenvalues must cross the origin. Calculations verify the eigenvalues are $-1+\sqrt{2} \approx 0.414$, $-1-\sqrt{2} \approx -2.414$, $1+\sqrt{10} \approx 4.162$, and $1-\sqrt{10} \approx -2.162$. The existence of strictly negative eigenvalues formally guarantees the state is a strict saddle.

\textbf{Case 2: The Worst-Case Topologies of~\cite{christodoulou2005price},\cite{roughgarden2015intrinsic}} \\
We evaluate the exact topological construction that tightly binds the affine Price of Anarchy at $5/2$ (\cite{christodoulou2005price}, \cite{roughgarden2015intrinsic} Theorem 5.6, Example 5.7). The game consists of $k=3$ players $\{0, 1, 2\}$ and two disjoint cycles of resources $u$ and $v$, totaling 6 resources: $\{u_0, u_1, u_2, v_0, v_1, v_2\}$. The cost functions are the identity $c(x) = x$. 

Each player $i$ chooses between an optimal strategy $Q_i = \{u_i, v_i\}$ and a worst-case strategy $P_i = \{u_{i+1}, v_{i+1}, v_{i+2}\}$ (indices modulo 3). Let $q_i$ represent the probability that player $i$ plays $Q_i$.

To determine if the four-way interaction between $P_i, Q_i$ and $P_{i-1}, Q_{i-1}$ cancels out, we explicitly calculate the expected cost difference $\Delta_0(q)$ for Player 0 as a function of the probabilities of the other players ($q_1, q_2$).

First, we calculate the expected cost of Player 0's worst-case strategy, $\mathbb{E}[C_0(P_0)]$, evaluating the load on its specific edges $P_0 = \{u_1, v_1, v_2\}$. Assuming Player 0 plays $P_0$ deterministically, they add a load of $1$ to these edges:
\begin{itemize}
    \item Load on $u_1$: Player 0 plus Player 1 (if playing $Q_1$). Expected load $= 1 + q_1$.
    \item Load on $v_1$: Player 0 plus Player 1 (if playing $Q_1$) plus Player 2 (if playing $P_2$). Expected load $= 1 + q_1 + (1 - q_2) = 2 + q_1 - q_2$.
    \item Load on $v_2$: Player 0 plus Player 1 (if playing $P_1$) plus Player 2 (if playing $Q_2$). Expected load $= 1 + (1 - q_1) + q_2 = 2 - q_1 + q_2$.
\end{itemize}
Summing these yields the expected cost of the worst-case path:
\begin{equation}
\mathbb{E}[C_0(P_0)] = (1 + q_1) + (2 + q_1 - q_2) + (2 - q_1 + q_2) = 5 + q_1
\end{equation}

Next, we calculate the expected cost of Player 0's optimal strategy, $\mathbb{E}[C_0(Q_0)]$, evaluating the load on $Q_0 = \{u_0, v_0\}$ assuming Player 0 plays $Q_0$ deterministically:
\begin{itemize}
    \item Load on $u_0$: Player 0 plus Player 2 (if playing $P_2$). Expected load $= 1 + (1 - q_2) = 2 - q_2$.
    \item Load on $v_0$: Player 0 plus Player 1 (if playing $P_1$) plus Player 2 (if playing $P_2$). Expected load $= 1 + (1 - q_1) + (1 - q_2) = 3 - q_1 - q_2$.
\end{itemize}
Summing these yields the expected cost of the optimal path:
\begin{equation}
\mathbb{E}[C_0(Q_0)] = (2 - q_2) + (3 - q_1 - q_2) = 5 - q_1 - 2q_2
\end{equation}

The expected marginal cost difference for Player 0 is:
\begin{equation}
\Delta_0(q) = \mathbb{E}[C_0(P_0)] - \mathbb{E}[C_0(Q_0)] = (5 + q_1) - (5 - q_1 - 2q_2) = 2q_1 + 2q_2
\end{equation}

At the worst-case Nash equilibrium ($q = \mathbf{0}$), the indifference condition holds perfectly: $\Delta_0(\mathbf{0}) = 0$. However, the four-way interaction structurally does not cancel out. Taking the partial derivatives yields strictly non-zero values:
\begin{equation}
\frac{\partial \Delta_0}{\partial q_1} = 2 \quad \text{and} \quad \frac{\partial \Delta_0}{\partial q_2} = 2
\end{equation}
By the symmetric construction of the cycles, this holds for all players: $\Delta_1(q) = 2q_0 + 2q_2$ and $\Delta_2(q) = 2q_0 + 2q_1$. Thus, the Hessian matrix evaluated at the origin is explicitly:
\begin{equation}
\label{eq:hessian_3p}
H = \begin{pmatrix}
0 & -2 & -2 \\
-2 & 0 & -2 \\
-2 & -2 & 0
\end{pmatrix}
\end{equation}
Because this matrix is generated by the structure of the universal worst-case topology, $H \neq \mathbf{0}$ is a physical mandate of theoretical tightness. 

The characteristic equation of $H$ yields the explicit eigenvalues $\lambda \in \{2, 2, -4\}$. Because the spectrum contains at least one strictly negative eigenvalue ($\lambda_{\min} = -4$), the critical point $q = \mathbf{0}$ satisfies the definition of a \textbf{strict saddle}, establishing Property (i).

\paragraph{Part III: Physical Intersection via Explicit Local Geometry}
The standard invariant manifold theorems apply to vector fields defined on open subsets. Because we smoothly extended the potential function from the closed simplex to an open neighborhood $U \subset \mathbb{R}^N$, the state $q = \mathbf{0}$ formally acts as an interior fixed point within $U$. 

We must guarantee that the local stable manifold $W^s$ of any continuous learning dynamic does not unilaterally consume the physically valid probability simplex $(0,1)^N$. First, we note that because the explicit Hessians derived in Part II are non-singular, all eigenvalues of the Jacobian are non-zero. Consequently, the boundary state $\mathbf{0}$ is a strictly hyperbolic fixed point, and the center manifold is empty. 

By the Stable/Unstable Manifold Theorem, we must prove that the unstable manifold $W^u$ points into the positive orthant. This requires establishing that the eigenvectors corresponding to the negative eigenvalues of the Hessian $H$ (which define the unstable, repelling directions of the potential landscape) intersect the physical interior. Rather than relying on aggregate quadratic forms, we evaluate these unstable eigenvectors directly. 

For the 3-player worst-case topology (Equation \ref{eq:hessian_3p}), evaluating the vector $v = (1, 1, 1)^T$ yields:
$$H v = \begin{pmatrix} 0 & -2 & -2 \\ -2 & 0 & -2 \\ -2 & -2 & 0 \end{pmatrix} \begin{pmatrix} 1 \\ 1 \\ 1 \end{pmatrix} = \begin{pmatrix} -4 \\ -4 \\ -4 \end{pmatrix} = -4 v$$
Thus, $v = (1, 1, 1)^T$ is an exact eigenvector corresponding to the strictly negative eigenvalue $\lambda = -4$. 

Similarly, for the 4-player topology (Equation \ref{eq:hessian_4p}), the characteristic polynomial yields a strictly negative eigenvalue $\lambda = -1 - \sqrt{2}$. Solving the system $(H - \lambda I)v = 0$ yields the corresponding exact eigenvector:
$$v = \begin{pmatrix} 1 \\ 1 \\ \sqrt{2} - 1 \\ \sqrt{2} - 1 \end{pmatrix}$$

Because $\sqrt{2} > 1$, every component of this eigenvector is strictly positive. In both canonical topologies, the explicit eigenvectors dictating the unstable manifold $W^u$ consist strictly of positive components. Therefore, the tangent space of the repelling dynamics at the origin points strictly into the valid probability simplex $(0,1)^N$. This establishes the geometric foundation for physical repulsion.

\paragraph{Part IV: Continuous Learning Dynamics and Topological Instability}
We now link the local geometry of the potential function to the vector field of the continuous-time learning dynamic. We evaluate the local stability of the critical point $q=\mathbf{0}$ by bifurcating our analysis into two broad classes of gradient-like dynamics. 

First, we consider smooth Riemannian gradient flows of the form:
$$V(q) = M(q)\Delta(q) = -M(q)\nabla\Phi(q)$$
where $M(q)$ is a smooth, positive semi-definite matrix representing the inverse Riemannian metric tensor. Second, we generalize our analysis to encompass any continuous learning dynamic for which the exact potential function $\Phi$ acts as a strict Lyapunov function outside the equilibrium set. The mathematical machinery required to establish instability depends strictly on whether the dynamic's boundary behavior permits linear approximation.

\textbf{Case 1: Interior-Regular Metrics.} 
Assume the continuous dynamic is governed by an interior-regular Riemannian metric, where $M(q)$ remains symmetric and strictly positive-definite at the boundary ($M(\mathbf{0}) \succ 0$). We evaluate local stability via the Jacobian matrix, $J = DV(\mathbf{0})$. Because $q=\mathbf{0}$ is a critical point $(\nabla\Phi(\mathbf{0})=\mathbf{0})$, the product rule yields $J = -M(\mathbf{0})H$. 

Because $M(\mathbf{0})$ is positive-definite, it admits a unique positive-definite square root $M^{1/2}$, allowing us to express the Jacobian as similar to a symmetric matrix: $J = -M^{1/2}(M^{1/2} H M^{1/2})M^{-1/2}$. By Sylvester's Law of Inertia, the congruent matrix $M^{1/2} H M^{1/2}$ shares the exact same inertia as the Hessian $H$. As established in Part II, $H$ is non-singular and possesses at least one strictly negative eigenvalue, meaning $-H$ possesses at least one strictly positive eigenvalue. 

Consequently, $J$ is guaranteed to possess strictly positive eigenvalues for any interior-regular metric, proving the state is a strictly hyperbolic fixed point. By the Stable Manifold Theorem, the dimension of the stable convergent manifold $W^s$ is bounded to at most $N-1$. Because any smooth manifold of dimension strictly less than $N$ has a Lebesgue measure of exactly zero in $\mathbb{R}^N$, its intersection with the open $N$-dimensional physical state space $(0,1)^N$ trivially has a Lebesgue measure of exactly zero. This globally establishes Property (ii) without requiring the dynamic's specific unstable eigenvectors to perfectly align with those of the potential landscape.

\textbf{Case 2: Degenerate Metrics and General Gradient-Like Dynamics.} 
For proportional evolutionary dynamics (e.g., Replicator Dynamics) or other highly non-linear learning rules, the metric tensor may evaluate to exactly zero at pure strategy profiles ($M(\mathbf{0}) = \mathbf{0}$). In such cases, the Jacobian evaluates identically to the zero matrix ($J=\mathbf{0}$), rendering standard linear stability and manifold dimension theorems inapplicable.

However, we can establish Property (iii) universally for any gradient-like learning dynamic where the potential $\Phi$ acts as a strict Lyapunov function outside the equilibrium set. We do this directly via Chetaev's Instability Theorem. We define a continuously differentiable Chetaev function by inverting the potential function relative to the boundary equilibrium:
$$W(q) = \Phi(\mathbf{0}) - \Phi(q)$$
We define an escape region $\Omega$ as the set of physical states within a local neighborhood $B_\epsilon(\mathbf{0})$ where the potential is strictly less than the boundary equilibrium:
$$\Omega = \{ q \in (0,1)^N \cap B_\epsilon(\mathbf{0}) \mid \Phi(q) < \Phi(\mathbf{0}) \}$$

Because the explicit eigenvectors derived in Part III correspond to strictly negative eigenvalues ($\lambda < 0$), the quadratic form evaluates to $v^T H v = \lambda \|v\|^2 < 0$ along these interior rays. The exact potential function strictly decreases as trajectories move along these vectors into the interior. Thus, $\Omega$ is non-empty, and the origin $\mathbf{0}$ lies exactly on its boundary. By definition, $W(q) > 0$ strictly inside $\Omega$, and $W(q) = 0$ on the boundary of $\Omega$ adjacent to the origin.

We evaluate the time derivative of $W(q)$ along the learning trajectories inside $\Omega$. Because $\Phi(q)$ is a strict Lyapunov function outside the equilibrium set, and $\mathbf{0}$ is an isolated worst-case equilibrium, the potential must strictly decrease along any physical trajectory within the non-equilibrium region $\Omega$. Therefore:
$$\frac{d}{dt}\Phi(q) < 0 \implies \dot{W}(q) = - \frac{d}{dt}\Phi(q) > 0$$

Because $W(q)$ strictly increases along trajectories inside $\Omega$, the vector field actively forces physically valid states away from the origin toward higher values of $W(q)$ (lower values of $\Phi(q)$). By Chetaev's Theorem, this mathematically guarantees that the worst-case Nash equilibrium $y$ is topologically unstable. The dynamic actively repels interior trajectories away from it, establishing Property (iii).
\end{proof}

%\end{proof}

\begin{remark}[Genericity in Arbitrary Cost Functions]
We have explicitly established topological instability for the canonical bounding topologies of \cite{awerbuch2005price, christodoulou2005price, roughgarden2015intrinsic}. For extensions to arbitrary non-decreasing cost functions, the analytical recipe established in this proof fundamentally persists. First, the algebraic tightness of the efficiency bound universally forces the worst-case equilibrium to act as an interior critical point. Second, the independence of mixed strategies universally guarantees a zero-trace Hessian. Therefore, the sole remaining requirement to establish a strict saddle is that the symmetric Hessian does not perfectly vanish ($H \neq \mathbf{0}$). Because exact tightness requires highly coupled, overlapping resource cycles---such as the generalized families provided in~\cite{roughgarden2015intrinsic}---the presence of non-zero cross-derivatives is a generic algebraic condition of the network design. Consequently, we expect  network topologies constructed to perfectly bind the Price of Anarchy to structurally inherit this strict saddle geometry, rendering the topological instability of worst case NE a generic feature of tight bounding games.
\end{remark}

This topological instability follows a line of reasoning first explored in learning in congestion games in ~\cite{Kleinberg09multiplicativeupdates}. In that context, it was established that under the continuous-time limit of the Multiplicative Weights Update (replicator dynamics), players inherently avoid unstable mixed Nash equilibria, converging almost everywhere to weakly stable states, typically pure/strict Nash. By proving that the exact worst-case configurations of canonical PoA bounds are structurally mandated to be strict saddles (algebraically behave as mixed Nash), we effectively bridge these two narratives: the worst-case static bounds correspond precisely to the unstable fixed points that natural evolutionary dynamics are mathematically guaranteed to escape. Furthermore, this evasion of strict saddles is not merely a continuous-time artifact. Recent advances have generalized this phenomenon to a vast family of discrete-time algorithms. It is now well-established that standard first-order methods—including gradient descent even in the presence of non-isolated critical points, as well as mirror descent and block coordinate descent—almost always avoid strict saddles~\cite{panageas2016gradient,lee2019first}. Consequently, the algebraic necessity of the strict saddle ensures that practically any natural, decentralized learning algorithm will dynamically bypass the theoretical worst-case Price of Anarchy.

\subsection{Proof of Theorem~\ref{thm:Repeller_pure_NE}}
%Worst case pure Nash equilibria are global maxima of the potential in~\cite{christodoulou2005price,roughgarden2015intrinsic}}
\label{app:global_max}

\begin{reptheorem}{thm:Repeller_pure_NE}
Let $G$ be the congestion game constructed via the canonical smoothness arguments to yield the worst-case Price of Anarchy for affine cost functions in ~\cite{christodoulou2005price,roughgarden2015intrinsic} for $N=3$ agents. Let 
%$y^*$ be the socially optimal pure Nash equilibrium and 
 $y$ be the worst-case pure Nash equilibrium, then:
\begin{enumerate}
    \item The strategies supporting $y$ are almost everywhere strictly dominated in the measure-theoretic sense.
    \item The state $y$, alongside the manifold of states formed by unilateral deviations from it, constitutes the global maximum of the exact potential function $\Phi$.
    \item Under any gradient-like learning dynamic (where $\Phi$ acts as a strict Lyapunov function outside the equilibrium set), the worst-case equilibrium $y$ lies strictly in the global repeller of the interior state space.
\end{enumerate}
\end{reptheorem}

\begin{proof}
The proof proceeds by explicitly unpacking the canonical lower-bound construction that establishes the tightness of robust PoA bounds for affine latencies (formalized in \cite{roughgarden2015intrinsic}, Theorem 5.6, and \cite{christodoulou}, Theorem 2). We proceed in three stages: explicitly defining the worst-case game geometry, establishing the measure-theoretic dominance via partial derivatives, and deriving the topological consequences.

\vspace{0.2cm}
\noindent \textit{Step 1: The Geometry of the Worst-Case Double-Cycle Game} \\
To achieve the exact worst-case Price of Anarchy ratio of $5/2$ for affine cost functions $c(x) = ax + b$, the canonical construction requires a highly specific topological structure: the $N=3$ ``double-cycle'' congestion game. 

The game consists of exactly $3$ active players and a ground set of $6$ resources partitioned into two disjoint sets, $E_1 = \{h_1, h_2, h_3\}$ and $E_2 = \{g_1, g_2, g_3\}$. The resources in each set are arranged in a cyclic topology. Following the formal construction (e.g., Theorem 5.6 in Roughgarden [2015]), the resources in $E_1$ and $E_2$ are assigned the proportionally weighted cost functions $\eta \cdot c(x)$ and $(1-\eta) \cdot c(x)$, respectively.

Each player $i \in \{1, 2, 3\}$ possesses exactly two strategies: a suboptimal strategy $P_i$ and an optimal strategy $Q_i$. The optimal block sizes required to strictly bind the $5/2$ affine inequalities dictate the lengths of these contiguous segments:
\begin{itemize}
    \item Strategy $P_i$ consists of precisely $x_1 = 2$ consecutive resources in $E_1$ and $x_2 = 1$ resource in $E_2$, constructed cyclically: $P_i = \{g_{i+1}, h_{i-1}, h_{i+1}\}$ (evaluating indices modulo $3$).
    \item Strategy $Q_i$ consists of precisely $x_1^* = 1$ resource in $E_1$ and $x_2^* = 1$ resource in $E_2$, constructed as $Q_i = \{h_i, g_i\}$.
    \item By construction, for every player $i$, the strategies $P_i$ and $Q_i$ are strictly disjoint ($P_i \cap Q_i = \emptyset$), while $Q_i \cap Q_j = \emptyset$ for all $i \neq j$.
\end{itemize}

Let $y$ denote the worst-case pure Nash equilibrium where every player $i$ selects $P_i$, and let $y^*$ denote the socially optimal state where every player selects $Q_i$. 

In order to maximize the PoA ratio, the parameter $\eta$ must be algebraically tuned such that at the state $y$, every player is exactly indifferent between their two strategies ($C_i(y) = C_i(Q_i, y_{-i})$). Equating the expected costs under the base latency $c(x) = ax$ yields the exact balancing condition:
$$\eta a(x_1^2) + (1-\eta) a(x_2^2) = \eta a(x_1^*)(x_1+1) + (1-\eta) a(x_2^*)(x_2+1)$$
Substituting the exact block sizes ($x_1=2, x_2=1, x_1^*=1, x_2^*=1$) yields:
$$\eta a(4) + (1-\eta) a(1) = \eta a(3) + (1-\eta) a(2)$$
$$3\eta + 1 = \eta + 2 \implies 2\eta = 1 \implies \eta = 1/2$$

The indifference constraint mathematically mandates that $\eta = 1/2$. (see also Example 5.7 in  \cite{roughgarden2015intrinsic}). Consequently, the affine cost functions for both cycles possess identical slopes ($a/2$), ensuring uniform marginal congestion costs across the entire topology. To simplify notation moving forward with absorb the multiplicative constant $\eta = 1/2$ into the slope and intercept variables $a\leftarrow a/2, b\leftarrow b/2$.   
For generalized affine costs where the static intercept $b > 0$, this perfect indifference is preserved by strategically assigning non-strategic ``dummy players'' to specific resources to absorb the asymmetric static costs, locking the active players into the exact equilibrium:
$$C_i(y) = C_i(Q_i, y_{-i})$$
\vspace{0.2cm}
\noindent \textit{Step 2: Almost Everywhere Strict Dominance} \\
Let $q \in [0,1]^3$ represent the mixed strategy profile of the population, where $q_j$ denotes the probability that agent $j$ plays their optimal strategy $Q_j$ (and $1-q_j$ is the probability they play $P_j$). For any agent $i$, let $\Delta_i(q)$ denote the expected cost difference between their two strategies:
$$\Delta_i(q) = \mathbb{E}[C_i(P_i, q_{-i})] - \mathbb{E}[C_i(Q_i, q_{-i})]$$
By the strict indifference condition established in Step 1, when all opponents play the suboptimal strategy ($q_{-i} = \mathbf{0}$), we have $\Delta_i(\mathbf{0}) = 0$.

We evaluate the stability of this indifference by calculating the partial derivative of the cost gap with respect to an opponent's shifting probability $q_j$. By the linearity of expectation for affine costs $c(x) = ax + b$, the intercept $b$ and the constant dummy loads vanish under differentiation. The rate of change depends strictly on the slope $a > 0$ and the collision combinatorics between the sets:
$$\frac{\partial \mathbb{E}[C_i(P_i)]}{\partial q_j} = a \left( |P_i \cap Q_j| - |P_i \cap P_j| \right)$$
$$\frac{\partial \mathbb{E}[C_i(Q_i)]}{\partial q_j} = a \left( |Q_i \cap Q_j| - |Q_i \cap P_j| \right)$$

Evaluating Player $1$ against Opponent $2$ ($j=2$):
\begin{itemize}
    \item $P_1 \cap Q_2 = \{h_2, g_2\} \implies$ Size $2$. $P_1 \cap P_2 = \{h_3\} \implies$ Size $1$. Thus, $\frac{\partial \mathbb{E}[C_1(P_1)]}{\partial q_2} = a(2 - 1) = a$.
    \item $Q_1 \cap Q_2 = \emptyset \implies$ Size $0$. $Q_1 \cap P_2 = \{h_1\} \implies$ Size $1$. Thus, $\frac{\partial \mathbb{E}[C_1(Q_1)]}{\partial q_2} = a(0 - 1) = -a$.
    \item The gradient of the gap is: $\frac{\partial \Delta_1}{\partial q_2} = a - (-a) = 2a > 0$.
\end{itemize}

Evaluating Player $1$ against Opponent $3$ ($j=3$):
\begin{itemize}
    \item $P_1 \cap Q_3 = \{h_3\} \implies$ Size $1$. $P_1 \cap P_3 = \{h_2\} \implies$ Size $1$. Thus, $\frac{\partial \mathbb{E}[C_1(P_1)]}{\partial q_3} = a(1 - 1) = 0$.
    \item $Q_1 \cap Q_3 = \emptyset \implies$ Size $0$. $Q_1 \cap P_3 = \{h_1, g_1\} \implies$ Size $2$. Thus, $\frac{\partial \mathbb{E}[C_1(Q_1)]}{\partial q_3} = a(0 - 2) = -2a$.
    \item The gradient of the gap is: $\frac{\partial \Delta_1}{\partial q_3} = 0 - (-2a) = 2a > 0$.
\end{itemize}

Because both directional derivatives are strictly positive ($2a > 0$), any infinitesimal shift by any opponent away from $q = \mathbf{0}$ strictly increases the cost gap $\Delta_i(q)$, making $P_i$ strictly more expensive than $Q_i$. Consequently, for any profile $q_{-i} \neq \mathbf{0}$ (which constitutes the entire Lebesgue measure of the opponents' mixed strategy space), $Q_i$ strictly outperforms $P_i$.\footnote{As a reminder, the cost functions are linear in the load probabilities jointly so the gradient analysis is sufficient to judge the cost gaps for joint changes in $q_j$'s, i.e., all $q_{-i} \ne 0$.}  
This mathematically guarantees that the strategy $P_i$ supporting the worst-case pure NE is strictly dominated everywhere except at the exact, measure-zero coordinate of the equilibrium itself.

\vspace{0.2cm}
\noindent \textit{Step 3: Maximum Potential} \\
Every congestion game admits an exact potential function $\Phi(s)$ that precisely mirrors the change in cost for any unilateral deviator:
$$ \Phi(s) - \Phi(s_i', s_{-i}) = C_i(s) - C_i(s_i', s_{-i}) $$
Because $y$ is constructed as a weakly strict Nash equilibrium ($\Delta_i(\mathbf{0}) = 0$), a unilateral deviation by agent $i$ to $Q_i$ yields zero marginal change in cost. Substituting this into the potential function identity yields $\Phi(y) - \Phi(Q_i, y_{-i}) = 0$. This establishes that the potential is completely flat along the $k$ edges of the hypercube connecting $y$ to the states $(Q_i, y_{-i})$. 

Crucially, by the strict dominance established in Step 2, for any state $s$ not residing on these specific edges, unilaterally moving toward the suboptimal strategy $P_i$ strictly increases the agent's expected cost. Because the potential function tracks the agent's marginal cost perfectly, moving toward $y$ strictly increases the global potential $\Phi$. Therefore, the set containing $y$ and its immediately adjacent degenerate edges constitutes the strict global maximum of the potential landscape.

\vspace{0.2cm}
\noindent \textit{Step 4: The Global Repeller} \\
Consider any continuous-time, gradient-like learning dynamic (e.g., replicator dynamics) operating on the state space. In exact potential games, the potential function $\Phi$ serves as a strict Lyapunov function for these dynamics, satisfying $\frac{d}{dt}\Phi(p(t)) \le 0$, with equality holding strictly if and only if the state $p(t)$ is a stationary fixed point. 

Let $p(0) \in (0,1)^k$ be an arbitrary interior initial condition. Since the state space $[0,1]^k$ is compact and $\Phi$ is bounded below, the sequence $\Phi(p(t))$ must converge monotonically to some limit $\Phi^*$. Let $\Omega(p(0))$ denote the $\omega$-limit set of the trajectory (the set of all accumulation points). By standard topological arguments, $\Omega(p(0))$ is a non-empty, compact, and forward-invariant set. For any accumulation point $q \in \Omega(p(0))$, the continuity of $\Phi$ guarantees that $\Phi(q) = \lim_{t \to \infty} \Phi(p(t)) = \Phi^*$. 

Because $\Omega(p(0))$ is an invariant set, any trajectory initialized at a point $q \in \Omega(p(0))$ must remain entirely within $\Omega(p(0))$ for all future times. This implies that $\Phi$ must remain perfectly constant along any such trajectory. However, the strict Lyapunov property dictates that $\Phi$ strictly decreases everywhere except at equilibria. Therefore, $\Omega(p(0))$ must consist exclusively of equilibrium points.

We now evaluate the dynamical viability of the worst-case pure Nash equilibrium $y$. Because $y$ strictly maximizes the potential function across the interior (as proven in Step 3), we have $\Phi(p(0)) < \Phi(y)$ for any interior initialization. The strict monotonic decrease of the potential function mathematically guarantees that $\Phi(p(t)) \le \Phi(p(0)) < \Phi(y)$ for all $t \ge 0$. Consequently, the entire trajectory $p(t)$ is permanently trapped in a lower contour set of $\Phi$, strictly bounded away from $y$. The state $y$, therefore, cannot be an element of the limit set $\Omega(p(0))$ for any interior trajectory. It acts as a global topological repeller, mathematically forbidding the learning dynamics from ever approaching the worst-case efficiency bound.
\end{proof}

\subsection{Proof of Theorem~\ref{thm:inverse_gt}}

Prior to the formal proof, we rigorously define the genericity condition necessary to preclude degenerate collinearity in the agents' payoff spaces. By fixing the opponents' strategy profile $s_{-i} \in S_{-i}$, we can express the cost function of agent $i$ for a pure strategy $j$ as a vector $\mathbf{c}_{i,j} \in \mathbb{R}^{|S_{-i}|}$. 

\begin{definition}[Non-Collinear Genericity]
    Let $G$ be an $N$-player game where each agent possesses $m \ge 3$ pure strategies. Let $p \in \prod_i \Delta(S_i)$ be an interior Nash equilibrium. For each agent $i$ and pure strategy $j$, define the normalized cost vector $\mathbf{c}'_{i,j} \in \mathbb{R}^{|S_{-i}|}$ such that its coordinate for opponent profile $s_{-i}$ is given by $c'_i(j, s_{-i}) = c_i(j, s_{-i}) - c_i(p_i, s_{-i})$. The game $G$ satisfies \textit{non-collinear genericity} if, for every agent $i$:
    \begin{enumerate}
        \item[(i)] No normalized cost vector is identically zero ($\mathbf{c}'_{i,j} \neq \mathbf{0}$ for all $j$).
        \item[(ii)] The set of normalized cost vectors $\{\mathbf{c}'_{i,1}, \dots, \mathbf{c}'_{i,m}\}$ spans a subspace of $\mathbb{R}^{|S_{-i}|}$ of dimension strictly greater than $1$.
    \end{enumerate}
\end{definition}

\begin{reptheorem}{thm:inverse_gt}
Let $G$ be a generic game with $N$ agents and $m \ge 3$ strategies each, and let $p$ be an interior Nash equilibrium of $G$. The equivalence class $INV_G(p)$ of games for which $p$ is a Nash equilibrium includes at least $2^{Nm}$ distinct games, no two of which are game-theoretically equivalent to each other.
\end{reptheorem}

\begin{proof}
The proof proceeds constructively in four stages: (1) extracting the normalized base game, (2) generating a combinatorial explosion of $2^{Nm}$ derivative games, (3) verifying the strict preservation of the equilibrium, and (4) proving the mutually exclusive game-theoretic non-equivalence of the entire set.

\vspace{0.2cm}
\noindent \textit{Step 1: Normalization of the Base Game} \\
Given the base game $G$ with cost functions $c_i$, we construct a strategically equivalent normalized game $G'$ by applying a \textit{strategy-independent cost offset} corresponding to the expected cost  of the equilibrium strategy $p_i$ against the opponents' profile:
$$c'_i(j, s_{-i}) = c_i(j, s_{-i}) - c_i(p_i, s_{-i})$$
where $j \in \{1, \dots, m\}$ is a pure strategy and $c_i(p_i, s_{-i}) = \sum_{k=1}^m p_i(k) c_i(k, s_{-i})$. Because the subtracted term $c_i(p_i, s_{-i})$ is strictly independent of agent $i$'s chosen action $j$, $G'$ is strategically equivalent to $G$.

By construction, the expected cost of playing the mixed strategy $p_i$ against any $s_{-i}$ in $G'$ is identically zero:
$$\sum_{j=1}^m p_i(j) c'_i(j, s_{-i}) = 0$$
Since $p$ is an interior Nash equilibrium in $G$, it remains an interior equilibrium in $G'$. A fundamental property of interior equilibria is that all pure strategies in the support must yield the exact same expected cost. Therefore, the expected cost of any pure strategy $j$ evaluated against the equilibrium profile $p_{-i}$ must also be exactly zero:
$$\mathbb{E}_{s_{-i} \sim p_{-i}}[c'_i(j, s_{-i})] = 0 \quad \mbox{for all } i, j$$

\vspace{0.2cm}
\noindent \textit{Step 2: Constructing the $2^{Nm}$ Derivative Games} \\
We establish a family of combinatorial sign assignments. For each agent $i$ and each pure strategy $j$, we assign a scalar sign $\sigma_{i,j} \in \{-1, 1\}$. Across $N$ agents and $m$ strategies, this yields $2^{Nm}$ distinct sign matrices $\sigma \in \{-1, 1\}^{N \times m}$.

For each sign matrix $\sigma$, we define a unique game $G^\sigma$ with cost functions:
$$c^\sigma_i(j, s_{-i}) = \sigma_{i,j} c'_i(j, s_{-i})$$

\vspace{0.2cm}
\noindent \textit{Step 3: Verification of Equilibrium Preservation} \\
We now verify that $p$ remains an exact interior Nash equilibrium in every derived game $G^\sigma$. We calculate the expected cost of playing pure strategy $j$ against $p_{-i}$ in $G^\sigma$:
$$\mathbb{E}_{s_{-i} \sim p_{-i}}[c^\sigma_i(j, s_{-i})] = \mathbb{E}_{s_{-i} \sim p_{-i}}[\sigma_{i,j} c'_i(j, s_{-i})] = \sigma_{i,j} \mathbb{E}_{s_{-i} \sim p_{-i}}[c'_i(j, s_{-i})]$$
Because we established in Step 1 that $\mathbb{E}_{s_{-i} \sim p_{-i}}[c'_i(j, s_{-i})] = 0$, it immediately follows that:
$$\sigma_{i,j} \cdot 0 = 0$$
Because every pure strategy yields an expected cost of exactly zero against $p_{-i}$, no agent has a strictly profitable deviation. Thus, $p$ is a valid interior Nash equilibrium across all $2^{Nm}$ games.

\vspace{0.2cm}
\noindent \textit{Step 4: Proof of Strict Non-Equivalence} \\
We demonstrate that the constructed set of $2^{Nm}$ games consists entirely of pairwise game-theoretically non-equivalent instances, thereby establishing a strict exponential lower bound on the size of the equivalence class $INV_G(p)$. Assume, for the sake of contradiction, that there exist two distinct sign matrices, $\sigma \neq \tau$, such that $G^\sigma$ and $G^\tau$ are game-theoretically equivalent.

By definition of game-theoretic equivalence, this implies that for every agent $i$, there exists a positive scalar $a_i > 0$ and a \textit{strategy-independent offset} $d_i(s_{-i})$ such that for all $j$ and all $s_{-i}$:
$$c^\sigma_i(j, s_{-i}) = a_i c^\tau_i(j, s_{-i}) + d_i(s_{-i})$$
Substituting the definition of our derivative games from Step 2 yields:
$$\sigma_{i,j} c'_i(j, s_{-i}) = a_i \tau_{i,j} c'_i(j, s_{-i}) + d_i(s_{-i})$$
Let $\gamma_{i,j} = \sigma_{i,j} - a_i \tau_{i,j}$. We can rewrite the equation as:
$$\gamma_{i,j} c'_i(j, s_{-i}) = d_i(s_{-i})$$

Crucially, the right-hand side of this equation, $d_i(s_{-i})$, is entirely independent of the chosen strategy $j$. Formulated in $\mathbb{R}^{m^{N-1}}$, this requires that the scaled vector $\gamma_{i,j} \mathbf{c}'_{i,j}$ must be identical for all pure strategies $j$. Let us denote this common vector as $\mathbf{D}_i \in \mathbb{R}^{m^{N-1}}$.

We now face two mutually exclusive cases:

\begin{itemize}
    \item \textbf{Case 1: $\mathbf{D}_i \neq \mathbf{0}$ $(\mathbf{D}_i \neq[0, 0, \dots, 0])$.} If $\mathbf{D}_i$ is a strictly non-zero vector, then the scalar multiplier $\gamma_{i,j}$ must be non-zero for all $j$. If there existed even a single strategy $j$ where $\gamma_{i,j} = 0$, the scalar multiplication would yield $\gamma_{i,j} \mathbf{c}'_{i,j} = \mathbf{0}$, forcing the common vector $\mathbf{D}_i = \mathbf{0}$ and contradicting the premise of this case. Because $\gamma_{i,j} \neq 0$ for all $j$, we can divide by it to obtain $\mathbf{c}'_{i,j} = \frac{1}{\gamma_{i,j}} \mathbf{D}_i$. This implies that the normalized payoff vectors for all pure strategies $j$ are merely scalar multiples of the single vector $\mathbf{D}_i$. Consequently, the set $\{\mathbf{c}'_{i,1}, \dots, \mathbf{c}'_{i,m}\}$ spans a 1-dimensional subspace. However, by condition (ii) of Non-Collinear Genericity for $m \ge 3$, these vectors must span a subspace of dimension strictly greater than 1. This yields a strict contradiction.
    
    \item \textbf{Case 2: $\mathbf{D}_i = \mathbf{0}$.} In this case, we have $\gamma_{i,j} \mathbf{c}'_{i,j} = \mathbf{0}$ for all $j$. By condition (i) of Non-Collinear Genericity, the normalized cost vector $\mathbf{c}'_{i,j}$ is strictly non-zero. Therefore, the scalar multiplier must be zero: $\gamma_{i,j} = 0$ for all $j$.
\end{itemize}

Since Case 1 yields a contradiction, Case 2 must hold. Thus, for all $j$:
$$\gamma_{i,j} = 0 \implies \sigma_{i,j} = a_i \tau_{i,j}$$
We are given that $\sigma_{i,j}, \tau_{i,j} \in \{-1, 1\}$ and $a_i > 0$. The only way for a strictly positive scalar to map one sign element identically to another is if $a_i = 1$. Consequently, $\sigma_{i,j} = \tau_{i,j}$ for all $i, j$.

This implies the matrices are identical ($\sigma = \tau$), which directly contradicts our initial assumption that the games were derived from distinct sign assignments. Therefore, all $2^{Nm}$ games in $INV_G(p)$ are strictly distinct and game-theoretically non-equivalent.
\end{proof}

\subsection{Proof of Theorem~\ref{prop:zero_trace}}

\begin{reptheorem}{prop:zero_trace}
In any potential game where players randomize independently, let $q^*$ be a Nash equilibrium. Let $\mathcal{S}^*$ denote the restricted state space defined by the support of $q^*$. If $q^*$ is partially or fully mixed, and generic within its support (i.e., the restricted Hessian of the potential evaluated at $q^*$ is non-zero), then $q^*$ is a strict saddle point of the potential function restricted to $\mathcal{S}^*$. 
\end{reptheorem}

\begin{proof}
Let $\Phi(q)$ be the exact expected potential function of the game. Because the agents randomize independently, $\Phi(q)$ is formed via the standard multilinear extension over the players' mixed strategy spaces. 

Geometrically, if $q^*$ is partially mixed, it resides in the relative interior of a specific face of the global probability polytope. We restrict our analysis to this face, denoted $\mathcal{S}^*$, which corresponds to the subspace spanned by the strategies in the support of $q^*$. To rigorously evaluate the local geometry within $\mathcal{S}^*$, we formulate the restricted gradient $\nabla_{\mathcal{S}^*} \Phi$ and the restricted Hessian matrix $H_{\mathcal{S}^*} = \nabla^2_{\mathcal{S}^*} \Phi(q^*)$ using strictly independent variables. 

For an arbitrary player $i$ with $m \ge 2$ strategies in their support, let $x_1, \dots, x_{m-1}$ be their independent probability variables within $\mathcal{S}^*$, such that the final support strategy is constrained as $q_{i,m} = 1 - \sum_{k=1}^{m-1} x_k$. The expected potential restricted to $\mathcal{S}^*$ can be expressed as a linear combination of the conditional potentials when player $i$ plays a pure strategy within their support:
$$ \Phi(x, q_{-i}) = \sum_{k=1}^{m-1} x_k \Phi_{i,k}(q_{-i}) + \left(1 - \sum_{k=1}^{m-1} x_k\right) \Phi_{i,m}(q_{-i}) $$
where $\Phi_{i,k}(q_{-i})$ is the expected potential conditioned on player $i$ deterministically playing strategy $k$. Regrouping the terms yields:
$$ \Phi(x, q_{-i}) = \Phi_{i,m}(q_{-i}) + \sum_{k=1}^{m-1} x_k \left[ \Phi_{i,k}(q_{-i}) - \Phi_{i,m}(q_{-i}) \right] $$

First, we establish that $q^*$ is a critical point within the restricted subspace $\mathcal{S}^*$. Taking the first derivative with respect to an arbitrary independent variable $x_k$ yields the exact marginal expected potential difference:
$$ \frac{\partial \Phi}{\partial x_k} = \Phi_{i,k}(q_{-i}) - \Phi_{i,m}(q_{-i}) $$
Because $q^*$ is a Nash equilibrium, player $i$ must be exactly indifferent between all strategies actively played in their support. Consequently, the marginal potentials evaluate to exactly the same value: $\Phi_{i,k}(q_{-i}^*) = \Phi_{i,m}(q_{-i}^*)$. Therefore, $\frac{\partial \Phi}{\partial x_k} = 0$ for all independent variables defining $\mathcal{S}^*$, mathematically guaranteeing that $\nabla_{\mathcal{S}^*} \Phi(q^*) = \mathbf{0}$. Thus, $q^*$ is a critical point of the restricted potential.

Next, we evaluate the diagonal entries of the restricted Hessian $H_{\mathcal{S}^*}$ to classify this critical point. Taking the second partial derivative with respect to $x_k$ requires differentiating the marginal difference:
$$ [H_{\mathcal{S}^*}]_{x_k x_k} = \frac{\partial^2 \Phi}{\partial x_k^2} = \frac{\partial}{\partial x_k} \left[ \Phi_{i,k}(q_{-i}) - \Phi_{i,m}(q_{-i}) \right] $$
Crucially, $\Phi_{i,k}(q_{-i})$ and $\Phi_{i,m}(q_{-i})$ depend exclusively on the opposing players' distributions. They possess strictly zero dependence on player $i$'s own mixing probabilities. Consequently, the second derivative vanishes universally: 
$$ [H_{\mathcal{S}^*}]_{x_k x_k} = 0 $$

Because this property holds for every independent probability variable defining the support for every player, the main diagonal of the restricted Hessian $H_{\mathcal{S}^*}$ is entirely populated by zeros. This mandates that the trace of the restricted Hessian is equal to zero:
$$ \text{Trace}(H_{\mathcal{S}^*}) = \sum [H_{\mathcal{S}^*}]_{x_k x_k} = 0 $$

By the spectral theorem for symmetric matrices, the trace is equal to the sum of the eigenvalues. Since $q^*$ is generic within its support ($H_{\mathcal{S}^*} \neq \mathbf{0}$), it is  impossible for the non-zero eigenvalues to sum to zero unless the restricted spectrum contains at least one strictly negative eigenvalue and at least one strictly positive eigenvalue. Because $q^*$ is a restricted critical point possessing directions of both strict positive and negative curvature within its own support manifold, it  satisfies the definition of a strict saddle within $\mathcal{S}^*$. 
\end{proof}

\section{Proofs of Claims in Section 4}
\label{app:PoA_metrics}

This appendix provides the formal proofs establishing the structural collapse of both individual rationality baselines and macroscopic efficiency metrics (Price of Anarchy and Average Price of Anarchy) in congestion games.

\subsection{Proof of Theorem~\ref{thm:not_minmax}}

\begin{reptheorem}{thm:not_minmax}[Divergence from Baseline Rationality in Balls and Bins Games]
    Consider the class of canonical balls and bins games (symmetric load balancing games with $N$ agents, $M \ge 2$ identical resources, and monomial cost functions $c(x) = x^d$ for $d \ge 1$). The set of $(\lambda, \mu)$-approximate optimal states dictated by the robust Price of Anarchy framework strictly includes states where every agent simultaneously incurs a cost arbitrarily worse---scaling as $-\Theta(N^d)$---than their absolute (correlated) min-max safety guarantee.
\end{reptheorem}

\begin{proof}
    We start the theorem proof by explicitly constructing a smooth affine congestion game where the bounds guaranteed by robust Price of Anarchy are  permissive, mathematically classifying strictly dominated outcomes as ``approximately optimal.''

    \vspace{0.2cm}
    \noindent \textit{Step 1: The Centralized Optimum and Min-Max Baseline} \\
    Consider a symmetric congestion game with $N$ agents (assume $N$ is even for simplicity) and $M=2$ identical resources. Let the cost function for each resource be the standard identity affine function, $c_e(x) = x$. 

    The centralized social optimum $s^*$ is achieved when the agents distribute themselves equally across both resources, inducing a load of $N/2$ on each. The corresponding optimal social cost is $C(s^*) = 2(N/2)^2 = N^2/2$.

    We compute the baseline individual rationality threshold (the min-max safety cost) for an arbitrary agent $i$. The adversarial coalition of $N-1$ opponents maximizes agent $i$'s minimum expected cost by uniformly distributing their mass across both resources, inducing an expected load of $(N-1)/2$ on each. Agent $i$'s best response to this symmetric environment yields an expected cost of exactly $1 + (N-1)/2 = (N+1)/2$. Thus, any rationalizable state must guarantee agent $i$ an expected cost no greater than $(N+1)/2$.

    \vspace{0.2cm}
    \noindent \textit{Step 2: Evaluation of the Robust PoA Bound} \\
    We evaluate the system under the robust Price of Anarchy framework. For affine congestion games, the canonical smoothness parameters are $(\lambda=5/3, \mu=1/3)$, yielding a robust PoA bound of $\frac{\lambda}{1-\mu} = 5/2$. Consequently, the worst-case analysis framework guarantees, and inherently accepts as approximately optimal, any state distribution $\sigma$ whose expected social cost satisfies $\mathbb{E}[C(\sigma)] \le \frac{5}{2} C(s^*) = \frac{5}{4}N^2$.

    \vspace{0.2cm}
    \noindent \textit{Step 3: The Inclusion of Perfect Miscoordination} \\
    We identify a problematic state within this sanctioned set. Let $s_{\text{crowd}}$ be the state of perfect miscoordination where all $N$ agents select the exact same resource. The social cost of this state is $C(s_{\text{crowd}}) = N^2$. Because $N^2 \le \frac{5}{4}N^2$, this state explicitly satisfies the smoothness inequality and is mathematically classified as a $(\lambda, \mu)$-approximate optimal state.

    However, at $s_{\text{crowd}}$, every single agent experiences a realized cost of $N$. Evaluating this against the individual rationality baseline yields a strict violation of the safety guarantee:
    $$ u_i^{\text{MinMaxCenter}}(s_{\text{crowd}}) = \frac{N+1}{2} - N = -\frac{N-1}{2} = -\Theta(N) $$
    Thus, the set of approximately optimal states defined by robust PoA guarantees structurally includes outcomes where every single agent suffers an efficiency loss arbitrarily worse than their absolute worst-case adversarial safety guarantee. 

\medskip
The structural slackness identified % in Theorem~\ref{thm:not_minmax} 
thus far  
is not an isolated artifact of affine costs; rather, it amplifies exponentially with the non-linearity of the game. We formalize this divergence below.

% \begin{proposition}[Amplification of Structural Slackness]
% \label{prop:polynomial_amplification}
%     In smooth congestion games with polynomial cost functions of degree $d \ge 1$, the set of $(\lambda, \mu)$-approximate optimal states includes outcomes where the violation of the min-max individual rationality baseline scales as $-\Theta(N^d)$.
% \end{proposition}

% \begin{proof}

We extend the construction above 
%of Theorem~\ref{thm:not_minmax} 
to a symmetric congestion game with $N$ agents and $M=2$ identical resources, where the latency function is given by the monomial $c(x) = x^d$.

    The centralized social optimum $s^*$ again equidistributes the agents, yielding a social cost of $C(s^*) = 2(N/2)^{d+1} = N^{d+1}/2^d$.

    Next, we compute the exact min-max safety cost for an arbitrary agent $i$. The coalition of $N-1$ opponents seeks to choose a joint distribution over the resources to maximize the minimum expected cost of agent $i$. By the strict convexity of $x^d$ for $d > 1$, the variance of the load strictly increases the expected cost. Therefore, the absolute worst-case adversarial strategy is to perfectly correlate and assign all $N-1$ opponents to a single resource, chosen uniformly at random. Agent $i$'s best response to this adversarial environment is to select either resource, encountering $N-1$ opponents with probability $1/2$ and $0$ opponents with probability $1/2$. This establishes the exact min-max safety cost:
    $$ c_i(\text{minmax}^*_i, s_{-i}) = \frac{1}{2}(N^d) + \frac{1}{2}(1^d) = \frac{N^d+1}{2} $$

    We now evaluate the state of perfect miscoordination, $s_{\text{crowd}}$, where all $N$ agents select the exact same resource. The social cost of this state is $C(s_{\text{crowd}}) = N^{d+1}$. The ratio of this cost to the social optimum is exactly $C(s_{\text{crowd}}) / C(s^*) = 2^d$. 

    It is a standard result in the Price of Anarchy literature that the robust PoA bound $\Phi_d$ for atomic congestion games with polynomial costs of degree $d$ scales as $\Phi_d = \Theta\left(\frac{d}{\ln d}\right)^d$ \cite{christodoulou2005price, aland2011exact}. Because $\Phi_d$ strictly dominates $2^d$ for all $d \ge 1$ (e.g., $\Phi_1 = 2.5 > 2$, $\Phi_2 \approx 9.6 > 4$), the state $s_{\text{crowd}}$ strictly satisfies the robust PoA inequality:
    $$ \mathbb{E}[C(s_{\text{crowd}})] \le \Phi_d C(s^*) $$
    Consequently, $s_{\text{crowd}}$ is mathematically validated by the worst-case framework as a $(\lambda, \mu)$-approximate optimal state. Yet, evaluating the individual performance of the agents at this state yields:
    $$ u_i^{\text{MinMaxCenter}}(s_{\text{crowd}}) = \frac{N^d+1}{2} - N^d = -\frac{N^d - 1}{2} = -\Theta(N^d) $$
    Thus, as the degree of the cost functions increases, the analytical bounds provided by robust PoA become exponentially decoupled from the threshold of baseline individual rationality.
\end{proof}

\subsection{Proof of Theorem~\ref{lem:poisson_binomial_minmax}}

To rigorously evaluate the min-max safety cost subject to independent randomization (as claimed in Section~\ref{sec:minmax_regret}), we must first formalize the exact probabilistic nature of the load induced by an adversarial coalition.

\begin{lemma}[Poisson Binomial Opponent Load]
\label{lem:opponent_load_distribution}
    Let $\pi_{-i} \in \prod_{j \neq i} \Delta(M)$ be an arbitrary mixed strategy profile for the $N-1$ opponents of agent $i$ in an $M$-bin congestion game. For any fixed bin $e \in [M]$, let $X_e$ denote the total load induced on bin $e$ by the opponents. Then $X_e$ follows a Poisson Binomial distribution parameterized by $\mathbf{p}_e = (p_{j,e})_{j \neq i}$, where $p_{j,e}$ is the probability that agent $j$ selects bin $e$.
\end{lemma}
\begin{proof}
    By the definition of a mixed strategy profile in a normal-form game, the agents randomize their actions independently. Let $s_j \in [M]$ denote the pure strategy realized by agent $j$. The joint probability measure over the opponents' pure action profiles $s_{-i} \in [M]^{N-1}$ is the exact product measure: $\mathbb{P}[s_{-i}] = \prod_{j \neq i} p_{j,s_j}$.
    
    For each opponent $j \neq i$, define the indicator random variable $B_{j,e} = \mathbb{I}\{s_j = e\}$. Under the product measure, the sequence of variables $\{B_{j,e}\}_{j \neq i}$ is mutually independent. The expected value of each indicator is $\mathbb{E}[B_{j,e}] = \mathbb{P}[s_j = e] = p_{j,e}$. Thus, $B_{j,e} \sim \text{Bernoulli}(p_{j,e})$.
    
    The total opponent load on bin $e$ is the algebraic sum of these indicators: $X_e = \sum_{j \neq i} B_{j,e}$. Because $X_e$ is the sum of $N-1$ mutually independent Bernoulli random variables, it strictly satisfies the definition of a Poisson Binomial random variable. Thus, $X_e \sim \text{PB}(\mathbf{p}_e)$.
\end{proof}

\begin{reptheorem}{lem:poisson_binomial_minmax}[Independent Min-Max Equivalence]
    In symmetric balls and bins games with $N$ agents, $M$ bins, and convex monomial costs $c(x) = x^d$ ($d \ge 1$), the expected cost evaluated at the fully mixed interior Nash equilibrium is equal to the independent min-max safety cost.
\end{reptheorem}
\begin{proof}
    Consider an agent $i$ facing $N-1$ opponents. The adversarial coalition selects independent mixed strategies $\mathbf{p}_j \in \Delta(M)$ for each opponent $j \neq i$ to maximize agent $i$'s minimum expected cost across all bins. The adversary's objective is to maximize:
    $$ \min_{e \in [M]} C_e(\mathbf{p}) \quad \text{where} \quad C_e(\mathbf{p}) = \mathbb{E}\left[\left(1 + X_e\right)^d\right] $$
    
    For any fixed bin $e$, let $\mu_e = \mathbb{E}[X_e] = \sum_{j \neq i} p_{j,e}$ denote the expected opponent load. Because each of the $N-1$ opponents must allocate exactly $1$ unit of probability mass across the $M$ bins, the sum of the expected loads is strictly conserved: $\sum_{e=1}^M \mu_e = N-1$. By the Pigeonhole Principle, the minimum expected load across the bins is strictly bounded: $\min_{e \in [M]} \mu_e \le \frac{N-1}{M}$.

    We now evaluate the adversary's objective by bifurcating the analysis based on the degree of the latency function.

    \textbf{Case 1: Linear Costs ($d = 1$).}
    By the linearity of expectation, the variance of the load distribution is strictly irrelevant, yielding $C_e(\mathbf{p}) = \mathbb{E}[1 + X_e] = 1 + \mu_e$. The adversary's objective simplifies exactly to maximizing $1 + \min_{e \in [M]} \mu_e$. Applying the bounding inequality established above, the maximum safety cost the adversary can inflict is strictly bounded above by $1 + \frac{N-1}{M}$.

    \textbf{Case 2: Strictly Convex Monomial Costs ($d \ge 2$).}
    By Lemma~\ref{lem:opponent_load_distribution}, the load $X_e \sim \text{PB}(\mathbf{p}_e)$. By Hoeffding's Theorem (1956) on the extrema of expected values of convex functions of independent Bernoulli sums, the expectation $C_e(\mathbf{p})$ is strictly Schur-concave with respect to the probability vector $\mathbf{p}_e$. Consequently, for any fixed mean $\mu_e$, the expectation $C_e(\mathbf{p})$ is strictly upper-bounded by evaluating the expectation under a homogeneous Binomial distribution with the identical mean:
    $$ C_e(\mathbf{p}) \le \mathbb{E}\left[(1 + Y(\mu_e))^d\right] := g(\mu_e) $$
    where $Y(\mu_e) \sim \text{Binomial}\left(N-1, \frac{\mu_e}{N-1}\right)$. 

    To evaluate the adversary's objective, we establish the monotonic behavior of the expected cost. For any two probabilities $p < p'$, the distribution $\text{Binomial}(N-1, p')$ strictly first-order stochastically dominates (FOSD) the distribution $\text{Binomial}(N-1, p)$. Because the latency function $(1+x)^d$ is strictly monotonically increasing for $x \ge 0$, its expectation under the Binomial distribution strictly increases with the success probability. Consequently, the composite function $g(\mu_e)$ is strictly monotonically increasing with respect to the mean $\mu_e$.

    We bound the adversary's objective using a direct chain of inequalities. Because $g$ is strictly increasing:
    $$ \min_{e \in [M]} C_e(\mathbf{p}) \le \min_{e \in [M]} g(\mu_e) = g\left(\min_{e \in [M]} \mu_e\right) \le g\left(\frac{N-1}{M}\right) $$
    
    The upper bound $g\left(\frac{N-1}{M}\right)$ represents the expected cost when the opponent load on the target bin is distributed according to $\text{Binomial}(N-1, 1/M)$. By the conditions of equality in Hoeffding's Theorem, this exact maximum is achieved if and only if the adversary sets the success probabilities to be perfectly homogeneous: $p_{j,e} = 1/M$ for all $j, e$.

    Thus, in both regimes, the absolute maximum of the adversary's objective is perfectly achieved by the homogeneous profile $p_{j,e} = 1/M$ for all $j$ and $e$. We note that this specific uniform profile is identically the fully mixed interior Nash equilibrium of the game.\footnote{For the linear case ($d=1$), this uniform profile is easily verified as the uniquely possible fully mixed equilibrium: an agent $i$ is fully mixed if and only if they are indifferent among all $M$ bins, forcing the expected opponent loads to be identical ($\mu_e^{(i)} = \frac{N-1}{M}$). Substituting $\mu_e^{(i)} = T_e - p_{i,e}$ mathematically forces $p_{i,e}$ to be identical across all agents, yielding the unique solution $p_{i,e} = 1/M$. For $d \ge 2$, the uniqueness of this profile as the min-max strategy is directly guaranteed by the strict equality conditions of Hoeffding's Theorem established in Case 2.} Therefore, the adversary's optimal independent min-max strategy mathematically collapses to the interior Nash equilibrium, proving the exact equivalence for all $d \ge 1$.
\end{proof}

\subsection{Proof of Theorem~\ref{thm:congestion_not_smooth}}

\begin{reptheorem}{thm:congestion_not_smooth}
The Price of Anarchy of affine congestion games is unbounded even when restricting to affine costs with arbitrarily small positive slopes. Moreover, the Pure Price of Anarchy of general affine congestion games is unbounded. This holds even for minimal games comprising exactly two agents, two strategies, and four congested elements, where all congested elements possess positive, non-decreasing affine latency functions on the active domain.
\end{reptheorem}

\begin{proof}
    We establish the theorem constructively by mapping the two minimal normal-form matrices presented in Section~\ref{sec:poa_fragility} to exact, physically realizable affine congestion games.

    \vspace{0.2cm}
    \noindent \textit{Step 1: Unbounded Pure Price of Anarchy} \\
    To realize the first payoff matrix, we define a congestion game with $N=2$ agents and a ground set of four resources: $E = \{e_{AA}, e_{AB}, e_{BA}, e_{BB}\}$. The strategy set for the row agent (Agent 1) is $S_1 = \{\{e_{AA}, e_{AB}\}, \{e_{BA}, e_{BB}\}\}$. The strategy set for the column agent (Agent 2) is $S_2 = \{\{e_{AA}, e_{BA}\}, \{e_{AB}, e_{BB}\}\}$. 

    We define the affine latency functions $c(x) = ax + b$ such that the costs at loads $x \in \{1, 2\}$ are given by:
    \begin{align*}
        c_{AA}(1) &= c_{AB}(1) = c_{BA}(1) = c_{BB}(1) = \delta \\
        c_{AA}(2) &= w + 1 - \delta \\
        c_{AB}(2) &= w - \delta \\
        c_{BA}(2) &= 1 - \delta \\
        c_{BB}(2) &= w + 1 - \delta
    \end{align*}
    To ensure all latency functions are strictly monotonically increasing on the physical domain ($x \ge 1$), we require $c_e(2) > c_e(1) > 0$ for all $e$. Solving this system of inequalities yields the single strict parameter condition: $0 < 2\delta < 1 < w$. Provided this condition holds, all elements exhibit positive, strictly increasing affine latency functions. 
    
    By mapping the agents' action profiles to the resulting congestion on these elements, we exactly recover the first normal-form matrix. The optimal pure Nash equilibrium resides at $(B,A)$ with a social cost of exactly $2$. The worst-case pure Nash equilibrium resides at $(A,B)$ with a social cost of exactly $2w$. Consequently, the Pure Price of Anarchy is exactly $w$. Since $w$ can be set arbitrarily large, the Pure PoA is strictly unbounded, irrespective of the non-decreasing physical properties of the latencies.

    \vspace{0.2cm}
    \noindent \textit{Step 2: Unbounded Mixed Price of Anarchy (Small Slopes)} \\
    To realize the second payoff matrix, we define a completely symmetric congestion game with $N=2$ agents and exactly two resources: $E = \{e_1, e_2\}$. Both agents share the identical strategy set $S_1 = S_2 = \{\{e_1\}, \{e_2\}\}$. 

    Both resources share an identical affine cost function defined by:
    $$ c(x) = (\epsilon - \epsilon^2)x - \epsilon + 2\epsilon^2 $$
    for some parameter $0 < \epsilon < 1/2$. Because $\epsilon < 1/2$, the slope $(\epsilon - \epsilon^2)$ is strictly positive, ensuring the latency function is strictly monotonically increasing. Evaluated at integer loads, the costs are $c(1) = \epsilon^2$ and $c(2) = \epsilon$. Because $\epsilon > \epsilon^2 > 0$, the costs are strictly positive everywhere on the active domain. Notably, the intercept $(- \epsilon + 2\epsilon^2)$ is strictly negative.

    This maps perfectly to the second normal-form matrix. The pure Nash equilibria are the coordinated states $(A,B)$ and $(B,A)$, yielding the optimal social cost of $2\epsilon^2$. However, the fully mixed interior Nash equilibrium (where both agents randomize uniformly) yields an expected social cost of $\epsilon + \epsilon^2$. The Mixed Price of Anarchy is therefore exactly $\frac{\epsilon+ \epsilon^2}{2\epsilon^2} = \frac{1+ \epsilon}{2\epsilon}$. 
    
    By taking the limit as $\epsilon \to 0^+$, the positive slope $(\epsilon - \epsilon^2)$ becomes arbitrarily small (approaching zero), yet the Mixed Price of Anarchy diverges to $+\infty$. This mathematically proves that classical PoA bounds collapse entirely when negative intercepts are permitted, even for arbitrarily flat, everywhere-positive latency functions.
\end{proof}

\subsection{Proof of Theorem~\ref{thm:ergodic_good}}

\begin{reptheorem}{thm:ergodic_good}
    For the first parametric family of affine congestion games yielding unbounded PoA (Figure~\ref{fig:payoff_matrices}), under a uniform Lebesgue measure prior, the expected limit social cost is at most $2$ times the optimal social cost, with instances reaching an average-case inefficiency of at least $1+\frac{1}{27}(18-2\sqrt{3}\pi) \approx 1.2636$. For the second family, given any prior absolutely continuous with respect to the Lebesgue measure, the average-case performance strictly equals the optimal social cost. 

    Conversely, by linearly translating the costs of the first family of games, we can construct affine congestion games with two agents and two strategies where, under a uniform Lebesgue prior, the average-case system performance becomes arbitrarily worse than the optimal social welfare. 
\end{reptheorem}

\begin{proof}
    We establish the exact Average Price of Anarchy by computing the Lebesgue measure of the regions of attraction for the pure Nash equilibria. We proceed by establishing a universal geometric upper bound for all $w$, followed by an exact measure-theoretic integration for $w=2$.

    \vspace{0.2cm}
    \noindent \textit{Step 1: Volume-Preserving Diffeomorphism} \\
    Let $x$ and $y$ denote the probability that the row and column agents, respectively, play their first strategy in the $2 \times 2$ affine congestion game (Family 1 from Figure~\ref{fig:payoff_matrices}). The system admits an interior Nash equilibrium at $(x, y) = (\frac{w}{w+1}, \frac{1}{w+1})$. 
    
    To symmetrically align the phase space, we apply the change of variables $x' = x$ and $y' = 1-y$. The Jacobian determinant of this transformation is identically $-1$, ensuring the mapping is a volume-preserving diffeomorphism. The Lebesgue measure of any region of attraction in the transformed space is strictly identical to its measure in the original game. This transformation maps the congestion game into a classic $w$-coordination game, possessing two strict pure Nash equilibria at $(0,0)$ and $(1,1)$, and a symmetric interior saddle at $(\frac{w}{w+1}, \frac{w}{w+1})$. 

    \vspace{0.2cm}
    \noindent \textit{Step 2: Universal Upper Bound via Forward-Invariant Triangles} \\
    To establish the universal upper bound of $2$ on the Average Price of Anarchy for this family, we reproduce a geometric bounding argument originally developed by Panageas and Piliouras~\cite{panageas2016average}. We include this analysis to ensure our measure-theoretic evaluation remains rigorously self-contained.

    In the transformed coordinates, the socially optimal pure Nash equilibrium resides at $(0,0)$ with a normalized social cost of $1$, while the suboptimal pure Nash equilibrium resides at $(1,1)$ with a cost of $w$. Consider the two geometric triangles formed by the vertices $\{(0,0), (1,0), (\frac{w}{w+1}, \frac{w}{w+1})\}$ and $\{(0,0), (0,1), (\frac{w}{w+1}, \frac{w}{w+1})\}$. Evaluation of the replicator vector field along the boundaries of these triangles confirms they are strictly forward-invariant. Because the interior Nash equilibrium is a strict saddle (and thus its stable manifold has Lebesgue measure zero), the interior of these triangles resides entirely within the basin of attraction of the optimal equilibrium $(0,0)$.

    The combined Lebesgue measure (area) of these two triangles is exactly $\frac{w}{w+1}$. By pessimistically assigning the maximum possible cost $w$ to the entire remaining state space $1 - \frac{w}{w+1} = \frac{1}{w+1}$, we obtain a strict upper bound on the expected social cost under a uniform prior:
    $$ \text{APoA} \le \left(\frac{w}{w+1} \times 1\right) + \left(\frac{1}{w+1} \times w\right) = \frac{2w}{w+1} $$
    For all $w \ge 1$, we have $\frac{2w}{w+1} < 2$. This formally guarantees that the Average PoA is strictly bounded by 2 under the standard algebraic constraints.

    \vspace{0.2cm}
    \noindent \textit{Step 3: The Exact Invariant Manifold for $w=2$} \\
    To prove the lower bound on inefficiency, we evaluate the specific instance where $w=2$. The continuous-time replicator equations in the transformed coordinates are given by $\dot{x} = x(1-x)(3y - 2)$ and $\dot{y} = y(1-y)(3x - 2)$. By dividing the differential equations, the time dependency is strictly eliminated, yielding the separable ordinary differential equation:
    $$ \frac{3x-2}{x(1-x)} dx = \frac{3y-2}{y(1-y)} dy $$
    Integrating both sides by partial fractions explicitly constructs the time-invariant function $I(x,y)$ that remains strictly constant along any trajectory:
    $$ I(x,y) = -2\ln(x) - \ln(1-x) + 2\ln(y) + \ln(1-y) = C $$
    The stable manifold (the separatrix) of the interior saddle at $(2/3, 2/3)$ must satisfy $I(2/3, 2/3) = 0$. Consequently, the exact algebraic equation defining the separatrix that partitions the state space is $x^2(1-x) = y^2(1-y)$. Solving this algebraic curve for $y$ yields the closed-form boundary of the basin of attraction:
    $$ y(x) = \frac{1}{2}\left(1 - x + \sqrt{1 + 2x - 3x^2}\right) $$

    \begin{figure}[ht!]
        \centering
        \includegraphics[width=0.55\textwidth]{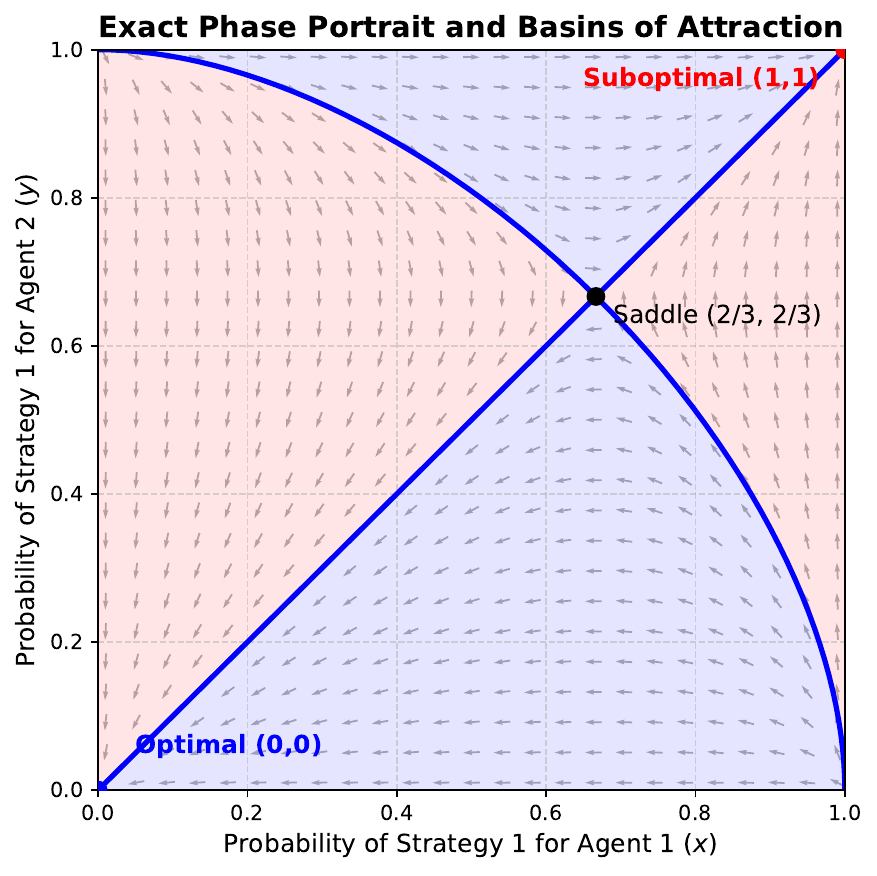}
        \caption{\textbf{Exact Measure-Theoretic Basins of Attraction.} The simulated phase portrait of the replicator vector field after the volume-preserving transformation. The blue curve represents the mathematically exact invariant stable manifold $x^2(1-x) = y^2(1-y)$ defining the separatrix. Because the saddle point is positioned at $(2/3, 2/3)$, it severely constricts the basin of the high-potential (worst-case) equilibrium at $(1,1)$. Trajectories originating below the separatrix converge to the socially optimal equilibrium at $(0,0)$. The forward-invariant triangles described in Step 2 are strictly contained within the optimal basin.}
        \label{fig:separatrix_apoa}
    \end{figure}

    \vspace{0.2cm}
    \noindent \textit{Step 4: Measure-Theoretic Integration} \\
    Under a uniform Lebesgue measure prior over initial conditions, the probability of converging to the socially optimal equilibrium $(0,0)$ is strictly equal to the area under the separatrix $y(x)$. Integrating this closed-form expression over the unit interval evaluates exactly to:
    $$ \int_0^1 \frac{1}{2}\left(1 - x + \sqrt{1 + 2x - 3x^2}\right) dx = \frac{1}{27}\left(9 + 2\sqrt{3}\pi\right) \approx 0.7364 $$
    Consequently, the region of attraction for the worst-case, high-potential Nash equilibrium at $(1,1)$ has a Lebesgue measure of $1 - 0.7364 = 0.2636$. 

    \vspace{0.2cm}
    \noindent \textit{Step 5: Average Price of Anarchy Evaluation} \\
    Let the optimal social cost at $(0,0)$ be normalized to $1$. In this configuration, the social cost of the worst-case Nash equilibrium at $(1,1)$ evaluates to $w = 2$. The Average Price of Anarchy is the expected social cost evaluated over the Lebesgue measure of the respective attractors:
    $$ \text{APoA} = (0.7364 \times 1) + (0.2636 \times 2) = 1.2636 $$
    The strict analytical evaluation of the integral guarantees this average-case inefficiency bound is robust and exact.
    
    For the second parametric family in Figure~\ref{fig:payoff_matrices}, an identical analysis reveals that the mixed Nash equilibrium possesses a region of attraction of exactly Lebesgue measure zero. Consequently, almost all initial conditions converge to the optimal pure Nash equilibrium, yielding an APoA of strictly $1$.

    \vspace{0.2cm}
    \noindent \textit{Step 6: The Collapse of APoA under Data-Driven Relaxations} \\
    Finally, we establish that Average PoA is structurally unbounded when the non-negativity constraint on affine intercepts is relaxed. Consider the first family of games, and subtract a uniform constant $1-\epsilon$ from all payoff entries. 

    To formally execute the uniform translation of all payoff entries by $1-\epsilon$, we apply a uniform constant shift to the underlying affine latency functions. Because every strategy in the constructed congestion game utilizes exactly two resources, subtracting a constant $C = \frac{1-\epsilon}{2}$ from each resource cost exactly subtracts $1-\epsilon$ from the total cost of any strategy profile. We define the translated resource costs as $c_e'(x) = c_e(x) - \frac{1-\epsilon}{2}$. 

    To guarantee that the resulting latency functions remain strictly positive and strictly monotonically increasing on the physical domain ($x \ge 1$), we constrain the free parameter $\delta$. By choosing $\delta = \frac{1}{2} - \frac{\epsilon}{4}$, we satisfy both requirements simultaneously: the base cost evaluates to $c_e'(1) = \delta - C = \frac{\epsilon}{4} > 0$, ensuring strict positivity, and the minimum marginal cost increment evaluates to $c_{BA}'(2) - c_{BA}'(1) = 1 - 2\delta = \frac{\epsilon}{2} > 0$, preserving strict monotonicity. Consequently, this transformation yields a mathematically valid, everywhere-positive affine congestion game.

    Crucially, the continuous-time replicator vector field is invariant to uniform translations of the payoff matrix. Because the vector field is identical, the exact invariant manifold $I(x,y)$, the separatrix $y(x)$, and the resulting Lebesgue measures of the basins of attraction ($0.7364$ and $0.2636$) remain perfectly unchanged. 
    
    However, the social cost metric is strictly non-invariant to this translation. After the translation, the cost of the socially optimal state approaches $\epsilon$, while the cost of the suboptimal state approaches $1+\epsilon$. As $\epsilon \to 0^+$, the ratio of the suboptimal cost to the optimal cost diverges to $+\infty$. Because the suboptimal state retains a strictly positive basin of attraction ($0.2636$), the expected social cost evaluated over the Lebesgue prior diverges to $+\infty$. 
\end{proof}

\section{Proof of Theorem~\ref{thm:not_rationalizable}}
\label{app:CCE}

To formally establish Theorem~\ref{thm:not_rationalizable}, we prove two fundamental structural lemmas. The first establishes the existence of non-rationalizable Coarse Correlated Equilibria (CCEs) fully supported on strictly dominated strategies in smooth games. The second establishes that strongly uncoupled learning dynamics cannot prevent convergence to these arbitrary target CCEs. 

\begin{lemma}[Non-Rationalizable CCEs in Smooth Games]
\label{lem:app_unnatural_cce}
    There exist smooth games (e.g., linear congestion games and valid utility games) possessing strong coarse correlated equilibria (CCE) where all agents assign strictly positive probability exclusively to strictly dominated strategies. Furthermore, there exist extreme CCEs (vertices of the feasible CCE polytope) exhibiting this exact property.
\end{lemma}
\begin{proof}
    We prove this lemma constructively for both canonical classes of smooth games.

    \vspace{0.2cm}
    \noindent \textit{Step 1: Construction in Linear Congestion Games} \\
    We define a symmetric congestion game with two agents and four strategies per agent. The game comprises three distinct resources. Resources $e_1$ and $e_2$ have identical affine latency functions $c(x) = x$. The third resource, $e_d$, acts as a fixed-cost penalty with $c(x) = \epsilon$, where $0 < \epsilon < 1/2$. The strategy set for both agents is $S = \{ \{e_1\}, \{e_1, e_d\}, \{e_2\}, \{e_2, e_d\} \}$. 

    By construction, strategy $\{e_1, e_d\}$ is strictly dominated by $\{e_1\}$, and $\{e_2, e_d\}$ is strictly dominated by $\{e_2\}$, as the inclusion of $e_d$ incurs a strictly positive additional cost regardless of the opponent's action. We propose the correlated distribution $\pi$ that assigns $1/2$ probability to the outcome $(\{e_1, e_d\}, \{e_2, e_d\})$ and $1/2$ probability to $(\{e_2, e_d\}, \{e_1, e_d\})$. Under $\pi$, the expected cost for each agent is exactly $1 + \epsilon$. 

    We evaluate the CCE inequalities. The optimal unconditional deviation for either agent is to permanently play an undominated strategy, e.g., $\{e_1\}$. Against the marginal distribution of the opponent (which places $1/2$ probability on utilizing $e_1$), the expected cost of this deviation is $\frac{1}{2}(1) + \frac{1}{2}(2) = 1.5$. Since we constrained $\epsilon < 1/2$, it holds that $1 + \epsilon < 1.5$. Thus, any unilateral deviation strictly increases the expected cost, mathematically confirming that $\pi$ is a strong CCE supported entirely on strictly dominated strategies.

    \vspace{0.2cm}
    \noindent \textit{Step 2: Construction in Valid Utility Games} \\
    To eliminate the ambiguity inherent in the inequality-based definition of valid utility games, we construct an exact \textit{basic utility system}~\cite{vetta2002nash}. We define a ground set of six independent elements $E = \{a_1, a_2, b_1, b_2, c_1, c_2\}$. We define a modular (and thus submodular) social welfare function $W(s) = \sum_{e \in U(s)} v(e)$, with values $v(a_1) = v(a_2) = v(b_1) = v(b_2) = \epsilon$, and $v(c_1) = v(c_2) = 1$, where $0 < \epsilon < 1/2$.
    
    The strategy set for agent 1 is $S_1 = \{\{a_1, c_1\}, \{c_1\}, \{a_2, c_2\}, \{c_2\}\}$, and for agent 2 is $S_2 = \{\{b_1, c_1\}, \{c_1\}, \{b_2, c_2\}, \{c_2\}\}$. In a basic utility system, the private utility of agent $i$ is defined exactly as their marginal contribution to the social welfare: $u_i(s) = W(s) - W(\emptyset, s_{-i})$. Because $W$ is modular, this evaluates exactly to the sum of the values of the elements agent $i$ \textit{exclusively} secures. This exact specification natively satisfies both defining inequalities of a valid utility game.

    By construction, the singleton strategy $\{c_1\}$ is strictly dominated by $\{a_1, c_1\}$ for agent 1. Because the element $a_1$ is uniquely available to agent 1, agent 1 secures $a_1$ exclusively whenever it is played. Thus, playing $\{a_1, c_1\}$ guarantees a payoff exactly $\epsilon$ strictly greater than playing $\{c_1\}$, regardless of agent 2's strategy. By identical logic, $\{c_2\}$ is strictly dominated by $\{a_2, c_2\}$ for agent 1, and the corresponding singleton strategies are strictly dominated for agent 2.

    Consider the distribution $\pi$ assigning $1/2$ probability to the outcome $(\{c_1\}, \{c_2\})$ and $1/2$ probability to $(\{c_2\}, \{c_1\})$. Under $\pi$, both agents exclusively secure an element of value $1$ with probability $1$. The expected utility for each agent is exactly $1$.

    We evaluate the CCE condition. The optimal unconditional deviation for agent 1 is to permanently play a strictly dominating strategy, e.g., $\{a_1, c_1\}$. Against the marginal distribution of $\pi$, agent 2 plays $\{c_2\}$ with probability $1/2$, allowing agent 1 to exclusively secure $\{a_1, c_1\}$ for a payoff of $1 + \epsilon$. With probability $1/2$, agent 2 plays $\{c_1\}$, causing agent 1 to exclusively secure only $a_1$ for a payoff of $\epsilon$. The expected utility of this optimal deviation is $\frac{1}{2}(1+\epsilon) + \frac{1}{2}(\epsilon) = 0.5 + \epsilon$.

    Because we constrained $\epsilon < 1/2$, we have $0.5 + \epsilon < 1$. Thus, any unilateral deviation strictly decreases expected utility, confirming the distribution is a strong CCE.

    \vspace{0.2cm}
    \noindent \textit{Step 3: Extremality via Polytope Geometry} \\
    Finally, we establish that these properties hold for extreme CCEs. The set of valid CCEs constitutes a bounded convex polytope defined by linear inequalities. We formulate a linear program over this polytope with the objective of maximizing the sum of probabilities assigned to the outcomes where all agents play their strictly dominated strategies. We have demonstrated in Steps 1 and 2 that a feasible solution exists with an objective value of exactly $1$. Because the maximum of a linear objective over a bounded convex polytope is  guaranteed to be achieved at an extreme point (vertex), there must exist an extreme CCE that places probability $1$ exclusively on strictly dominated profiles.
\end{proof}

\begin{lemma}[Strongly Uncoupled Convergence to Target CCEs]
\label{lem:app_uncoupled_cce}
    Given any coarse correlated equilibrium (CCE) defined by rational probabilities in a normal-form game, there exists a set of $N$ strongly uncoupled no-regret learning algorithms such that, in self-play, they converge to the target CCE at an optimal time-average rate of $O(1/T)$.
\end{lemma}
\begin{proof}
    We prove the lemma by explicitly constructing a compliant algorithm.

    \vspace{0.2cm}
    \noindent \textit{Step 1: Cyclic Representation of the Target Equilibrium} \\
    Let $C$ be a target CCE. Because the probabilities defining $C$ are rational, there exists an integer $K$ such that $C$ can be represented as a uniform distribution over a deterministic, cyclic sequence of pure action profiles $V = (v^0, v^1, \dots, v^{K-1})$. 

    \vspace{0.2cm}
    \noindent \textit{Step 2: Construction of the Grim-Trigger Dynamic} \\
    We construct a learning dynamic for each agent $i$ operating in two strict phases:
    \begin{itemize}
    \item \textbf{Cooperative Phase:} At time $t$, agent $i$ plays their assigned coordinate from the profile $v^{t \pmod K}$. The agent then observes their realized cost vector $c^t$. If this observed cost vector exactly matches a privately stored, pre-computed target cost vector $C_i^{t \pmod K}$, agent $i$ remains in the cooperative phase.
    \item \textbf{Penalty Phase:} If the observed cost vector deviates from the expected target cost vector $C_i^{t \pmod K}$, agent $i$ immediately and permanently switches to a standard, optimal no-regret algorithm (e.g., Optimistic Mirror Descent), initializing it as if time $t$ were round 1.
    \end{itemize}

    \vspace{0.2cm}
    \noindent \textit{Step 3: Verification of the Optimal Regret Bound} \\
    In self-play, all agents adhere to the sequence $V$, and the penalty phase is never triggered. The empirical distribution of play cycles through $V$, converging to the target CCE exactly. The total external regret evaluated at any time $T$ is strictly bounded by $O(K)$ (the length of the cycle), yielding an optimal time-average regret rate of $O(1/T)$. 

    Under arbitrary adversarial play, the total regret is partitioned into the rounds prior to the first observed deviation, and the rounds following it. Because $C$ is a valid CCE, the regret accumulated during the cooperative phase is strictly bounded by $O(K)$. Once a deviation is detected, the agent assumes worst-case adversarial play and transitions to a standard no-regret algorithm (e.g., Optimistic Mirror Descent). While such algorithms can achieve $O(\text{polylog}(T))$ regret under mutual predictable play, against an arbitrary adversary they guarantee the optimal worst-case bound of $O(\sqrt{T})$. The total regret is therefore bounded by $O(K) + O(\sqrt{T}) = O(\sqrt{T})$, mathematically establishing the algorithm as Hannan consistent under all conditions.

    \vspace{0.2cm}
    \noindent \textit{Step 4: Satisfaction of Strongly Uncoupled Constraints} \\
    Crucially, this dynamic satisfies the strict informational constraints of \textit{strongly uncoupled learning}~\cite{daskalakis2011near}. Agent $i$ possesses zero knowledge of the opponents' utility functions or the full game matrix. The private storage requirement is strictly bounded by $O(1)$, as the agent only needs to store their private $K$-step sequence of scheduled actions and expected target cost vectors. Verification requires only a direct numerical comparison of the observed vector against the stored target, demanding absolutely no structural knowledge of the opponents' matrices or the global sequence $V$. The computations performed in each round are highly efficient, running in polynomial time relative to the action space.
\end{proof}

\noindent
We are now ready to complete the proof of theorem~\ref{thm:not_rationalizable}.

%\subsection*{Proof of Theorem~\ref{thm:not_rationalizable}}
\begin{reptheorem}{thm:not_rationalizable}
    There exist smooth games (e.g., linear congestion games and valid utility games) where adversarially selected, strongly uncoupled no-regret algorithms converge at an optimal rate of $O(1/T)$ to a strong coarse correlated equilibrium, whilst simultaneously, on every period, all agents play strictly dominated strategies. 
\end{reptheorem}
\begin{proof}
    The theorem follows immediately from the synthesis of Lemma~\ref{lem:app_unnatural_cce} and Lemma~\ref{lem:app_uncoupled_cce}. Lemma~\ref{lem:app_unnatural_cce} guarantees the existence of smooth games possessing extreme strong CCEs supported entirely on strictly dominated strategies, which by construction are representable by cyclic sequences of length $K=2$. Lemma~\ref{lem:app_uncoupled_cce} demonstrates that strongly uncoupled, optimal $O(1/T)$ no-regret algorithms can be explicitly constructed to converge to these exact cyclic sequences. Consequently, optimal no-regret learning in smooth games provides no structural guarantee against convergence to strictly dominated, non-rationalizable outcomes.
\end{proof}

\section{Proof of Theorem~\ref{thm:linear_proximal}}
\label{app:SCE}

\begin{reptheorem}{thm:linear_proximal}[SCE/PCE Support Strictly Dominated Strategies]
    For Euclidean distance-generating functions, the set of (Bregman) Proximal Correlated Equilibria (which coincides exactly with Semicoarse Correlated Equilibria) can assign strictly positive probability mass to strictly dominated strategies.
\end{reptheorem}

\begin{proof}
    We establish the theorem constructively by defining a symmetric two-player normal-form game and explicitly verifying that a specific distribution, which places strictly positive mass on a strictly dominated strategy, satisfies all mathematical conditions of a Semicoarse Correlated Equilibrium (SCCE).

    \vspace{0.2cm}
    \noindent \textit{Step 1: Game Definition and Dominated Strategy Support} \\
    Let the action space for both agents be $A_i = \{A, A^-, C\}$. The utility function for Player 1, $u_1(a_1, a_2)$, is defined by the following payoff matrix:
    \begin{align*}
        u_1(A, A) &= 1, & u_1(A, A^-) &= 1, & u_1(A, C) &= 0 \\
        u_1(A^-, A) &= 0.8, & u_1(A^-, A^-) &= 0.8, & u_1(A^-, C) &= -0.2 \\
        u_1(C, A) &= 0, & u_1(C, A^-) &= 0, & u_1(C, C) &= 1
    \end{align*}
    By inspection, action $A$ strictly dominates action $A^-$ by a margin of $\epsilon = 0.2$ against every possible opponent strategy. 

    We propose the following correlated strategy profile $\sigma$, supported entirely on the diagonal:
    $$ \sigma(A,A) = 0.4, \quad \sigma(A^-,A^-) = 0.2, \quad \sigma(C,C) = 0.4 $$
    Crucially, $\sigma$ assigns a probability mass of $0.2$ directly to the strictly dominated profile $(A^-, A^-)$.

    \vspace{0.2cm}
    \noindent \textit{Step 2: Primal Verification via Extreme Rays} \\
    To verify the SCCE conditions, we define the marginal expected gain of unilaterally deviating from action $a$ to $a'$, weighted by the marginal probability of $a$:
    $$ g(a', a) = \sum_{a_2 \in A_2} \sigma(a, a_2) \left[ u_1(a', a_2) - u_1(a, a_2) \right] $$
    Evaluating this marginal gain across all action pairs yields:
    \begin{align*}
        g(A, A^-) &= 0.04 & g(A^-, A) &= -0.08 & g(C, A) &= -0.40 \\
        g(A, C) &= -0.40 & g(C, A^-) &= -0.16 & g(A^-, C) &= -0.48
    \end{align*}

    A distribution constitutes an SCCE if and only if the net expected gain is non-positive for all valid linear strategy modifications. As established in the literature (e.g., Theorem 3.3 in~\cite{ahunbay2025semicoarse}), these modifications correspond to convex combinations of specific extreme rays: cycle deviations and subset deviations. We systematically verify that every extreme ray yields a strictly non-positive net gain.

    \textbf{Cycle Deviations:}
    \begin{itemize}
        \item 2-cycle $\{A, A^-\}$: $g(A^-, A) + g(A, A^-) = -0.08 + 0.04 = -0.04 \le 0$
        \item 2-cycle $\{A, C\}$: $g(C, A) + g(A, C) = -0.40 - 0.40 = -0.80 \le 0$
        \item 2-cycle $\{A^-, C\}$: $g(C, A^-) + g(A^-, C) = -0.16 - 0.48 = -0.64 \le 0$
        \item 3-cycle $\{A, A^-, C\}$ (mapping each to the uniform distribution over the other two): \\
        $\frac{1}{2} \left[ g(A^-,A) + g(C,A) + g(A,A^-) + g(C,A^-) + g(A,C) + g(A^-,C) \right] = -0.74 \le 0$
    \end{itemize}

    \textbf{Subset Deviations (Sizes 1 and 2):}
    \begin{itemize}
        \item $\{A^-\} \to \text{Unif}(\{A, C\})$: $0.5 g(A, A^-) + 0.5 g(C, A^-) = 0.02 - 0.08 = -0.06 \le 0$
        \item $\{A\} \to \text{Unif}(\{A^-, C\})$: $0.5 g(A^-, A) + 0.5 g(C, A) = -0.04 - 0.20 = -0.24 \le 0$
        \item $\{C\} \to \text{Unif}(\{A, A^-\})$: $0.5 g(A, C) + 0.5 g(A^-, C) = -0.20 - 0.24 = -0.44 \le 0$
        \item $\{A^-, C\} \to A$ \textit{(Unconditional deviation to $A$)}: $g(A, A^-) + g(A, C) = 0.04 - 0.40 = -0.36 \le 0$
        \item $\{A, C\} \to A^-$ \textit{(Unconditional deviation to $A^-$)}: $g(A^-, A) + g(A^-, C) = -0.08 - 0.48 = -0.56 \le 0$
        \item $\{A, A^-\} \to C$ \textit{(Unconditional deviation to $C$)}: $g(C, A) + g(C, A^-) = -0.40 - 0.16 = -0.56 \le 0$
    \end{itemize}
    Because every extreme ray of the modification cone yields a strictly negative or zero expected gain, any linear combination of these transformations is similarly bounded by zero. 

    \vspace{0.2cm}
    \noindent \textit{Step 3: Dual Verification via Exact Algebraic Certificate} \\
    We further provide the exact algebraic certificate solving the dual linear programming formulation of the SCCE constraints (e.g., Theorem 3.4 in~\cite{ahunbay2025semicoarse}). The distribution $\sigma$ is an SCCE if there exist non-negative variables $\gamma_1(a', a) \ge 0$ and unconstrained variables $\rho_1(a, a')$ satisfying:
    \begin{align}
        \gamma_1(a', a) + \rho_1(a, a') - \rho_1(a', a) + g(a', a) &= 0 \quad \forall a \neq a' \label{eq:scce_dual_1} \\
        \sum_{a \neq a'} \gamma_1(a', a) + \text{UnconditionalGain}(a') &= 0 \quad \forall a' \label{eq:scce_dual_2}
    \end{align}
    The base expected utility of $\sigma$ is $U(\sigma) = 0.4(1) + 0.2(0.8) + 0.4(1) = 0.96$. The unconditional deviation gains to each fixed action are:
    \begin{align*}
        \text{Gain}(A) &= (0.4(1) + 0.2(1) + 0.4(0)) - 0.96 = -0.36 \\
        \text{Gain}(A^-) &= (0.4(0.8) + 0.2(0.8) + 0.4(-0.2)) - 0.96 = -0.56 \\
        \text{Gain}(C) &= (0.4(0) + 0.2(0) + 0.4(1)) - 0.96 = -0.56
    \end{align*}
    Substituting these into constraint (\ref{eq:scce_dual_2}) requires the non-negative $\gamma_1$ variables to satisfy:
    \begin{align*}
        \gamma_1(A, A^-) + \gamma_1(A, C) &= 0.36 \\
        \gamma_1(A^-, A) + \gamma_1(A^-, C) &= 0.56 \\
        \gamma_1(C, A) + \gamma_1(C, A^-) &= 0.56
    \end{align*}
    We explicitly construct the feasible $\gamma_1$ assignment:
    \begin{align*}
        \gamma_1(A, A^-) &= 0.04, & \gamma_1(A^-, A) &= 0 \\
        \gamma_1(A, C) &= 0.32, & \gamma_1(C, A) &= 0.48 \\
        \gamma_1(A^-, C) &= 0.56, & \gamma_1(C, A^-) &= 0.08
    \end{align*}
    This strictly non-negative assignment perfectly satisfies the sum constraints. By enforcing anti-symmetry ($\rho_1(x,y) = -\rho_1(y,x)$), constraint (\ref{eq:scce_dual_1}) uniquely determines the exact $\rho_1$ variables:
    \begin{align*}
        \rho_1(A, A^-) &= 0.04 \\
        \rho_1(A, C) &= -0.04 \\
        \rho_1(C, A^-) &= -0.04
    \end{align*}
    All LP constraints are met with perfect algebraic equality. This confirms that the proposed distribution $\sigma$ constitutes a valid Semicoarse Correlated Equilibrium while assigning strictly positive probability mass to a strictly dominated strategy.
\end{proof}

\section{Proofs of Claims in Section~\ref{sec:CE} and Further Numerical Investigation of Dynamics}
\label{app:CE_chaos}

\subsection{Proof of Theorem~\ref{thm:chaos}}

To visualize the algebraic transformation required for the proof of Theorem~\ref{thm:chaos}, we provide a structural diagram of the bipartite mirroring technique in Figure~\ref{fig:bipartite_mirroring}. This transformation is strictly necessary to convert non-concave evolutionary dynamics into multilinear normal-form optimization while preserving the target invariant manifold.

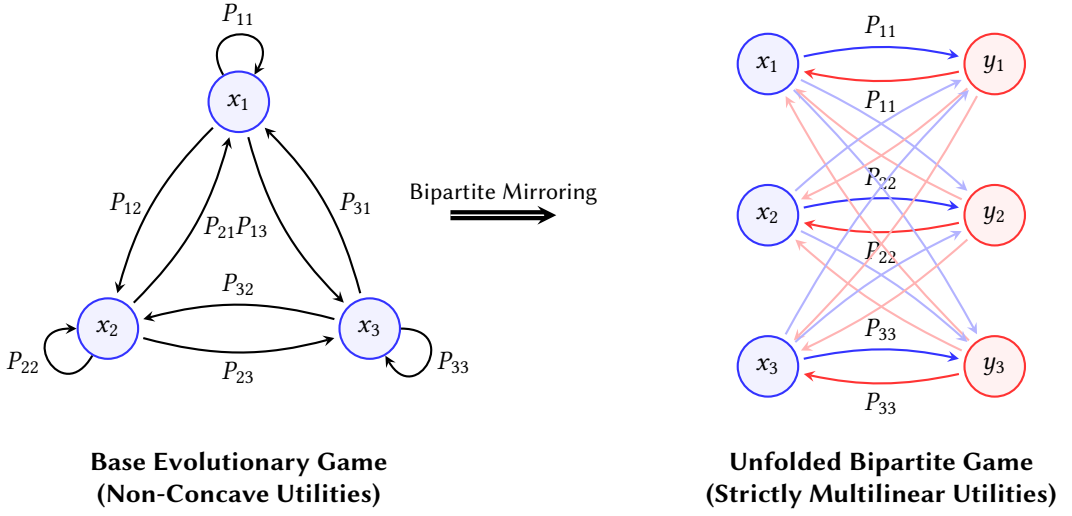
\begin{figure}[ht!]
    \centering
    \begin{tikzpicture}[
        >=stealth,
        agentX/.style={circle, draw=blue!80, fill=blue!5, thick, minimum size=0.8cm, font=\sffamily\bfseries},
        agentY/.style={circle, draw=red!80, fill=red!5, thick, minimum size=0.8cm, font=\sffamily\bfseries},
        edge/.style={->, thick, shorten >= 2pt, shorten <= 2pt},
        selfloop/.style={->, thick, min distance=8mm, looseness=5}
    ]

    % LEFT SIDE: Evolutionary Game (Self-Loops)
    % Placed in a perfect equilateral triangle
    \node[agentX] (x1) at (0, 2) {$x_1$};
    \node[agentX] (x2) at (-1.732, -1) {$x_2$};
    \node[agentX] (x3) at (1.732, -1) {$x_3$};

    % Edges separated by bending right (creates a nice symmetrical convex shape)
    \draw[edge] (x1) to[bend right=15] node[left, xshift=-0.1pt] {$P_{12}$} (x2);
    \draw[edge] (x2) to[bend right=15] node[right, xshift=0.1pt] {$P_{21}$} (x1);
    
    \draw[edge] (x2) to[bend right=15] node[below] {$P_{23}$} (x3);
    \draw[edge] (x3) to[bend right=15] node[above] {$P_{32}$} (x2);
    
    \draw[edge] (x3) to[bend right=15] node[right, xshift=2pt] {$P_{31}$} (x1);
    \draw[edge] (x1) to[bend right=15] node[left, xshift=-2pt] {$P_{13}$} (x3);

    % Self loops
    \draw[selfloop, in=60, out=120] (x1) to node[above] {$P_{11}$} (x1);
    \draw[selfloop, in=180, out=240] (x2) to node[left] {$P_{22}$} (x2);
    \draw[selfloop, in=300, out=360] (x3) to node[right] {$P_{33}$} (x3);

    \node[font=\sffamily\bfseries, align=center] at (0, -3) {Base Evolutionary Game\\(Non-Concave Utilities)};

    % MIDDLE: Arrow
    \draw[->, ultra thick, double] (2.8, 0.5) -- (4.2, 0.5) node[midway, above, font=\small\sffamily] {Bipartite Mirroring};

    % RIGHT SIDE: Bipartite Game (No Self-Loops, Bidirectional)
    \begin{scope}[xshift=7cm, yshift=0.5cm]
        % Population X
        \node[agentX] (X1) at (0, 2) {$x_1$};
        \node[agentX] (X2) at (0, 0) {$x_2$};
        \node[agentX] (X3) at (0, -2) {$x_3$};
        
        % Population Y
        \node[agentY] (Y1) at (3, 2) {$y_1$};
        \node[agentY] (Y2) at (3, 0) {$y_2$};
        \node[agentY] (Y3) at (3, -2) {$y_3$};

        % Horizontal Edges (X to Y and Y to X)
        \draw[edge, blue!80, bend left=12] (X1) to node[above, text=black] {$P_{11}$} (Y1);
        \draw[edge, red!80, bend left=12] (Y1) to node[below, text=black] {$P_{11}$} (X1);

        \draw[edge, blue!80, bend left=12] (X2) to node[above, text=black] {$P_{22}$} (Y2);
        \draw[edge, red!80, bend left=12] (Y2) to node[below, text=black] {$P_{22}$} (X2);

        \draw[edge, blue!80, bend left=12] (X3) to node[above, text=black] {$P_{33}$} (Y3);
        \draw[edge, red!80, bend left=12] (Y3) to node[below, text=black] {$P_{33}$} (X3);
        
        % Diagonal Edges (X to Y and Y to X) - slightly transparent to avoid clutter
        % X1 - Y2
        \draw[edge, blue!30, bend left=8] (X1) to (Y2);
        \draw[edge, red!30, bend left=8] (Y2) to (X1);
        % X1 - Y3
        \draw[edge, blue!30, bend left=8] (X1) to (Y3);
        \draw[edge, red!30, bend left=8] (Y3) to (X1);
        
        % X2 - Y1
        \draw[edge, blue!30, bend left=8] (X2) to (Y1);
        \draw[edge, red!30, bend left=8] (Y1) to (X2);
        % X2 - Y3
        \draw[edge, blue!30, bend left=8] (X2) to (Y3);
        \draw[edge, red!30, bend left=8] (Y3) to (X2);

        % X3 - Y1
        \draw[edge, blue!30, bend left=8] (X3) to (Y1);
        \draw[edge, red!30, bend left=8] (Y1) to (X3);
        % X3 - Y2
        \draw[edge, blue!30, bend left=8] (X3) to (Y2);
        \draw[edge, red!30, bend left=8] (Y2) to (X3);

        \node[font=\sffamily\bfseries, align=center] at (1.5, -3.5) {Unfolded Bipartite Game\\(Strictly Multilinear Utilities)};
    \end{scope}
    
    \end{tikzpicture}
    \caption{\textbf{Restoring Multilinearity via Bipartite Mirroring.} The base 3-agent polymatrix game (left) contains self-loops ($P_{ii}$), forcing an agent's utility to depend quadratically on their own mixed strategy, rendering the game non-concave and violating standard regret frameworks. By cloning the population into sets $X$ (blue) and $Y$ (red), self-loops are unfolded into bidirectional bipartite interactions between identical agents. The resulting 6-agent game is strictly multilinear, enabling optimal regret minimization while preserving the original chaotic vector field along the symmetric invariant manifold $x_i = y_i$.}
    \label{fig:bipartite_mirroring}
\end{figure}

\begin{reptheorem}{thm:chaos}
    There exists a one-parameter family of normal-form games with payoff symmetries such that continuous-time replicator dynamics, given symmetric, interior initial conditions, exhibits chaos while simultaneously:
    \begin{enumerate}
        \item[(i)] Minimizing swap regret at an optimal rate of $O(1/T)$ (fast convergence to CE).
        \item[(ii)] Ensuring the time-average strategies of all agents converge precisely to an interior Nash equilibrium.
    \end{enumerate}
    The games in question are polymatrix games with $6$ agents, each possessing exactly two strategies.
\end{reptheorem}

\begin{proof}
    We prove the theorem by formally constructing the requisite multilinear state space, verifying the bounded regret properties without relying on asymptotic approximations, and characterizing the exact topological nature of the strange attractor.

    \vspace{0.2cm}
    \noindent \textit{Step 1: Self-Contained $O(1/T)$ Swap Regret} \\
    We first establish that continuous-time replicator dynamics inherently minimize swap regret at an optimal $O(1/T)$ rate for any two-strategy game. Let $p(t) \in \Delta_2$ denote the mixed strategy of an agent over two actions, and let $u(t) \in \mathbb{R}^2$ be the instantaneous utility vector. The continuous-time replicator equation is $\dot{p}_j(t) = p_j(t) (u_j(t) - p(t)^\top u(t))$.

    To bound the external regret against any fixed action in hindsight, we evaluate the Kullback-Leibler (KL) divergence between a fixed target strategy $q$ and the state $p(t)$. Taking the time derivative yields $\frac{d}{dt} D_{KL}(q \parallel p(t)) = -\sum_j q_j \frac{\dot{p}_j(t)}{p_j(t)}$. Substituting the replicator equation, this simplifies exactly to $p(t)^\top u(t) - q^\top u(t)$, which is the instantaneous negative regret. Integrating over time $T$ yields the total external regret bound:
    $$ \text{ExtReg}(T, q) = \int_0^T (q^\top u(t) - p(t)^\top u(t)) dt = D_{KL}(q \parallel p(0)) - D_{KL}(q \parallel p(T)) \le D_{KL}(q \parallel p(0)) $$
    Because the initial condition $p(0)$ is strictly interior, $D_{KL}(q \parallel p(0))$ is a strictly bounded $O(1)$ constant. Thus, the total external regret for both pure strategies $1$ and $2$ is strictly bounded by $O(1)$. 
    
    Because the game possesses exactly two strategies, bounding external regret completely bounds swap regret. There are exactly four swap functions mapping $\{1, 2\}$ to itself: the identity map (zero regret), the map to constant $1$ (external regret for 1), the map to constant $2$ (external regret for 2), and the flip map ($1 \leftrightarrow 2$). By the linearity of expectation, the regret of the flip map is strictly upper-bounded by the sum of the external regrets for $1$ and $2$. Consequently, the total swap regret for any modification rule is bounded by $O(1)$, yielding an optimal time-average swap regret rate of $O(1/T)$ and guaranteeing fast convergence to the Correlated Equilibrium polytope.

    \vspace{0.2cm}
    \noindent \textit{Step 2: The Non-Concave Strange Attractor} \\
    We leverage a continuous-time 3-player, 2-strategy evolutionary polymatrix game introduced by Peixe and Rodrigues~\cite{peixe2022persistent}. This work rigorously establishes a one-parameter family of vector fields that undergo a supercritical Hopf bifurcation, leading to a Shilnikov homoclinic cycle to a saddle-focus. By the Shilnikov-Ovsyannikov theorem, the unfolding of this cycle generates suspended topological horseshoes, strictly positive Lyapunov exponents, and persistent Hénon-type strange attractors that survive over parameter sets of positive Lebesgue measure.

    However, this base system is an evolutionary game that includes ``self-loops,'' meaning an agent's sub-population interacts with itself. Algebraically, this forces the utility function of agent $i$ to depend on the quadratic form $x_i^\top P_{ii} x_i$. Because the utility functions are non-concave, the standard potential functions used in Step 1 fail, and standard regret guarantees do not apply.

    \vspace{0.2cm}
    \noindent \textit{Step 3: Bipartite Mirroring and the Invariant Manifold} \\
    To resolve the non-concavity while strictly preserving the Shilnikov cycle, we apply a bipartite mirroring transformation (visualized in Figure~\ref{fig:bipartite_mirroring}). We duplicate the 3 agents to create two distinct mirror populations: $X = (x_1, x_2, x_3)$ and $Y = (y_1, y_2, y_3)$. We configure the interactions such that $X$ agents only play against $Y$ agents, replacing the self-loop matrices $P_{ii}$ with bipartite matrices $P_{i, i'}$ mapping $x_i$ against $y_i$. 

    The resulting 6-agent normal-form game is strictly multilinear, rendering the $O(1)$ swap-regret bounds from Step 1 fully applicable to all 6 agents. Crucially, the replicator vector field preserves symmetric initial conditions. If initialized such that $x_i(0) = y_i(0)$ for all $i \in \{1, 2, 3\}$, the dynamics guarantee $\dot{x}_i(t) = \dot{y}_i(t)$ for all $t > 0$. The system is completely confined to a 3-dimensional symmetric invariant manifold $\mathcal{M} = \{(X, Y) \mid X=Y\}$ within the 6-dimensional state space. Restricted to $\mathcal{M}$, the multilinear vector field perfectly collapses back into the exact non-concave vector field of the Peixe-Rodrigues system, fully inheriting its Hénon-type strange attractor and positive Lyapunov exponents.

    \vspace{0.2cm}
    \noindent \textit{Step 4: Time-Average Equilibration} \\
    Finally, we prove that despite the macroscopic chaos on $\mathcal{M}$, the time-average of the trajectory converges exactly to an interior Nash equilibrium. 
    
    The replicator equation for agent $i$'s probability of playing strategy $1$ is $\dot{x}_i = x_i (1-x_i)(u_{i,1} - u_{i,2})$. Dividing by $x_i(1-x_i)$ and integrating over $[0,T]$ yields:
    $$ \frac{1}{T} \left( \ln\left(\frac{x_i(T)}{1-x_i(T)}\right) - \ln\left(\frac{x_i(0)}{1-x_i(0)}\right) \right) = \frac{1}{T}\int_0^T (u_{i,1}(t) - u_{i,2}(t)) dt $$
    Because the strange attractor is strictly contained within the interior of the simplex, $x_i(t)$ is strictly bounded away from the boundaries $0$ and $1$ for all $t$. Therefore, the logarithmic terms on the left-hand side are bounded by an absolute constant $O(1)$. As $T \to \infty$, the left-hand side rigorously evaluates to $0$.

    This implies the time-average difference in expected utilities evaluates to zero. Because the unfolded game is a multilinear polymatrix game, the expected utility differences are strict linear functions of the state variables $(X, Y)$. By the linearity of the integral, evaluating the utility differences at any accumulation point $(\hat{X}, \hat{Y})$ of the time-average trajectory yields exactly zero. Thus, all agents are rendered perfectly indifferent between their strategies. Because the attractor is strictly interior, any such accumulation point must be the unique interior Nash equilibrium of the game.
\end{proof}

\subsection{Numerical Verification of Structural Stability and Persistence of Chaos}
\label{app:robust_chaos_exp}

To further validate the structural stability of the chaotic dynamics observed in the one-parameter family of polymatrix replicators, we conducted a rigorous perturbation analysis on the generalized 12-coefficient system derived from Equation (6) in~\cite{peixe2022persistent}. Setting the base parameter to $\mu = 3.6$---which lies strictly within the chaotic regime---we calculated the maximal Lyapunov exponent (LE). The unperturbed base game yielded a positive LE of approximately $0.0489$, mathematically confirming the presence of a strange attractor. 

To test the persistence of this attractor under structural deformations (approximating a $C^2$ neighborhood of vector fields), we applied single ($k=1$) and pairwise ($k=2$) perturbations of $\pm \rho$ (with $\rho = 0.01$) to the polynomial coefficients of the vector field. Following the methodology of the main text, we adopted a conservative threshold of $5 \times 10^{-3}$ to filter out numerical noise, classifying any LE above this value as strictly positive and indicative of chaotic dynamics. 

The perturbation results provide robust computational support for the persistence of the chaotic regime identified in Theorem~\ref{thm:chaos}. Out of $24$ single-parameter perturbations, $23$ ($95.8\%$) retained a positive LE. For pairwise perturbations, $251$ out of $264$ ($95.1\%$) remained chaotic. The distribution of the maximum Lyapunov exponents across all $289$ simulated games (Figure~\ref{fig:structural_robust_chaos}) illustrates that the vast majority of perturbed systems cluster strictly above the chaotic threshold. These findings provide strong computational evidence that the Hénon-type strange attractors in this system are persistently observable and highly robust to small, generalized structural perturbations in the space of symmetric games.

\begin{figure}[ht!]
    \centering
    \includegraphics[width=0.95\textwidth]{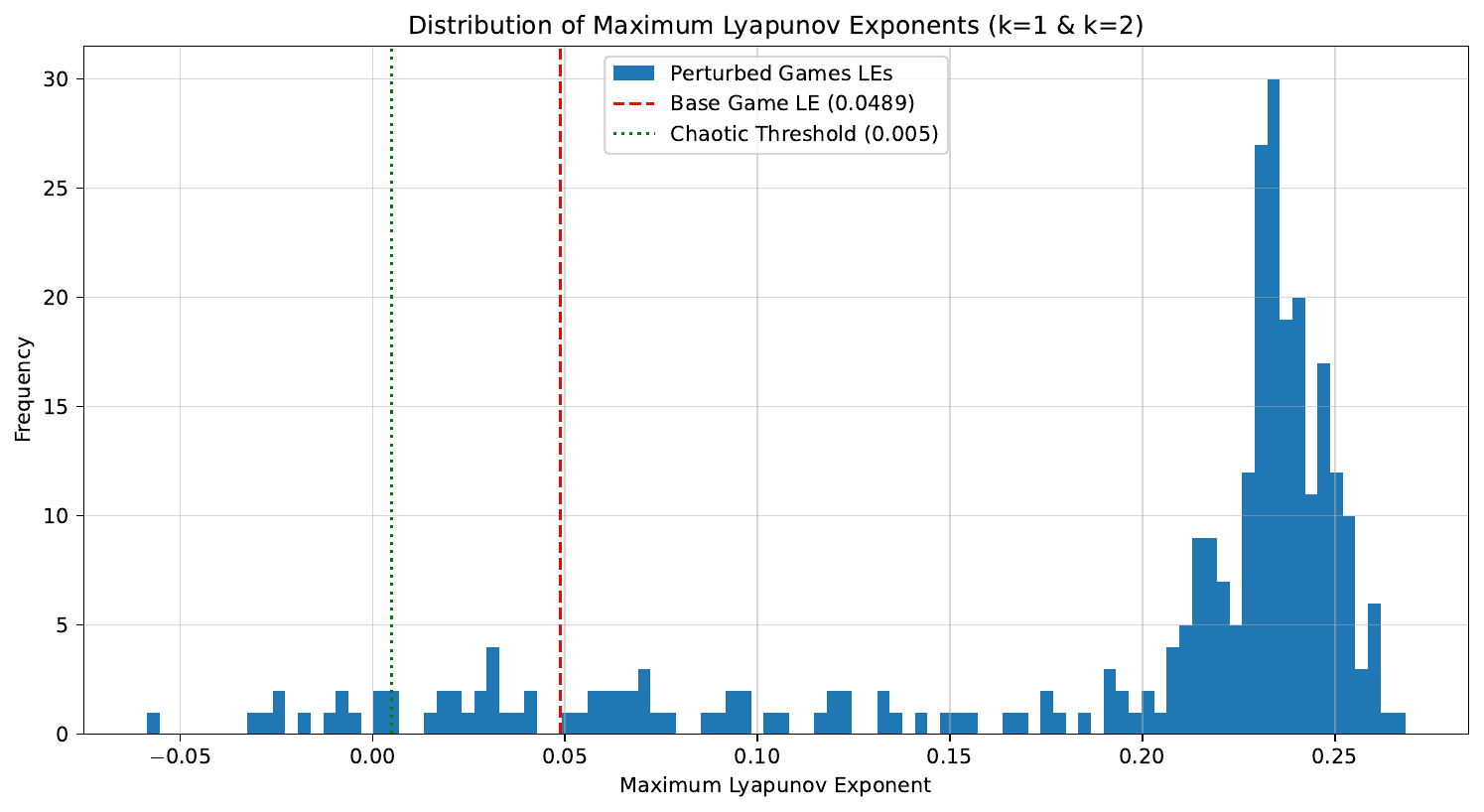}
    \caption{\textbf{Structural Robustness of Chaos.} The distribution of the maximum Lyapunov exponents across all $289$ simulated games illustrates that the vast majority of perturbed systems $(>95\%)$ cluster well above the positive chaotic threshold.}
    \label{fig:structural_robust_chaos}
\end{figure}

\subsection{Numerical Demonstration of Symmetry Breaking and Distinct Limit Cycles}
\label{app:symm_breaking_cycles}

In Theorem~\ref{thm:chaos}, we established that continuous-time replicator dynamics can exhibit chaotic limit sets in a 6-player, 2-strategy normal-form game, provided the system is initialized with perfect symmetry between the two mirror populations ($X$ and $Y$). To demonstrate the topological brittleness of this chaotic limit behavior under generic conditions, we numerically simulate the full 6D replicator vector field under asymmetric initializations.

\vspace{0.2cm}
\noindent \textbf{Experimental Setup:} We initialize the system near the unstable interior equilibrium, $O_\mu$. To break the invariant 3D manifold, we introduce a microscopic asymmetric perturbation. Specifically, the $X$ agents are initialized at $X_{init} = O_\mu + \xi$, and the mirror $Y$ agents are initialized at $Y_{init} = X_{init} + \delta$, where $\delta \sim 10^{-5}$ represents a random, asymmetric noise vector between the two populations. We simulate 10 such trajectories for $T = 300$ time units.

\vspace{0.2cm}
\noindent \textbf{Observations:} As shown in Figure~\ref{fig:symmetry_breaking}, the trajectories initially shadow the chaotic limit set, remaining tightly coupled. However, because the chaotic regime is not an open attractor in the full 6D space, the positive Lyapunov exponents rapidly amplify the microscopic noise $\delta$. The Euclidean distance between the $X$ and $Y$ populations grows exponentially, breaking the mirror symmetry. 

Crucially, the destruction of the chaotic limit set does not lead the system to a global, static point-mass equilibrium. As the dynamics are repelled toward the boundary of the simplex, four of the state variables asymptotically converge to pure strategies ($0$ or $1$), while the remaining two variables are captured by a periodic orbit. Depending on the exact vector of the initial microscopic perturbation $\delta$, the trajectories are expelled toward entirely distinct, lower-dimensional limit cycles residing on different 2D boundary faces of the state space. This reinforces the central thesis of Section~\ref{sec:CE}: even when fragile, problematic behaviors like measure-zero chaos are destroyed by generic noise, the resulting statistically typical behavior systematically evades the predictive power of static equilibria.

\begin{figure}[ht!]
    \centering
    \includegraphics[width=0.95\textwidth]{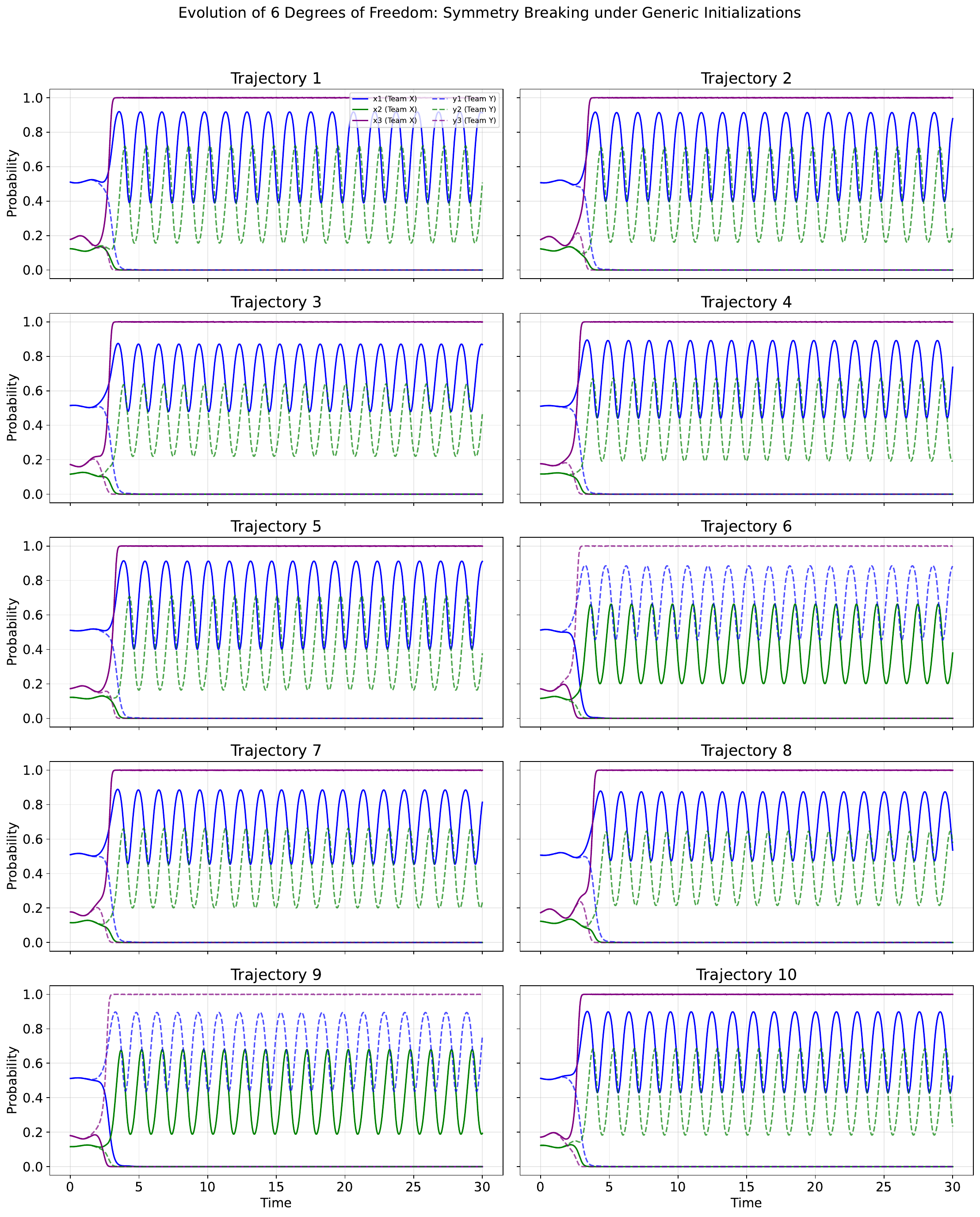}
    \caption{\textbf{Symmetry Breaking and Emergence of Distinct Limit Cycles.} Numerical simulation of the 6-agent mirrored polymatrix game under asymmetric initializations. Ten trajectories are initialized near the interior equilibrium with a microscopic asymmetric perturbation ($10^{-5}$) between the $X$ and $Y$ populations. While the trajectories initially shadow the chaotic limit set, the positive Lyapunov exponents rapidly amplify the noise, leading to an exponential divergence that permanently breaks the mirror symmetry. The system subsequently escapes to the boundary of the simplex, where the agents are captured by entirely distinct, lower-dimensional limit cycles residing on different boundary faces of the state space.}
    \label{fig:symmetry_breaking}
\end{figure}

\FloatBarrier

\section{Proofs of Claims in Section \ref{sec:nonatomic_dynamics}}

\subsection{Proof of Theorem \ref{thm:2p_degradation}}
\label{app:2p_proof}

We analyze the discrete-time map for a symmetric congestion game with latency $c_1(x) = c_2(x) = x^p$ under Multiplicative Weights Update. Let $x_t \in (0,1)$ be the flow on the first path. The update map is defined as:
\begin{equation}
x_{t+1} = f_a(x_t) = \frac{x_t}{x_t + (1-x_t)\exp\left(a \left[x_t^p - (1-x_t)^p\right]\right)}
\end{equation}
where $a > 0$ acts as the effective system load parameter. By symmetry, $x^* = 1/2$ is a fixed point.

\begin{reptheorem}{thm:2p_degradation}
Consider a symmetric non-atomic congestion game with two parallel paths and identical monomial latency functions $c(x) = x^p$ for $p \ge 1$. Under MWU with a fixed effective step-size parameter $a > 0$\footnote{The parameter $a$ formally couples the MWU learning rate $\epsilon$, the population size $N$, and the polynomial degree $p$. It is defined as $a = N^p \ln\left(\frac{1}{1-\epsilon}\right) \approx \epsilon N^p$. This relationship explicitly demonstrates that to maintain a stable effective step-size $a$, the raw learning rate $\epsilon$ must be aggressively throttled inversely to the maximum network cost $\mathcal{O}(N^p)$.}:
\begin{enumerate}
    \item The unique Wardrop equilibrium $x^* = 1/2$ is locally asymptotically stable if $a < \frac{2^{p+1}}{p}$.
    \item If $a > \frac{2^{p+1}}{p}$, the equilibrium strictly destabilizes and the interior system is captured by a global, period-2 attracting orbit $\{\sigma_a, 1-\sigma_a\}$. Accounting for the population size $N$, the empirical time-averaged social cost of this chaotic orbit is exactly $SC_{emp}(a) = N^{p+1}[\sigma_a^{p+1} + (1-\sigma_a)^{p+1}]$, which strictly exceeds the optimal social cost of $N^{p+1}(1/2)^p$.
    \item As the effective load $a \to \infty$ ($\sigma_a \to 0$), the ratio of this empirical time-averaged social cost to the optimal social cost converges exactly to $2^p$.
\end{enumerate}
\end{reptheorem}

\begin{proof}\leavevmode
\paragraph{Part 1: Local Asymptotic Stability.}
Let $h_a(x) = a(x^p - (1-x)^p)$. Applying the quotient rule to $f_a(x)$ and evaluating at the equilibrium $x = 1/2$, we note $h_a(1/2) = 0$ and $h_a'(1/2) = a p 2^{2-p}$. The derivative of the map at the fixed point simplifies exactly to:
\begin{equation}
f_a'(1/2) = 1 - \frac{1}{4}h_a'(1/2) = 1 - a p 2^{-p}
\end{equation}
For local asymptotic stability, we require $|f_a'(1/2)| < 1$. Because $a, p > 0$, the upper bound holds trivially. The lower bound requires $1 - a p 2^{-p} > -1$, which rigorously isolates the critical threshold $a < \frac{2^{p+1}}{p}$. If $a$ exceeds this threshold, $f_a'(1/2) < -1$ and the equilibrium strictly destabilizes.

\paragraph{Part 2: Global Convergence to the Period-2 Attractor.}
Before establishing the global attractor, we state a foundational result from one-dimensional topological dynamics regarding continuous maps lacking period-2 points.

\begin{lemma}[Convergence in Maps without Period-2 Points \cite{block2006dynamics,sharkovsky2013dynamics}]
\label{lem:no_period2}
Let $I \subset \mathbb{R}$ be a closed interval and $g: I \to I$ a continuous map. If $g$ possesses no periodic points of period 2, then the sequence of iterates $g^n(x)$ converges to a fixed point of $g$ for every $x \in I$.
\end{lemma}

We define the involution $\varphi(x) = 1-x$. Because the game is symmetric, $f_a$ commutes with $\varphi$. We introduce the auxiliary map $g_a = \varphi \circ f_a = 1 - f_a(x)$. Because $\varphi$ is an involution, evaluating $g_a$ twice yields $g_a^2(x) = 1 - f_a(1-f_a(x)) = f_a^2(x)$. Thus, the period-2 points of $f_a$ are mathematically identical to the fixed points of $g_a^2$.

The fixed points of $g_a$ satisfy $g_a(x) = x \iff f_a(x) = 1-x$. Substituting the exact definition of $f_a(x)$ and taking the natural logarithm, this condition is equivalent to finding the roots of:
\begin{equation}
D(x) = \ln\left(\frac{x}{1-x}\right) - \frac{a}{2}\left(x^p - (1-x)^p\right) = 0
\end{equation}

To rigorously determine the number of roots of $D(x)$ on the active interior $(0, 1/2)$, we must evaluate its critical points by taking the first derivative:
\begin{equation}
D'(x) = \frac{1}{x(1-x)} - \frac{ap}{2}\left(x^{p-1} + (1-x)^{p-1}\right)
\end{equation}
Setting $D'(x) = 0$ is algebraically equivalent to finding the intersections of an auxiliary polynomial $H_p(x)$ with a horizontal target line $y = \frac{2}{ap}$, where:
\begin{equation}
H_p(x) = x(1-x)\left(x^{p-1} + (1-x)^{p-1}\right)
\end{equation}

While $H_{p}(x)$ is trivially strictly monotonic on $(0,1/2)$ for $p\le3$ (as $H_{1}(x)=2x(1-x)$ and $H_{2}(x)=x(1-x)$ are downward-opening parabolas, and $H_{3}^{\prime}(x)=(1-2x)^{3}>0$),\footnote{In all three cases $H_p$ is strictly monotone increasing on $(0,1/2)$, so the equation $H_p(x) = \frac{2}{ap}$ has exactly one solution on $(0,1/2)$ whenever $\frac{2}{ap} < 2^{-p}$, i.e., whenever $a > \frac{2^{p+1}}{p}$. Thus $D'(x)=0$ exactly once on $(0,1/2)$, $D$ achieves a unique interior maximum, and since $D(0^+) = -\infty$, $D(1/2)=0$, and $D'(1/2) = 4 - ap\cdot 2^{1-p} < 0$ in the unstable regime, $D$ crosses zero from below exactly once, yielding a unique root $\sigma_a \in (0,1/2)$.} it loses monotonicity for higher degrees. To rigorously map its geometric profile for $p\ge4$, let $v=x$ and $w=1-x$.  
Because $x \in (0, 1/2)$, we strictly have $w > v > 0$. The critical points of $H_p(x)$ satisfy $\frac{d}{dx}(v^p w + v w^p) = 0$, which yields:
\begin{equation}
w^p - v^p = pvw(w^{p-2} - v^{p-2})
\end{equation}
Dividing both sides by $v^p$ and substituting the ratio $z = w/v$ (noting that $z > 1$ because $w > v$), this requirement transforms into finding the roots of the polynomial $K(z)$:
\begin{equation}
K(z) = z^p - pz^{p-1} + pz - 1 = 0
\end{equation}

We evaluate $K(z)$ strictly for the domain $z > 1$.
At the boundary, $K(1)=0$, and the initial slope evaluates to $K^{\prime}(1)=p(3-p)$. Because $p\ge4$, the initial slope is strictly negative, forcing the polynomial to initially dip below zero.

To prove $K(z)$ subsequently crosses zero exactly once, we analyze its concavity via the second derivative:
\begin{equation}
K''(z) = p(p-1)z^{p-3}(z - (p-2))
\end{equation}
Because $p(p-1)z^{p-3} > 0$, the sign of $K''(z)$ is dictated entirely by $(z - (p-2))$. Thus, $K''(z) < 0$ for $z \in (1, p-2)$ and $K''(z) > 0$ for $z > p-2$. This implies the slope $K'(z)$ strictly decreases to a unique global minimum at $z = p-2$ before strictly and monotonically increasing to $+\infty$. Consequently, $K'(z)$ transitions from negative to positive exactly once. 

Because $K'(z)$ has exactly one root, $K(z)$ starts at $0$, decreases to a unique global minimum, and then strictly increases to $+\infty$, guaranteeing it crosses zero exactly once for $z > 1$. Translated back to our original variables, this mathematically guarantees that $H_p(x)$ possesses exactly one local extremum---a global maximum---on the interval $(0, 1/2)$.

Thus, for $p\ge4$, $H_{p}(x)$ ascends from $H_{p}(0)=0$ to its unique peak, then strictly descends to $H_{p}(1/2)=2^{-p}$ (for $p\le3$, it simply strictly ascends to $2^{-p}$). Under the unstable regime, our threshold parameter is $a>\frac{2^{p+1}}{p}$, which dictates that the horizontal target line satisfies $\frac{2}{ap}<2^{-p}$. Because this target line is mathematically trapped below $2^{-p}$, it is forced to intersect $H_{p}(x)$ exactly once on its ascending branch for all $p\ge1$.

Therefore, $D'(x) = 0$ exactly once. As $x$ increases, $H_p(x)$ initially increases from 0, causing $D'(x)$ to transition from strictly positive to strictly negative. Because $D(x) \to -\infty$ as $x \to 0^+$ and $D(1/2) = 0$, the transition of $D'(x)$ dictates that $D(x)$ increases from $-\infty$ to a positive maximum before decreasing to $0$ at $x=1/2$. By the Intermediate Value Theorem, $D(x)$ must cross zero exactly once on $(0, 1/2)$. We denote this unique root as $\sigma_a$. 

By symmetry, $1-\sigma_a$ is the unique root in $(1/2, 1)$. Thus, $g_a$ has exactly three fixed points: $\sigma_a$, $1/2$, and $1-\sigma_a$.

Because $g_a^2 = f_a^2$, any period-2 point of $g_a$ must be a fixed point of $f_a^2$. We now rigorously determine the complete set of solutions to $f_a^2(x) = x$ on the interior $(0,1)$. Let $h_p(x) = x^p - (1-x)^p$. From the definition of $f_a(x)$, we have the algebraic identity $\frac{1-f_a(x)}{f_a(x)} = \frac{1-x}{x}\exp(a h_p(x))$. Imposing the period-2 condition yields:
\begin{equation}
f_a(f_a(x)) = \frac{f_a(x)}{f_a(x) + (1-f_a(x))\exp(a h_p(f_a(x)))} = x
\end{equation}

Dividing the numerator and denominator by $f_a(x)$ and substituting our identity gives:
\begin{equation}
\frac{1}{1 + \frac{1-x}{x}\exp\left(a [h_p(x) + h_p(f_a(x))]\right)} = x
\end{equation}

Multiplying by the denominator and simplifying yields $(1-x)\exp(a [h_p(x) + h_p(f_a(x))]) = 1-x$. Because $x \in (0,1)$, we divide by $1-x$ and take the natural logarithm. Since $a > 0$, the condition simplifies exactly to $h_p(f_a(x)) = -h_p(x)$. 

By definition, $-h_p(x) = (1-x)^p - x^p = h_p(1-x)$. Therefore, the period-2 condition is entirely equivalent to $h_p(f_a(x)) = h_p(1-x)$. We note that $h_p'(x) = p x^{p-1} + p(1-x)^{p-1} > 0$ for all $x \in (0,1)$ and $p \ge 1$. Because $h_p$ is strictly monotonic, it is injective, which strictly mandates that $f_a(x) = 1-x$. 

Thus, the fixed points of $f_a^2$ are mathematically identical to the solutions of $f_a(x) = 1-x$, which are exactly the fixed points of $g_a$. Therefore, $g_a$ has absolutely no points of period 2 strictly on the open interior $(0,1)$. 

To rigorously apply Lemma \ref{lem:no_period2}, which requires a closed interval, we must analyze the boundaries. Extending $f_a(x)$ continuously to $[0,1]$ yields $f_a(0) = 0$ and $f_a(1) = 1$, which corresponds to $g_a(0) = 1$ and $g_a(1) = 0$. Thus, $\{0, 1\}$ forms a period-2 orbit for $g_a$ on the boundary. However, evaluating the derivative of the original map at the boundaries yields $f_a'(0) = f_a'(1) = e^a$. Because $a > 0$, we strictly have $e^a > 1$, making $0$ and $1$ strictly repelling fixed points of $f_a$. 

Because the boundaries are strictly repelling ($f_a'(0) > 1$), the map locally expands away from the boundaries, satisfying $f_a(x) > x$ for sufficiently small $x > 0$. Therefore, there exists a small $\epsilon > 0$ such that the closed, compact subinterval $I =[\epsilon, 1-\epsilon]$ is forward-invariant under $f_a$ (and consequently under $g_a$), eventually absorbing all asymptotic trajectories originating in $(0,1)$. Because $I$ is a closed interval and $g_a|_I$ has no period-2 points, we may validly apply Lemma \ref{lem:no_period2}. The $g_a$-trajectory of every $x \in I$ must converge to one of the fixed points of $g_a$. Consequently, for any initial condition $x \in (0,1)$, the sequence $f_a^n(x)$ must asymptotically converge either to the equilibrium $1/2$ or to the interior 2-cycle $\{\sigma_a, 1-\sigma_a\}$.

It remains to formally prove that the basin of attraction of the fixed point $1/2$ has Lebesgue measure zero, ensuring that almost every initial condition is captured by the 2-cycle. 

Let $W^s(1/2) = \{x \in (0,1) \mid \lim_{n \to \infty} f_a^n(x) = 1/2\}$ denote the stable set of the equilibrium. In the unstable regime $a > \frac{2^{p+1}}{p}$, we established in Part 1 that the derivative at the fixed point satisfies $|f_a'(1/2)| > 1$. Because $f_a$ is continuously differentiable, there exists a local neighborhood $U = (1/2 - \epsilon, 1/2 + \epsilon)$ and a constant $\lambda > 1$ such that $|f_a'(x)| \ge \lambda$ for all $x \in U$.

Suppose there exists a trajectory that asymptotically converges to $1/2$ without ever landing exactly on it. By definition of convergence, there exists an integer $N$ such that for all $n \ge N$, $f_a^n(x) \in U$. Applying the Mean Value Theorem, for any $n \ge N$, there exists a $c \in U$ such that:
\begin{equation}
|f_a^{n+1}(x) - 1/2| = |f_a(f_a^n(x)) - f_a(1/2)| = |f_a'(c)| |f_a^n(x) - 1/2| \ge \lambda |f_a^n(x) - 1/2|
\end{equation}
By induction, for any $k > 0$, we strictly have:
\begin{equation}
|f_a^{N+k}(x) - 1/2| \ge \lambda^k |f_a^N(x) - 1/2|
\end{equation}
Because we assumed the trajectory never lands exactly on the fixed point, $|f_a^N(x) - 1/2| > 0$. Since $\lambda > 1$, the term $\lambda^k \to \infty$ as $k \to \infty$, which strictly forces the trajectory to leave the bounded neighborhood $U$, contradicting the assumption of asymptotic convergence.

Therefore, asymptotic convergence to $1/2$ is strictly impossible. The only way a trajectory can converge to $1/2$ is if it maps exactly onto $1/2$ in finite time. Mathematically, this implies the stable set is exactly the countable union of all pre-images of $1/2$:
\begin{equation}
W^s(1/2) = \bigcup_{k=0}^\infty f_a^{-k}(1/2)
\end{equation}

Because the latency functions $x^p$ are polynomials, the update map $f_a(x)$ is real-analytic on the open domain $(0,1)$. The pre-images of $1/2$ correspond to the roots of the analytic equations $f_a^k(x) - 1/2 = 0$. Because $f_a^k(x)$ is not identically $1/2$, the identity theorem for analytic functions guarantees that these roots are strictly isolated within $(0,1)$. 

While an infinite number of roots could theoretically accumulate at the boundaries $0$ or $1$, any set of strictly isolated points in $\mathbb{R}$ is at most countable. Therefore, for any finite iterate $k$, the set of pre-images $f_a^{-k}(1/2)$ is a countable set. 

Consequently, the entire stable set $W^s(1/2) = \bigcup_{k=0}^\infty f_a^{-k}(1/2)$ is a countable union of countable sets, making $W^s(1/2)$ strictly countable. Because any countable subset of $\mathbb{R}$ has a Lebesgue measure of exactly zero, the set of initial conditions asymptotically converging to $1/2$ is a zero-measure set. The trajectory of almost every $x \in (0,1)$ is thus captured by the global period-2 attractor.

\paragraph{Part 3: The $2^p$ Asymptotic Inefficiency.}
Accounting for the total population mass $N$, the true latency on a path with fractional flow $x$ is $c(Nx) = (Nx)^p = N^p x^p$. The optimal social cost at the Wardrop equilibrium ($x^*=1/2$) is therefore $SC_{opt} = N(1/2) \cdot N^p(1/2)^p + N(1/2) \cdot N^p(1/2)^p = N^{p+1}(1/2)^p$.

Because the 2-cycle satisfies $g_a(\sigma_a) = \sigma_a$, we can solve exactly for the system parameter $a$ as a function of the amplitude $\sigma_a$:
\begin{equation}
a = \frac{2 \ln\left(\frac{1-\sigma_a}{\sigma_a}\right)}{(1-\sigma_a)^p - \sigma_a^p}
\end{equation}

Because $a(\sigma_a)$ is continuous and finite on the open interval $(0, 1/2)$, any sequence of states yielding $a \to \infty$ must topologically accumulate at the boundaries of this domain. As $\sigma_a \to 1/2$, the denominator approaches zero, but the numerator concurrently approaches zero; applying L'Hôpital's rule yields a finite limit of $\frac{2^{p+1}}{p}$\footnote{To rigorously evaluate the $0/0$ indeterminate form, we apply L'Hôpital's rule: $\lim_{\sigma_a \to 1/2} a(\sigma_a) = \lim_{\sigma_a \to 1/2} \frac{\frac{d}{d\sigma_a}[2\ln(1-\sigma_a) - 2\ln(\sigma_a)]}{\frac{d}{d\sigma_a}[(1-\sigma_a)^p - \sigma_a^p]} = \frac{-8}{-p(1/2)^{p-1} - p(1/2)^{p-1}} = \frac{-8}{-p 2^{2-p}} = \frac{2^{p+1}}{p}$. This mathematically recovers the exact local stability threshold from Part 1, dynamically verifying that the period-2 orbit bifurcates continuously from the fixed point $1/2$ at this parameter value.}. Because divergence is strictly impossible at any interior point and at the boundary $\sigma_a \to 1/2$, the condition $a \to \infty$ leaves only one possibility, forcing the limit $\lim_{a \to \infty} \sigma_a = 0$.

The system flow strictly alternates between the states $\sigma_a$ and $1-\sigma_a$. Due to the symmetric identical latency functions, the absolute social cost evaluates to exactly the same value at both states of the cycle, yielding a constant empirical time-averaged social cost of $SC_{emp} = N^{p+1}[\sigma_a^{p+1} + (1-\sigma_a)^{p+1}]$. Evaluating the limit as $a \to \infty$ ($\sigma_a \to 0$) yields $\lim_{a \to \infty} SC_{emp} = N^{p+1}$. The ultimate ratio of empirical performance to optimal performance perfectly cancels the population scaling, yielding  $SC_{emp} / SC_{opt} = N^{p+1} /[N^{p+1}(1/2)^p] = 2^p$.
\end{proof}

\subsection{Construction of $L_2$ Chaos (Reference for Figure \ref{fig:l2_chaos})}
\label{app:l2_chaos}

We construct a non-atomic game to prove that Li-Yorke chaos emerges under $L_2$ regularization (PGD/FTRL) without activating boundary projections.
Let Link 2 have constant latency $c_2(1-x) = 2/3$. Let Link 1 have piecewise linear latency $c_1(x) = \frac{2}{3}x$ for $x \le 0.7$, and $c_1(x) = 5x - 3.033\dots$ for $x > 0.7$. 
The potential gradient is $g(x) = c_1(x) - c_2(1-x)$. 

Under PGD with learning rate $\eta = 1$, the unprojected map $f(x) = x - g(x)$ traces a forward orbit from $x_0 = 0.1$:

\begin{enumerate}
    \item $x_0 = 0.1 \implies g(0.1) = -0.6 \implies x_1 = 0.7$
    \item $x_1 = 0.7 \implies g(0.7) = -0.2 \implies x_2 = 0.9$
    \item $x_2 = 0.9 \implies g(0.9) = 0.8 \implies x_3 = 0.1$
\end{enumerate}

The sequence $0.1 \to 0.7 \to 0.9 \to 0.1$ is a period-3 cycle. Because $\max(x_t) = 0.9 < 1$ and $\min(x_t) = 0.1 > 0$, the Euclidean projection operator $\Pi_{[0,1]}$ remains the identity map. Because the projection is inactive, the unconstrained dual variable in FTRL is exactly equal to the primal variable in PGD. By the Li-Yorke Theorem, the existence of a period-3 interior orbit guarantees the presence of Li-Yorke chaos for both $L_2$ learning dynamics.

%\newpage
\end{document}